\documentclass{article}
\usepackage[utf8]{inputenc}

\usepackage[T1]{fontenc}
\usepackage{parskip}
\usepackage{ragged2e}
\usepackage{enumitem}
\usepackage{float}
\usepackage{lscape}
\usepackage{rotating}

\usepackage{fancyhdr}
\usepackage{etoolbox}
\usepackage{geometry}

\usepackage{amsmath}
\usepackage{amsfonts}
\usepackage{amssymb}
\usepackage{mathtools}
\usepackage{amsthm}
\usepackage{bbm}
\usepackage{nccmath}
\usepackage{scalerel}
\usepackage{actuarialsymbol}
\usepackage{mleftright}

\usepackage{multirow}
\usepackage{multicol}
\usepackage{tabularx}

\usepackage{graphicx}
\usepackage[table]{xcolor}
\usepackage{array}
\usepackage{subcaption}

\usepackage{xurl}
\usepackage[sectionbib,round]{natbib}
\usepackage[colorlinks, linkcolor = black,  linkcolor = blue,
urlcolor  = blue,
citecolor = blue,
anchorcolor = blue]{hyperref}
\usepackage{etoolbox}

\makeatletter
\newcounter{author}
\renewcommand*\author[1]{%
  \stepcounter{author}%
  \ifnum\c@author=1
    \gdef\@author{#1}%
  \else
    \xdef\@author{\unexpanded\expandafter{\@author\and#1}}%
  \fi
  \csgdef{author@\the\c@author}{#1}}
\newcommand*\email[1]{%
  \csgdef{email@\the\c@author}{#1}}
\newcommand*\address[1]{%
  \csgdef{address@\the\c@author}{#1}}
\AtEndDocument{%
  \xdef\author@count{\the\c@author}%
  \c@author=1
  \print@authors}
\newcommand*\print@authors{%
  \ifnum\c@author>\author@count
  \else 
    \print@author{\the\c@author}%
    \advance\c@author by 1
    \expandafter\print@authors
  \fi}
\newcommand*\print@author[1]{%
  \par\medskip
  \begin{tabular}{@{}l@{}}%
    \textsc{\csuse{author@#1}}\\
    \csuse{address@#1}\\
    \textit{E-Mail}:
    \href{mailto:\csuse{email@#1}}{\csuse{email@#1}}
  \end{tabular}}
\patchcmd{\NAT@test}{\else \NAT@nm}{\else \NAT@hyper@{\NAT@nm}}{}{}
\makeatother

\numberwithin{equation}{section}

\newtagform{normalsize}[\normalsize]{\normalsize(}{\normalsize)}

\usetagform{normalsize}

\theoremstyle{definition}

\numberwithin{definition}{section}
\theoremstyle{thoerem}
\newtheorem{theorem}{Theorem}
\numberwithin{theorem}{section}
\theoremstyle{plain}
\newtheorem{proposition}{Proposition}
\numberwithin{proposition}{section}
\theoremstyle{plain}
\newtheorem{corollary}{Corollary}
\numberwithin{corollary}{section}
\theoremstyle{remark}
\newtheorem{remark}{Remark}
\numberwithin{remark}{section}
\mathtoolsset{showonlyrefs}

\newcounter{example}[section]

\providecommand{\keywords}[1]
{
  \small	
  \textbf{\textit{Keywords---}} #1
}

\author{Jos\'e Miguel Flores-Contró}
\address{\textit{Department of Actuarial Science} \\ \textit{Faculty of Business and Economics} \\
\textit{University of Lausanne} \\ \textit{Lausanne, Switzerland}}
\email{josemiguel.florescontro@unil.ch}

\title{The Gerber-Shiu Expected Discounted Penalty Function: An Application to Poverty Trapping}

\begin{document}

\date{\vspace{-2ex}}

\maketitle

\begin{abstract}
In this article, we consider a risk process to model the capital of a household. Our work focuses on the analysis of the trapping time of such a process, where trapping occurs when a household\rq s capital level falls into the poverty area. A function analogous to the classical Gerber-Shiu function is introduced, which incorporates information on the trapping time, the capital surplus immediately before trapping and the capital deficit at trapping. We derive, under some assumptions, a model belonging to the family of generalised beta (GB) distributions that describes the distribution of the capital deficit at trapping given that trapping occurs. Affinities between the capital deficit at trapping and a class of poverty measures, known as the Foster-Greer-Thorbecke (FGT) index, are presented. The versatility of this model to estimate FGT indices is assessed using household microdata from Burkina Faso\rq s \textit{Enquête Multisectorielle Continue (EMC)} 2014.
\vspace{0.5cm}


\end{abstract}

\section{Introduction} \label{Introduction-Section1}

Recently, risk theory has proven to be a powerful tool to analyse a household\rq s infinite-time trapping probability (the probability of a household\rq s capital falling into the area of poverty at some point in time) (see, for instance, \cite{Article:Kovacevic2011}, \cite{Article:Flores-Contro2021} and \cite{Article:Henshaw2023}).  The classical risk process, also known as the Cram\'er-Lundberg model, which was introduced by Cram\'er and Lundberg at the beginning of the last century \citep{Book:Lundberg1903,Book:Lundberg1926, Book:Cramer1930}, has been adapted to better portray the capital of a household. For example, in \cite{Article:Kovacevic2011}, \cite{Article:Flores-Contro2021} and \cite{Article:Henshaw2023}, only the surplus of a household\rq s current capital above a critical capital level (or poverty line) grows exponentially, unlike the linear premium income for an insurer\rq s surplus in the Cram\'er-Lundberg model. Moreover, \cite{Article:Kovacevic2011} and \cite{Article:Henshaw2023} consider household capital losses as a proportion of the accumulated capital, yielding absolute losses that are serially correlated with each other and with the inter-arrival times of loss events. In contrast, losses in the Cram\'er-Lundberg model are given by a sequence of i.i.d. claim sizes and are subtracted from the insurer\rq s surplus rather than prorated.  Similar models with prorated jumps have been studied outside the actuarial science domain (see \cite{Article:Altman2002}, \cite{InProceedings:Altman2005} and \cite{Article:Lopker2008}, for an application of this type of model on data transmission over the internet; \cite{Article:Eliazar2004} and \cite{Article:Eliazar2006} for their use in representing the behaviour of physical systems with a growth-collapse pattern; and \cite{Article:Derfel2012} for the adoption of these processes to modelling the division and growth of cell-populations). 

This article examines the household capital process with proportional losses originally introduced in \cite{Article:Kovacevic2011} and subsequently studied in \cite{Article:Azais2015} and \cite{Article:Henshaw2023}. Previous work on this capital process focuses solely on studying the infinite-time trapping probability. Indeed, \cite{Article:Kovacevic2011} and \cite{Article:Azais2015} use numerical methods to estimate the trapping probability, without aiming to find an analytical solution for the probability. However, as stated by \cite{Book:Asmussen2010}, the ideal situation in risk theory is to derive closed-form solutions for trapping probabilities. To this end, \cite{Article:Henshaw2023} apply Laplace transform techniques to solve the infinitesimal generator of the household\rq s capital risk process and obtain a closed-form expression for the infinite-time trapping probability under the assumption of $Beta(\alpha, 1)-$distributed remaining proportions of capital.

Although the infinite-time trapping probability is a very important indicator for studying poverty dynamics, policy makers and other stakeholders may need additional information on other quantities to fully understand a household\rq s transition into poverty. A clear example of a quantity of interest is the income short-fall (or income gap), which is defined as the absolute value of the difference between a poor household\rq s income (or consumption) and some poverty line. A household\rq s income short-fall serves as key component in a number of poverty measures (see, for instance, \cite{Article:Sen1976}, where a simple poverty measure, the income-gap ratio, assesses the percentage of household\rq s mean income short-fall from the poverty line and \cite{Article:Foster1984}, where the well-known Foster-Greer-Thorbecke (FGT) index weights the income gaps of the poor to estimate the aggregate poverty of an economic entity). The primary objective of incorporating household levels of income short-fall in poverty measures is the elimination of certain measurement issues. That is, numerous poverty measures, such as the  head-count index, which calculates the proportion of the population living below the poverty line and has been considered as one of the most common indices for measuring poverty since the first studies of poverty were conducted (see, \cite{Book:Booth1889} and \cite{Book:Rowntree1901}), ignore the depth of poverty and the distribution of income among the poor, making them deficient as poverty indicators \citep{Article:Sen1976}. Consequently, this underlines the importance of exploring additional quantities such as a household\rq s income short-fall.

Apart from facilitating the study of the infinite-time trapping probability, classical risk theory provides additional tools that allow the examination of other quantities of interest, such as a household\rq s income short-fall at the trapping time (the time at which a household\rq s capital falls into the area of poverty), thus granting a much deeper understanding of a household\rq s transition into poverty. In particular, the Gerber-Shiu expected discounted penalty function, which was originally introduced by \cite{Article:Gerber1998}, gives information about three quantities: the time of ruin, the deficit at ruin, and the surplus prior to ruin, corresponding to the first time an insurer\rq s surplus becomes negative, the undershoot and the overshoot of the insurer\rq s surplus at ruin, respectively. These three random variables play an important role within the risk management strategy of an insurance company. For instance, risk measures such as the Value-at-Risk and the Tail-Value-at-Risk have a close link with the deficit at ruin, while from a monitoring perspective, the surplus prior to ruin could be thought of as an early warning signal for the insurance company. The (ruin) time at which any such event takes place is then of critical importance \citep{Article:Landriault2009}. Extensive literature on these variables exists for the Cram\'er-Lundberg model and its variations (see, for example, \cite{Article:Gerber1997}, \cite{Article:Gerber1998}, \cite{Article:Lin1999}, \cite{Article:Lin2000}, \cite{Article:Chiu2003}, \cite{Article:Landriault2009} and references therein). 

Certainly, a household\rq s trapping time can be thought of as the ruin time of an insurer, while the capital surplus prior to trapping and the capital deficit at trapping are analogous to the insurer\rq s surplus prior to ruin and the deficit at ruin, respectively. Therefore, the Gerber-Shiu expected discounted penalty function can be applied to study these quantities. Recently, for example, \cite{Article:Flores-Contro2021} emloyed the Gerber-Shiu expected discounted penalty function to study the distribution of the trapping time of a household\rq s capital risk process with deterministic growth and $Exp(\alpha)-$distributed losses. Using classical risk theory techniques,  \cite{Article:Flores-Contro2021} also assess how the introduction of an insurance policy alters the distribution of the trapping time. \cite{Inbook:Kovacevic2021} have also recently highlighted the importance of studying such trapping times to optimise the retention rates of insurance policies purchased by households. In this article, for the household capital process with proportional losses, we obtain closed-form expressions for the Gerber-Shiu expected discounted penalty function under the assumption of $Beta(\alpha, 1)-$distributed remaining proportions of capital. Thus, the first contribution of this article lies in the derivation of analytical equations for the Gerber-Shiu expected discounted penalty function, which to the best of our knowledge, have not been previously obtained for this particular risk process. 

Given the importance of the income short-fall and its key role in widely used poverty measures, the second contribution of this paper lies in obtaining a compelling microeconomic foundation, which emerges from the derivation of the Gerber-Shiu expected discounted penalty function for the household capital process, to model the distribution of the income short-fall. This is particularly important as parametric estimation of income distributions has long been used to model income since the introduction of the \cite{Inbook:Pareto1896} law. One of the main advantages of parametric estimation of income distributions is that explicit formulas, as functions of the parameters of the theoretical income distribution, are available to measure poverty and inequality. This allows, for example, to further interpret the shape parameters of the theoretical income distribution, as well as to carry out sensitivity analyses of poverty measures to variations in the shape parameters \citep{Article:Graf2014}.  In economics, it is well-known that the processes of income generation and distribution must be connected, underpinned by a microeconomic foundation, to the functional form of any model that adequately represents the distribution of personal income \citep{Article:CallealtaBarroso2020}. Our results reveal that the distribution of a household\rq s income short-fall belongs to the generalised beta (GB) distribution family, a group of models that have been widely used in economics for modelling income. 

To assess the validity of our results, we fit the derived GB model to household microdata from Burkina Faso\rq s \textit{Enquête Multisectorielle Continue (EMC) 2014}. Poverty measures are estimated using both the observed income short-fall data and the fitted theoretical income short-fall distribution. Goodness-of-fit tests and comparisons between theoretical and empirical poverty measures suggest that risk theory is a promising theoretical framework for studying poverty dynamics. That is, by appropriately adapting the classical Cram\'er-Lundberg model to better portray a household\rq s \ capital, risk theory provides a vast framework with a diverse set of tools to explore. The application of risk theory techniques to study poverty dynamics is just beginning and its potential is yet to be discovered.

The remainder of the article is organised as follows. In Section \ref{TheCapitalofaHousehold-Section2}, we introduce the capital of a household and its connection with the Cram\'er-Lundberg model. Section \ref{TheGeneralisedBetaDistributionFamily-Section3} provides a brief discussion on the GB distribution family and its application in economics for modelling income. In Section \ref{WhenandHowHouseholdsBecomePoor?-Section2}, the trapping time and the Gerber-Shiu expected discounted penalty function are defined. Moreover, an Integro-Differential Equation (IDE) for the Gerber-Shiu expected discounted penalty function is also derived.  We obtain in Section \ref{TheTrappingTime-Subsection21} a closed-form expression for the Laplace transform of the trapping time when the remaining proportion of capital is $Beta(\alpha, 1)-$distributed. Apart from characterising uniquely the probability distribution of the trapping time, Section \ref{TheTrappingTime-Subsection21} also shows how the Laplace transform of the trapping time can be applied to estimate other quantities of interest such as the expected trapping time. Likewise, Section \ref{TheCapitalDeficitatTrapping-Subsection22} studies the capital deficit at trapping by means of the Gerber-Shiu expected discounted penalty function for $Beta(\alpha, 1)-$distributed remaining proportions of capital and shows that the distribution of the capital deficit at trapping given that trapping occurs is described by a model belonging to the GB distribution family. Section \ref{AClassofPovertyMeasures-Section3} introduces the FGT index in more detail and discusses affinities between the index and the capital deficit at trapping. Built on Sections \ref{TheCapitalDeficitatTrapping-Subsection22} and \ref{AClassofPovertyMeasures-Section3}, a GB distribution is fitted to household microdata from Burkina Faso\rq s \textit{Enquête Multisectorielle Continue (EMC) 2014} in Section \ref{AnApplicationtoBurkinaFasosHouseholdMicrodata-Section4}. In addition, FGT indices are estimated using the fitted distribution. To evaluate the adequacy of the model, empirical values of the poverty measures are compared with theoretical estimates and goodness-of-fit tests are assessed. Lastly, concluding remarks are discussed in Section \ref{Conclusion-Section5}.

\section{The Capital of a Household} \label{TheCapitalofaHousehold-Section2}

In classical risk theory, the insurance risk process with deterministic investment $(U_{t})_{t\geq 0}$ is given by 

\vspace{0.3cm}

\begin{align}
  U_{t} =u+pt+ \nu \int_{0}^{t} U_{s} \, ds-\sum_{i=1}^{P_{t}} Y_{i},
   \label{TheCapitalofaHousehold-Section2-Equation1}
\end{align} 

\vspace{0.3cm}

where $u=U_{0} \geq 0$ is the insurer\rq s initial surplus, $p$ is the incoming premium rate per unit time, $\nu$ is the risk-free interest rate, $(P_{t})_{t\geq 0}$ is a Poisson process with intensity $\mathcal{I}$ counting the number of claims in the time interval $[0,t]$ and $(Y_{i})_{i=1}^\infty$ is a sequence of i.i.d. claim sizes with distribution function $G_{Y}$. Initially introduced by \cite{Article:Segerdahl1942}, this model was subsequently studied by \cite{Article:Harrison1977} and \cite{Article:Sundt1995}. Readers may wish to consult \cite{Article:Paulsen1998} for a detailed literature review on this model. 

Adopting traditional risk theory techniques, this article examines ideas proposed in \cite{Article:Kovacevic2011}. In particular, we study a household\rq s capital process $(X_{t})_{t\geq 0}$ with a deterministic exponential growth and multiplicative capital loss (collapse) structure. The process grows exponentially with a rate $r > 0$, which incorporates household rates of consumption ($0<a<1$), income generation ($0<b$) and investment or savings ($0<c<1$), above a critical capital (or poverty line) $x^{*} > 0$ whereas below this critical threshold it remains constant. At time $T_{i}$, the $i\text{th}$  capital loss event time of a Poisson process $(N_{t})_{t\geq 0}$ with parameter $\lambda$, the capital process jumps (downwards) to $Z_{i} \cdot X_{T_{i}}$, where $(Z_{i})_{i=1}^\infty$ is a sequence of i.i.d. random variables with distribution function $G_{Z}$ supported in $[0,1]$, independent of the process $N_{t}$, representing the proportions of remaining capital after each loss event. Therefore, a household\rq s capital process in between jumps is given by

\vspace{0.3cm}

\begin{align}
    X_{t}=\begin{cases} \left(X_{T_{i-1}}-x^{*}\right) e^{r \left(t-T_{i-1}\right)}+x^{*} & \textit { if } X_{T_{i-1}}>x^{*}, \\ X_{T_{i-1}} & \textit{otherwise},
    \end{cases}
    \label{TheCapitalofaHousehold-Section2-Equation2}
\end{align}

\vspace{0.3cm}

for $T_{i-1}\leq t < T_{i}$ and $T_{0}=0$. On the other hand, at the jump times $t = T_{i}$, the capital process is given by

\vspace{0.3cm}

\begin{align}
    X_{T_{i}}=\begin{cases} \left[\left(X_{T_{i-1}}-x^{*}\right) e^{r \left(T_{i}-T_{i-1}\right)}+x^{*}\right] \cdot Z_{i} & \textit { if } X_{T_{i-1}}>x^{*}, \\ X_{T_{i-1}} \cdot Z_{i} & \textit{otherwise}.
    \end{cases}
    \label{TheCapitalofaHousehold-Section2-Equation3}
\end{align}

\vspace{0.3cm}

The stochastic process $(X_{t})_{t\geq 0}$ is a piecewise-determinsitic Markov process \citep{Article:Davis1984, Book:Davis1993} and its infinitesimal generator is given by

\vspace{0.3cm}

\begin{align}
    (\mathcal{A} f)(x)=r(x-x^{*}) f^{\prime}(x) +\lambda \int_{0}^{1} \left[f(x \cdot z) - f(x)\right] \mathrm{d} G_{Z}(z),  \qquad x \ge x^{*}.
    \label{TheCapitalofaHousehold-Section2-Equation4}
\end{align}

\vspace{0.3cm}

There exist many similarities between the household capital process and other well-known risk processes. For instance, observe that when $p=0$, the insurance risk process \eqref{TheCapitalofaHousehold-Section2-Equation1} is equivalent to the household capital process above the critical capital $x^{*}=0$ with claim losses subtracted from the insurer\rq s surplus rather than prorated. Furthermore, taking the logarithm of a discretised version of the household capital process, that is, setting the critical capital $x^{*}=0$ and taking the logarithm of \eqref{TheCapitalofaHousehold-Section2-Equation3}, yields a version of the classical risk process (see, for instance, \cite{Article:Kovacevic2011} and \cite{Article:Henshaw2023}), also known as the Cram\'er-Lundberg model, introduced by Cram\'er and Lundberg at the beginning of the last century \citep{Book:Lundberg1903,Book:Lundberg1926, Book:Cramer1930}. This model considers linear premium income for the surplus of an insurance company with losses given by a sequence of i.i.d. claim sizes. Clearly, the Cram\'er-Lundberg model could also be seen as a particular case of the risk process \eqref{TheCapitalofaHousehold-Section2-Equation1} with $\nu=0$. Despite these resemblances, there are also a number of discrepancies between the household capital process and those commonly studied in the actuarial science literature. Firstly, only the surplus of a household\rq s current capital above the critical capital grows exponentially. Secondly, household losses are defined as a proportion of the accumulated capital, yielding absolute losses that are serially correlated with each other and with the inter-arrival times of loss events \citep{Article:Kovacevic2011, Article:Henshaw2023}.

\section{The Generalised Beta Distribution Family} \label{TheGeneralisedBetaDistributionFamily-Section3}

The probability density function (p.d.f.) of the generalised beta (GB) distribution family is given by

\vspace{0.3cm}

\begin{align}
GB(y ; a, b, c, p, q)=\frac{|a| y^{a p-1}\left(1-(1-c)(y / b)^a\right)^{q-1}}{b^{a p} \mathrm{~B}(p, q)\left(1+c(y / b)^a\right)^{p+q}} \qquad \textit { for } \qquad 0<y^a<\frac{b^a}{1-c},
    \label{TheGeneralisedBetaDistributionFamily-Section3-Equation1}
\end{align}

\vspace{0.3cm}

and zero otherwise, where $a \neq 0$; $ 0 \leq  c \leq 1$; $b,p,q > 0$; and $B(p,q)=\int_0^1 t^{p-1}(1-t)^{q-1} d t$ denotes the beta function (see, for instance, equation (6.2.1) from \cite{Book:Abramowitz1964}). The GB includes other distributions as special or limiting cases (see, for example, \cite{Article:McDonald1995}). In particular, the beta of the first kind (B1), with p.d.f.

\vspace{0.3cm}

\begin{align}
B1(y ; b, p, q):=GB(y ; a=1, b, c =0, p, q)=\frac{y^{p-1}(b-y)^{q-1}}{b^{p+q-1} B(p, q)} \qquad \textit { for } \qquad 0<y<b,
    \label{TheGeneralisedBetaDistributionFamily-Section3-Equation2}
\end{align}

\vspace{0.3cm}

arises as the model that describes the distribution of a household\rq s income short-fall, for the particular case in which the remaining proportions of capital $Z_{i}$ are $Beta(\alpha,1)-$distributed. Indeed, the results obtained in Section \ref{TheCapitalDeficitatTrapping-Subsection22} validate the adequacy of the B1 distribution as a model of income distribution and, in particular, as a model for the distribution of the income short-fall. \cite{Article:Thurow1970} was the first to adopt the standard beta distribution ($Beta(p,q) := B1(y ; b = 1, p, q) $) to analyse factors contributing to income inequality among whites and blacks. One of the main advantages of the beta distribution is that it includes the gamma distribution as a limiting case and therefore provides at least as good a fit as the gamma. This is an important feature, especially since the gamma distribution has also been considered to model income distribution \citep{Article:Salem1974}. In the 1980s, seeking to improve the goodness of fit of the two-parameter standard beta distribution, \cite{Article:McDonald1984} introduced the generalized beta of the first and second kind ($GB1:=GB(y ; a, b, c = 0, p, q)$ and $GB2:=GB(y ; a, b, c = 1, p, q)$), two four-parameter distributions that nest most of the previously used models of two and three parameters as special cases or limit distributions (e.g. the Singh-Maddala distribution \citep{Article:Singh976}). Subsequently, \cite{Article:McDonald1995} introduced \eqref{TheGeneralisedBetaDistributionFamily-Section3-Equation1}, a five-parameter distribution that has clearly played an important role for modelling income. In fact, many distributions (belonging or not to the GB distribution family) with a varying number of parameters have been used in the literature to model income (see \cite{Article:Hlasny2021} for a detailed survey). 

\section{When and How Households Become Poor?} \label{WhenandHowHouseholdsBecomePoor?-Section2}

Let 

\vspace{0.3cm}

\begin{align}
    \tau_{x}:=\inf \left\{t \geq 0: X_{t}<x^{*} \mid X_{0}=x\right\}
    \label{WhenandHowHouseholdsBecomePoor?-Section2-Equation1}
\end{align}

\vspace{0.3cm}

denote the time at which a household with initial capital $x \ge x^{*}$ falls into the area of poverty (the trapping time), where $\psi(x) = \mathbb{P}(\tau_{x} < \infty)$ is the infinite-time trapping probability. To study the distribution of the trapping time, we apply the Gerber-Shiu expected discounted penalty function at ruin, a concept commonly used in actuarial science \citep{Article:Gerber1998}, such that with a force of interest $\delta \ge 0$ and initial capital $x \ge x^{*}$, we consider

\vspace{0.3cm}

\begin{align}
    m_{\delta}(x)= \mathbb{E}\left[w(X_{\tau^{-}_{x}}- x^{*},\mid X_{\tau_{x}}-x^{*}\mid)e^{-\delta \tau_{x}} \mathbbm{1}_{\{\tau_{x} < \infty\}}\right],
    \label{WhenandHowHouseholdsBecomePoor?-Section2-Equation2}
\end{align}

\vspace{0.3cm}

where $\mathbbm{1}_{\{A\}}$ is the indicator function of a set $A$, and $w(x_{1}, x_{2})$, for $0 \leq x_{1} < \infty$ and  $0 < x_{2} \leq x^{*} $, is a non-negative penalty function of $x_{1}$, the capital surplus prior to the trapping time, and $x_{2}$, the capital deficit at the trapping time. For more details on the so-called Gerber-Shiu risk theory, interested readers may wish to consult \cite{Book:Kyprianou2013}. The function $m_{\delta}(x)$ is useful for deriving results in connection with joint and marginal distributions of $\tau_{x}$, $X_{\tau^{-}_{x}}- x^{*}$ and $\mid X_{\tau_{x}}-x^{*}\mid$. For example, when $\delta$ is considered as the argument, \eqref{WhenandHowHouseholdsBecomePoor?-Section2-Equation2} can be viewed in terms of a Laplace transform. That is, \eqref{WhenandHowHouseholdsBecomePoor?-Section2-Equation2} is the Laplace transform of the trapping time $\tau_{x}$ if one sets $w(x_{1}, x_{2})=1$\footnote{Recall that, for a continuous random variable $Y$, with p.d.f. $f_{Y}$, the Laplace transform of $f_{Y}$ is given by the expected value $\mathcal{L}\{f_{Y}\}\left(s\right)=\mathbb{E}\left[e^{-sY}\right]$.}. Another choice, for any fixed $y$, is $w(x_{1}, x_{2}) = \mathbbm{1}_{\{x_{2} \leq y\}}$ for $\delta =0$, for which \eqref{WhenandHowHouseholdsBecomePoor?-Section2-Equation2} leads to the distribution function of the capital deficit at trapping. It is not difficult to realise that, by appropriately choosing a penalty function $w(x_{1}, x_{2})$ and force of interest $\delta$, various risk quantities can be modeled. \cite{Article:He2023} provide a non-exhaustive list of such risk quantities. In this article, we are mainly interested in studying the Laplace transform of the trapping time and the distribution of the capital deficit at trapping. Thus, we will focus our analysis on the choices mentioned above. Following \cite{Article:Gerber1998}, our goal is to derive a functional equation for $m_{\delta}(x)$ by applying the law of iterated expectations to the right-hand side of \eqref{WhenandHowHouseholdsBecomePoor?-Section2-Equation2}.

\vspace{0.3cm}

\begin{theorem}\label{WhenandHowHouseholdsBecomePoor?-Section2-Theorem1}
The Gerber-Shiu expected discounted penalty function at trapping, $m_{\delta}(x)$, for $x\geq x^{*}$, satisfies the following Integro-Differential Equation (IDE)

\vspace{0.3cm}

\begin{align}
    \begin{split}
        r(x-x^{*})m_{\delta}'(x)-(\delta + \lambda)m_{\delta}(x) +\lambda \int_{x^{*}/x}^{1}m_{\delta}(x\cdot z)dG_{Z}(z) = -\lambda A(x),
        \label{WhenandHowHouseholdsBecomePoor?-Section2-Equation3}
    \end{split}
\end{align}

\vspace{0.3cm}

where $A(x):=\int_{0}^{x^{*}/x}w(x-x^{*}, x^{*}-x\cdot z)dG_{Z}(z)$, with boundary conditions

\vspace{0.3cm}

\begin{align}
m_{\delta}(x^{*})=
\frac{\lambda}{\delta + \lambda} A(x^{*})
\ \ \ \text{and} \ \ \ \lim_{x\rightarrow\infty}m_{\delta}(x) = 0.
\label{WhenandHowHouseholdsBecomePoor?-Section2-Equation4}
\end{align}

\end{theorem}

\vspace{0.3cm}

\begin{proof}

For $h>0$, consider the time interval $(0,h)$, and condition on the time $t$ and the proportion $z$ of remaining capital after the first capital loss in this time interval. Since the inter-arrival times of losses are exponentially distributed, the probability that there is no loss up to time $h$ is $e^{-\lambda h}$, and the probability that the first capital loss occurs between time $t$ and time $t+dt$ is $e^{-\lambda t} \lambda dt$. If 

\vspace{0.3cm}

\begin{align}
    z < \frac{x^{*}}{(x-x^*)e^{rt}+x^*}, \hspace{2cm}
    \label{WhenandHowHouseholdsBecomePoor?-Section2-Equation5}
\end{align}

\vspace{0.3cm}

where $0 < \frac{x^{*}}{(x-x^*)e^{rt}+x^*} \leq 1$, trapping has occurred with the first loss. Hence,

\vspace{0.3cm}

%\footnotesize
\begin{align}
    \begin{split}
        m_{\delta}(x)&= e^{-(\delta + \lambda)h}m_{\delta}((x-x^*)e^{rh}+x^*)\\
        &+\int_{0}^{h}\left[\int_{0}^{\frac{x^{*}}{(x-x^*)e^{rt}+x^*}}w\left((x-x^*)e^{rt}, x^*-((x-x^*)e^{rt}+x^*)\cdot z\right)dG_{Z}(z)\right]e^{-(\delta + \lambda)t}\lambda dt\\
        &+\int_{0}^{h}\left[\int_{\frac{x^{*}}{(x-x^*)e^{rt}+x^*}}^{1}m_{\delta}(((x-x^*)e^{rt}+x^*)\cdot z)dG_{Z}(z)\right]e^{-(\delta + \lambda)t}\lambda dt.
        \label{WhenandHowHouseholdsBecomePoor?-Section2-Equation6}
    \end{split}
\end{align}
\normalsize

\vspace{0.3cm}

Note that every part of the above integral equation (IE) is differentiable with respect to (w.r.t.) $h$. Thus, by symmetry one can also establish the differentiability of $m_{\delta}(x)$ w.r.t. $x$ (see, for example, Remark 1.11 in \cite{Book:Asmussen2010} where a similar argument is presented for the ruin probability of risk processes with non-proportional random-valued losses). Differentiating \eqref{WhenandHowHouseholdsBecomePoor?-Section2-Equation6} w.r.t. $h$ and setting $h=0$, \eqref{WhenandHowHouseholdsBecomePoor?-Section2-Equation3} is obtained.

\vspace{0.3cm}

\end{proof}

\vspace{0.3cm}

\subsection{The Trapping Time} \label{TheTrappingTime-Subsection21}

As noted previously, specifying the penalty function such that $w(x_{1}, x_{2}) = 1$, \eqref{WhenandHowHouseholdsBecomePoor?-Section2-Equation2} becomes the Laplace transform of the trapping time, also interpreted as the expected present value of a unit payment due at the trapping time. Thus, equation \eqref{WhenandHowHouseholdsBecomePoor?-Section2-Equation3} can then be written such that

\vspace{0.3cm}

\begin{align}
    \begin{split}
        0 = r(x-x^{*})m_{\delta}'(x)
        -(\delta + \lambda)m_{\delta}(x)
        +\lambda G_{Z}\left(\frac{x^{*}}{x}\right)
        +\lambda \int_{x^{*}/x}^{1}m_{\delta}(x\cdot z)dG_{Z}(z).
        \label{TheTrappingTime-Subsection21-Equation1}
    \end{split}
\end{align}

\vspace{0.3cm}

\begin{remark}

In general, it is not straightforward to obtain the solution of \eqref{TheTrappingTime-Subsection21-Equation1} for general distribution functions $G_{Z}$. Hence, throughout this article, it will be assumed that $Z_{i} \sim Beta(\alpha,1)$, case for which the distribution function is $G_{Z}(z) = z^{\alpha}$ and the p.d.f. is $g_{Z}(z)=\alpha z^{\alpha - 1}$ for $0 < z < 1$, where $\alpha >0$. Under this assumption, one can derive a closed-form expression for the Laplace transform of the trapping time.

\end{remark}

\vspace{0.3cm}

\begin{proposition}\label{TheTrappingTime-Subsection21-Proposition1}

Consider a household capital process defined as in \eqref{TheCapitalofaHousehold-Section2-Equation2} and \eqref{TheCapitalofaHousehold-Section2-Equation3}, with initial capital $x\ge x^{*}$, capital growth rate $r$, intensity $\lambda > 0$ and remaining proportions of capital with distribution $Beta(\alpha, 1)$ where $\alpha >0$; that is, $Z_{i}\sim Beta(\alpha, 1)$. The Laplace transform of the trapping time is given by

\vspace{0.3cm}

\begin{align}
    m_{\delta}(x)=\frac{\lambda \cdot  { }_{2} F_{1}\left(b, b-c+1 ; b-a+1 ; y(x)^{-1}\right)}{(\lambda + \delta){ }_{2} F_{1}\left(b, b-c+1 ; b-a+1 ; 1\right)} {y(x)^{-b}},
    \label{TheTrappingTime-Subsection21-Equation2}
 \end{align}

\vspace{0.3cm}

where $\delta \ge 0$ is the force of interest for valuation, ${ }_{2} F_{1}\left(\cdot \right)$ is Gauss\rq s Hypergeometric Function as defined in \eqref{TheTrappingTime-Subsection21-Equation7}, $y(x)=\frac{x}{x^{*}}$, $a=\frac{-(\delta + \lambda - \alpha r) - \sqrt{(\delta + \lambda -\alpha r)^{2}+4 r \alpha \delta}}{2r}$, $b=\frac{-(\delta + \lambda - \alpha r) + \sqrt{(\delta + \lambda -\alpha r)^{2}+4 r \alpha \delta}}{2r}$ and $c= \alpha$.

\end{proposition}

\begin{proof}

Under the assumption $Z_{i}\sim Beta(\alpha, 1)$, the IDE \eqref{TheTrappingTime-Subsection21-Equation1} can be written such that

\vspace{0.3cm}

\begin{align}
    \begin{split}
        0 = r(x-x^{*})m_{\delta}'(x)
        -(\delta + \lambda)m_{\delta}(x)
        +\lambda \left(\frac{x^{*}}{x}\right)^{\alpha}+\lambda \int_{x^{*}/x}^{1}m_{\delta}(x\cdot z)\alpha z^{\alpha-1}dz.
        \label{TheTrappingTime-Subsection21-Equation3}
    \end{split}
\end{align}

\vspace{0.3cm}

Applying the operator $\frac{d}{dx}$ to both sides of \eqref{TheTrappingTime-Subsection21-Equation3}, together with a number of algebraic manipulations, yields to the following second order Ordinary Differential Equation (ODE)

\vspace{0.3cm}

\begin{align}
    \begin{split}
        0 &= r(x^{2}-xx^{*})m_{\delta}''(x)+\left[(r(1+\alpha)-\delta-\lambda)x-r\alpha x^{*}\right]m_{\delta}'(x)-\alpha \delta m_{\delta}(x).
        \label{TheTrappingTime-Subsection21-Equation4}
    \end{split}
\end{align}

\vspace{0.3cm}
Letting $f(y):=m_{\delta}(x)$, such that $y$ is associated with the change of variable $y:=y(x)=\frac{x}{x^{*}}$, equation \eqref{TheTrappingTime-Subsection21-Equation4} reduces to Gauss\rq s Hypergeometric Differential Equation \citep{Book:Slater1960}

\vspace{0.3cm}

\begin{align}
    y(1-y)\cdot f''(y) + [c - (1+a+b)y] f'(y) - ab f(y) =0,
    \label{TheTrappingTime-Subsection21-Equation5}
\end{align}

\vspace{0.3cm}

for $a=\frac{-(\delta + \lambda - \alpha r) - \sqrt{(\delta + \lambda -\alpha r)^{2}+4 r \alpha \delta}}{2r}$, $b=\frac{-(\delta + \lambda - \alpha r) + \sqrt{(\delta + \lambda -\alpha r)^{2}+4 r \alpha \delta}}{2r}$ and $c= \alpha$, with regular singular points at $y=0, 1, \infty$ (corresponding to $x=0,x^{*},\infty$, respectively). A general solution of \eqref{TheTrappingTime-Subsection21-Equation5} in the neighborhood of the singular point $y=\infty$ is given by

\vspace{0.3cm}

\footnotesize
\begin{align}
    f(y):=m_{\delta}(x)= A_{1}{y(x)}^{-a} { }_{2} F_{1}\left(a, a-c+1 ; a-b+1 ; {y(x)}^{-1}\right)+A_{2}{y(x)}^{-b} { }_{2} F_{1}\left(b, b-c+1 ; b-a+1 ; {y(x)}^{-1}\right),\\
    \label{TheTrappingTime-Subsection21-Equation6}
\end{align}
\normalsize

\vspace{0.3cm}

for arbitrary constants $A_{1},A_{2} \in \mathbb {R}$ (see for example, equations (15.5.7) and (15.5.8) of \cite{Book:Abramowitz1964}). Here, 

\vspace{0.3cm}

\begin{align}
    { }_{2} F_{1}(a, b ; c ; z)=\sum_{n=0}^{\infty} \frac{(a)_{n}(b)_{n}}{(c)_{n}} \frac{z^{n}}{n !}
    \label{TheTrappingTime-Subsection21-Equation7}
\end{align}

is Gauss\rq s Hypergeometric Function \citep{Article:Gauss1812} and $(a)_{n}=\frac{\Gamma(a+n)}{\Gamma(n)}$ denotes the Pochhammer symbol \citep{Book:Seaborn1991}.

To determine the constants $A_1$ and $A_2$, we use the boundary conditions at $x^*$ and at infinity. The boundary condition $\lim\limits_{x\to\infty} m_{\delta}(x) = 0$, thus implies
that $A_{1}=0$. Letting $x=x^{*}$ in \eqref{TheTrappingTime-Subsection21-Equation3} and \eqref{TheTrappingTime-Subsection21-Equation6} yields

\vspace{0.3cm}

\begin{align}
    \frac{\lambda}{\lambda + \delta}=A_{2} \cdot { }_{2} F_{1}\left(b, b-c+1 ; b-a+1 ; 1\right).
    \label{TheTrappingTime-Subsection21-Equation8}
\end{align}

\vspace{0.3cm}

Hence, $A_{2}=\frac{\lambda }{(\lambda + \delta){ }_{2} F_{1}\left(b, b-c+1 ; b-a+1 ; 1\right)}$ and the Laplace transform of the trapping time is given by \eqref{TheTrappingTime-Subsection21-Equation2}.

\end{proof}

\begin{remark}
Figure \ref{TheTrappingTime-Subsection21-Figure1-a} shows that the Laplace transform of the trapping time approaches the trapping probability as $\delta$ tends to zero, i.e. 
    
    \vspace{0.3cm}
    
    \begin{align}
        \lim _{\delta \downarrow 0} m_{\delta}(x) =\mathbb{P}(\tau_{x}<\infty)\equiv\psi(x).
        \label{TheTrappingTime-Subsection21-Equation9}
    \end{align}
    
    \vspace{0.3cm}
    
    As $\delta\to 0$, \eqref{TheTrappingTime-Subsection21-Equation2} yields 
    
    \vspace{0.3cm}
    
    \begin{align}
        \psi(x) = \frac{{ }_{2} F_{1}\left(\alpha - \frac{\lambda}{r}, 1-\frac{\lambda}{r} ; 1+\alpha-\frac{\lambda}{r} ; y(x)^{-1}\right)}{{ }_{2} F_{1}\left(\alpha - \frac{\lambda}{r}, 1-\frac{\lambda}{r} ; 1+\alpha-\frac{\lambda}{r} ; 1\right)}y(x)^{\frac{\lambda}{r}-\alpha},
        \label{TheTrappingTime-Subsection21-Equation10}
    \end{align}

    \vspace{0.3cm}
 
    for $\alpha > \frac{\lambda}{r}$. Indeed, \eqref{TheTrappingTime-Subsection21-Equation10} was recently derived in \cite{Article:Henshaw2023} using Laplace transform techniques. Figure \ref{TheTrappingTime-Subsection21-Figure1-b} displays the trapping probability $\psi(x)$ for the capital process $X_{t}$. Note that, as mentioned in \cite{Article:Henshaw2023}, we can further simplify the expression for the trapping probability using some properties of Gauss\rq s Hypergeometric Function. Namely, 
    
        \vspace{0.3cm}
    
    \begin{align}
        { }_{2} F_{1}(a, b ; c ; 1)=\frac{\Gamma(c) \Gamma(c-a-b)}{\Gamma(c-a) \Gamma(c-b)}, \qquad \left(c \neq 0,-1,-2,..., \mathbbm{R}\left(c-a-b\right)>0\right)
        \label{TheTrappingTime-Subsection21-Equation11}
    \end{align}
    
    \vspace{0.3cm}

    (see, for example, equation (15.1.20) of \cite{Book:Abramowitz1964}). Applying this relation, we obtain
    
    \vspace{0.3cm}
    
    \begin{align}
        \psi(x) = \frac{\Gamma(\alpha) \cdot { }_{2} F_{1}\left(\alpha-\frac{\lambda}{r}, 1-\frac{\lambda}{r} ; 1+\alpha-\frac{\lambda}{r} ; y(x)^{-1}\right)}{\left(\alpha - \frac{\lambda}{r}\right) \Gamma\left(\alpha-\frac{\lambda}{r}\right) \Gamma\left(\frac{\lambda}{r}\right)} {y(x)^{\frac{\lambda}{r}-\alpha}}.
        \label{TheTrappingTime-Subsection21-Equation12}
    \end{align}
    
    \vspace{0.3cm}

\begin{figure}[H]
	\begin{subfigure}[b]{0.5\linewidth}
	   % Plot generated with the R code: UninsuredTrappingTimeLaplaceTransform.R
       % We could generate other plots if needed.
  		\includegraphics[width=8cm, height=8cm]{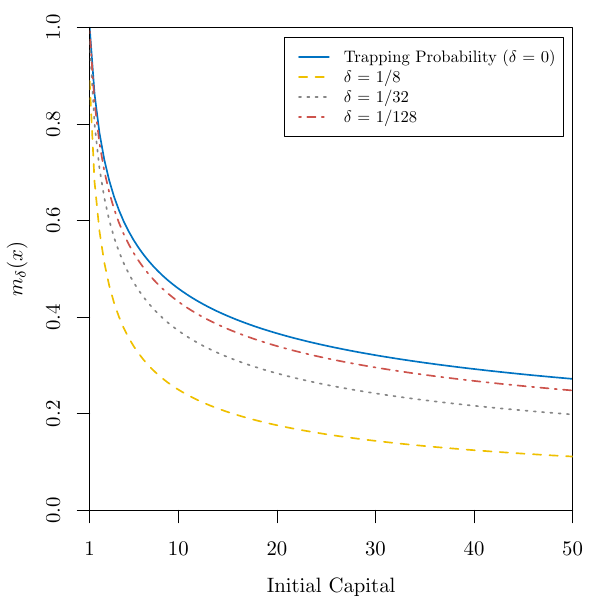}
		\caption{}
  		\label{TheTrappingTime-Subsection21-Figure1-a}
	\end{subfigure}
	\begin{subfigure}[b]{0.5\linewidth}
	   % Plot generated with the R code: UninsuredTrappingProbability.R
        % We could generate other plots if needed. Maybe be could look for a seed to fix the obtained results.
  		\includegraphics[width=8cm, height=8cm]{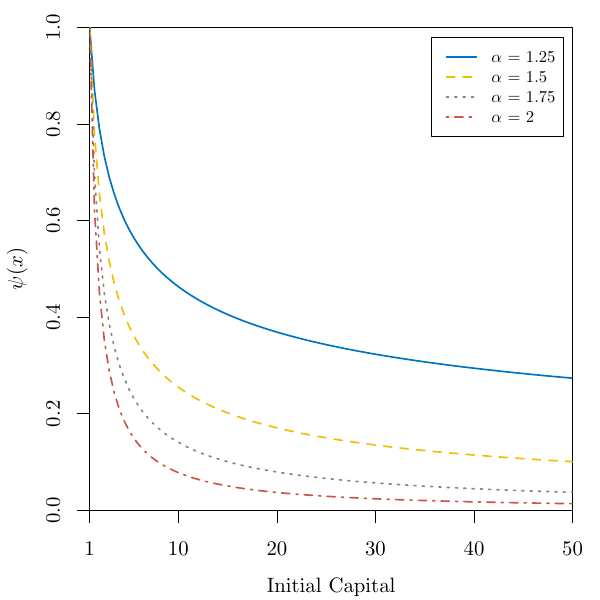}
		\caption{}
  		\label{TheTrappingTime-Subsection21-Figure1-b}
	\end{subfigure}
	\caption{(a) Laplace transform $m_{\delta}(x)$ of the trapping time when $Z_{i} \sim Beta(1.25, 1)$, $a = 0.1$, $b = 3$, $c = 0.4$, $\lambda = 1$, $x^{*} = 1$ for $\delta = 0, \frac{1}{8}, \frac{1}{32}, \frac{1}{128}$ (b) Trapping probability $\psi(x)$ when $Z_{i} \sim Beta(\alpha, 1)$, $a = 0.1$, $b = 3$, $c = 0.4$, $\lambda = 1$, $x^{*} = 1$ for $\alpha = 1.25, 1.5, 1.75, 2$.}
	\label{TheTrappingTime-Subsection21-Figure1}
\end{figure}

\end{remark}

\begin{remark}
    As an application of the Laplace transform of the trapping time, one particular quantity of interest is the expected trapping time; i.e. the expected time at which a household will fall into the area of poverty. This can be obtained by taking the derivative of $m_{\delta}(x)$:

    \vspace{0.3cm}

    \begin{align}
        \mathbb{E}\left[\tau_{x} ;\tau_{x}<\infty \right]=-\left.\frac{d}{d \delta} m_{\delta}(x)\right|_{\delta=0},
        \label{TheTrappingTime-Subsection21-Equation13}
    \end{align}
    
    \vspace{0.3cm}

    where $\mathbb{E}\left[\tau_{x} ;\tau_{x}<\infty \right]$ is equivalent to $\mathbb{E}\left[\tau_{x} \mathbbm{1}_{\{\tau_{x}<\infty\}} \right]$. As such, we differentiate Gauss\rq s Hypergeometric Function w.r.t. its first, second and third parameters. Denote

    \vspace{0.3cm}

{\allowdisplaybreaks
    \begin{align}
        { }_{2} F_{1}^{(a)}(a, b ; c ; z)&\equiv \frac{d}{da} { }_{2} F_{1}(a, b ; c ; z), \\
        { }_{2} F_{1}^{(b)}(a, b ; c ; z)&\equiv \frac{d}{db} { }_{2} F_{1}(a, b ; c ; z), \text{ and } \\
        { }_{2} F_{1}^{(c)}(a, b ; c ; z)&\equiv \frac{d}{dc} { }_{2} F_{1}(a, b ; c ; z).
        \label{TheTrappingTime-Subsection21-Equation14}
    \end{align}
 }
    
    \vspace{0.3cm}

    A closed-form expression of the aforementioned derivatives is given in terms of the Kamp\'e de F\'eriet Function \citep{Book:Appell1926}:
    
    \vspace{0.3cm}

    \begin{align}
        \begin{array}{l}
        F_{R, S, U}^{A, B, D}\left(\begin{array}{l}
        a_{1}, \ldots, a_{A} ; b_{1}, \ldots, b_{B} ; d_{1}, \ldots, d_{D} ; \\
        r_{1}, \ldots, r_{R} ; s_{1}, \ldots, s_{S} ; u_{1}, \ldots, u_{U} ;
        \end{array}x , y\right)
        =\sum\limits_{m=0}^{\infty} \sum\limits_{n=0}^{\infty} \frac{\prod\limits_{j=1}^{A}\left(a_{j}\right)_{m+n} \prod\limits_{j=1}^{B}\left(b_{j}\right)_{m} \prod\limits_{j=1}^{D}\left(d_{j}\right)_{n}}{\prod\limits_{j=1}^{R}\left(r_{j}\right)_{m+n} \prod\limits_{j=1}^{S}\left(s_{j}\right)_{m} \prod\limits_{j=1}^{U}\left(u_{j}\right)_{n}} \frac{x^{m}}{m !} \frac{y^{n}}{n !}
        \end{array},\\
        \label{TheTrappingTime-Subsection21-Equation15}
    \end{align}
    
    \vspace{0.3cm}
    
    such that (see, for example, equations (9a) and (9b) of \cite{Article:Ancarani2009}),

    \vspace{0.3cm}

    \begin{align}
        \begin{split}
            { }_{2} F_{1}^{(a)}(a, b ; c ; z)=\frac{z b}{c} F_{2,1,0}^{2,2,1}\left(\begin{array}{c}
            a+1, b+1 ; 1, a; 1 ; \\ 
            2, c+1 ; a+1 ; ;
        \end{array} z, z\right),\\ \\
        { }_{2} F_{1}^{(b)}(a, b ; c ; z)=\frac{z a}{c} F_{2,1,0}^{2,2,1}\left(\begin{array}{c}
            a+1, b+1 ; 1, b; 1; \\
            2, c+1 ; b+1 ; ;
        \end{array} z, z\right) \text{ and }\\ \\
       { }_{2} F_{1}^{(c)}(a, b ; c ; z)=-\frac{z a b}{c^{2}} F_{2,1,0}^{2,2,1}\left(\begin{array}{c}
        a+1, b+1 ; 1, c ; 1 ; \\
        2, c+1 ; c+1 ; ;
        \end{array} z, z\right).
        \end{split}
        \label{TheTrappingTime-Subsection21-Equation16}
    \end{align}
    
    \vspace{0.3cm}

\end{remark}

\vspace{0.3cm}

This is not the first time that the Kamp\'e de F\'eriet function appears in ruin theory, as it arises in the study of some risk processes that consider the payment of dividends provided by the insurer (see, for example, \cite{Article:AlbrecherCani2017}).

\vspace{0.3cm}

\begin{corollary}\label{TheTrappingTime-Subsection21-Corollary1}

The expected trapping time under the household capital process defined as in \eqref{TheCapitalofaHousehold-Section2-Equation2} and \eqref{TheCapitalofaHousehold-Section2-Equation3}, with initial capital $x \geq x^{*}$, capital growth rate $r$, intensity $\lambda > 0$ and remaining proportions of capital with distribution $Beta(\alpha, 1)$ where $\alpha >0$; that is, $Z_{i}\sim Beta(\alpha, 1)$ is given by

\vspace{0.3cm}

    \footnotesize
    {\allowdisplaybreaks
    \begin{align}
        %\begin{split}
        \mathbb{E}\left[\tau_{x} ;\tau_{x}<\infty \right] &= \frac{ 1 }{r (\alpha  r-\lambda ) \Gamma \left(\frac{\lambda }{r}\right)^2 \Gamma \left(\alpha -\frac{\lambda }{r}+1\right)^2} \Gamma (\alpha) y(x)^{\frac{\lambda }{r}-\alpha} \\ \\ & \left[ \Gamma(\alpha) { }_{2} F_{1} \left(1 - \frac{\lambda}{r}, \alpha - \frac{\lambda}{r} ; 1 + \alpha - \frac{\lambda}{r} ; y(x)^{-1}\right) \left( (\alpha r + \lambda)  { }_{2} F_{1}^{(c)}\left(\alpha - \frac{\lambda}{r}, 1 - \frac{\lambda}{r} ; 1 + \alpha - \frac{\lambda}{r} ; 1\right) \right. \right. \\ \\ & \left. + \lambda \left({ }_{2} F_{1}^{(a)}\left(\alpha - \frac{\lambda}{r}, 1 - \frac{\lambda}{r} ; 1 + \alpha - \frac{\lambda}{r} ; 1\right) + { }_{2} F_{1}^{(b)}\left(\alpha - \frac{\lambda}{r}, 1 - \frac{\lambda}{r} ; 1 + \alpha - \frac{\lambda}{r} ; 1\right) \right) \right) \\ \\ & \left. + \left(\frac{1}{\lambda}\right) \Gamma\left(\frac{\lambda}{r}\right)\Gamma\left(1 + \alpha - \frac{\lambda}{r}\right) \left({ }_{2} F_{1}\left(1 - \frac{\lambda}{r}, \alpha - \frac{\lambda}{r} ; 1 + \alpha - \frac{\lambda}{r} ; y(x)^{-1}\right) \left(r(\alpha r - \lambda) + \lambda^{2} \ln{\left[y(x)\right]}\right) \right. \right. \\ \\ & - \lambda \left( \left(\alpha r + \lambda \right){ }_{2} F_{1}^{(c)}\left(\alpha - \frac{\lambda}{r}, 1 - \frac{\lambda}{r} ; 1 + \alpha - \frac{\lambda}{r} ; y(x)^{-1}\right) \right. \\ \\ & \left. \left. \left.  + \lambda \left({ }_{2} F_{1}^{(a)}\left(\alpha - \frac{\lambda}{r}, 1 - \frac{\lambda}{r} ; 1 + \alpha - \frac{\lambda}{r} ; y(x)^{-1}\right) + { }_{2} F_{1}^{(b)}\left(\alpha - \frac{\lambda}{r}, 1 - \frac{\lambda}{r} ; 1 + \alpha - \frac{\lambda}{r} ; y(x)^{-1}\right) \right) \right) \right) \right],\\
        %\end{split},
        \label{TheTrappingTime-Subsection21-Equation17}
    \end{align}
    }
    \normalsize
    
    \vspace{0.3cm}

where $y(x)=\frac{x}{x^{*}}$, ${ }_{2} F_{1}\left(\cdot \right)$ is Gauss\rq s Hypergeometric Function as defined in \eqref{TheTrappingTime-Subsection21-Equation7} and ${ }_{2} F_{1}^{(a)}(\cdot)$, ${ }_{2} F_{1}^{(b)}(\cdot)$ and ${ }_{2} F_{1}^{(c)}(\cdot)$ its derivatives w.r.t. the first, second and third parameters, respectively, as introduced in \eqref{TheTrappingTime-Subsection21-Equation16}.

\end{corollary}

\begin{proof}
    
    Calculating \eqref{TheTrappingTime-Subsection21-Equation13} and using \eqref{TheTrappingTime-Subsection21-Equation16}, one can derive the expected trapping time \eqref{TheTrappingTime-Subsection21-Equation17}.
        
\end{proof}

    \vspace{0.3cm}
    
    % Plot generated with the Mathematica code: ExpectedTrappingTime.nb
    % We could generate other plots if needed.
    \begin{figure}[H]
        \centering
        \includegraphics[width=8cm, height=8cm]{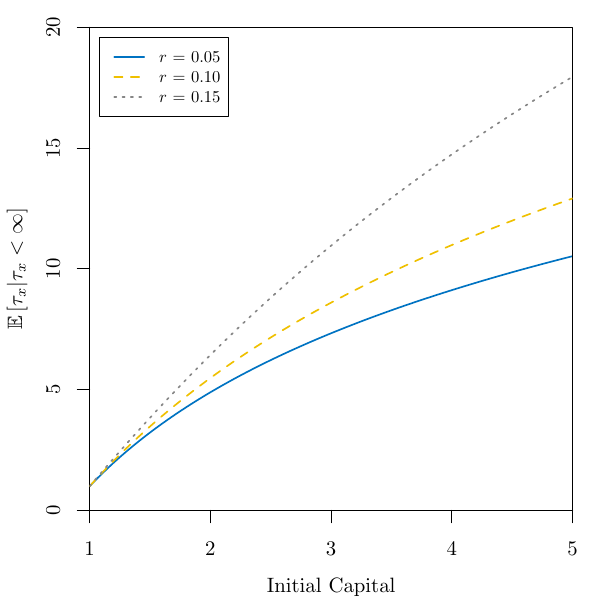}
 	    \caption{Expected trapping time given that trapping occurs $\mathbb{E}\left[\tau_{x}|\tau_{x}<\infty\right]$ when $Z_{i} \sim Beta(1.5, 1)$, $\lambda = 1$ and $x^{*} = 1$ for $r = 0.05,0.10,0.15$.}
	    \label{TheTrappingTime-Subsection21-Figure2}
    \end{figure}

\vspace{0.3cm}

Moreover, we can calculate the expected trapping time given that trapping occurs by taking the following ratio (see for example, equation (4.37) of \cite{Article:Gerber1998}),
    
    \vspace{0.3cm}
    
    \begin{align}
    \begin{split}
        \mathbb{E}\left[\tau_{x}|\tau_{x}<\infty\right]= \frac{\mathbb{E}\left[\tau_{x};\tau_{x}<\infty\right]}{\psi(x)}.
        \end{split}
        \label{TheTrappingTime-Subsection21-Equation18}
    \end{align}
    
    \vspace{0.3cm}  
    
A number of expected trapping times (given that trapping occurs) for varying values of the capital growth rate $r$ are displayed in Figure \ref{TheTrappingTime-Subsection21-Figure2}. One observes that the expected trapping time is an increasing function of both the capital growth rate $r$ and initial capital $x$.

\subsection{The Capital Deficit at Trapping} \label{TheCapitalDeficitatTrapping-Subsection22}

The capital deficit at trapping is the absolute value of the difference between a household\rq s level of capital at the trapping time and the critical capital, i.e. the amount ${\mid X_{\tau_{x}}-x^{*}\mid}$. Specifying the penalty function such that for any fixed $y$, $w(x_{1}, x_{2})=\mathbbm{1}_{\{x_{2} \leq y\}}$, \eqref{WhenandHowHouseholdsBecomePoor?-Section2-Equation2} becomes the distribution function of the capital deficit at the trapping time discounted at a force of interest $\delta \geq 0$. This choice leads to the following proposition

\vspace{0.3cm}

\begin{proposition}\label{TheCapitalDeficitatTrapping-Subsection22-Proposition1}
Consider a household capital process defined as in \eqref{TheCapitalofaHousehold-Section2-Equation2} and \eqref{TheCapitalofaHousehold-Section2-Equation3}, with initial capital $x\ge x^{*}$, capital growth rate $r$, intensity $\lambda > 0$ and remaining proportions of capital with distribution $Beta(\alpha, 1)$ where $\alpha >0$; that is, $Z_{i}\sim Beta(\alpha, 1)$. The distribution function of the discounted capital deficit at the trapping time is given by

\vspace{0.3cm}

\begin{align}
     F_{\delta}(y;\tau_{x}<\infty|x) &:=  \mathbb{E}\left[\mathbbm{1}_{\{\mid X_{\tau_{{\scaleto{x}{1.5pt}}}}-x^{*}\mid< y\}} e^{-\delta \tau_{x}} \mathbbm{1}_{\{\tau_{x} < \infty\}}\right]\\ \\&= m_{\delta}(x) \cdot \left[1 - \left(1 - \frac{y}{x^{*}}\right)^{\alpha} \right] \qquad \textit { for } \qquad 0 \leq y \leq x^{*},
    \label{TheCapitalDeficitatTrapping-Subsection22-Equation1}
 \end{align}

\vspace{0.3cm}

where $m_{\delta}(x)$ is the Laplace transform of the trapping time given by \eqref{TheTrappingTime-Subsection21-Equation2} and $\delta \ge 0$ is the force of interest for valuation.

\end{proposition}

\begin{proof}
The choice $w(x_{1}, x_{2})=\mathbbm{1}_{\{x_{2} \leq y\}}$ yields a modified version of the IDE \eqref{WhenandHowHouseholdsBecomePoor?-Section2-Equation3}, with \\ $A(x)=  y\left(x\right)^{-\alpha} -  \left(\frac{x^{*}-y}{x}\right)^{\alpha}$ for $y(x)=\frac{x}{x^{*}}$. Following a similar procedure to that of Proposition \ref{TheTrappingTime-Subsection21-Proposition1} leads to \eqref{TheCapitalDeficitatTrapping-Subsection22-Equation1}.
\end{proof}

\begin{remark}
One can easily obtain $f_{\delta}(y;\tau_{x}<\infty|x)$, the p.d.f. of the discounted capital deficit at the trapping time, by differentiating $F_{\delta}(y;\tau_{x}<\infty|x)$ w.r.t. $y$. That is,

 \vspace{0.3cm}

 %\small
 \begin{align}
     \begin{split}
         f_{\delta}(y;\tau_{x}<\infty|x) := \frac{d}{dy}F_{\delta}(y;\tau_{x}<\infty|x) = m_{\delta}(x) \cdot \frac{\alpha}{x^{*}} \left(1 - \frac{y}{x^{*}}\right)^{\alpha - 1}  \qquad \textit { for } \qquad 0 < y < x^{*},
         \label{TheCapitalDeficitatTrapping-Subsection22-Equation2}
     \end{split}
 \end{align}
 \normalsize

\vspace{0.3cm}
 
where $m_{\delta}(x)$ is the Laplace transform of the trapping time given by \eqref{TheTrappingTime-Subsection21-Equation2} and $\delta \ge 0$ is the force of interest for valuation.
 
\end{remark}

\vspace{0.3cm}

\begin{remark}
Note that, setting $\delta = 0$ yields $F(y;\tau_{x}<\infty|x)$, the distribution of the capital deficit at trapping. Furthermore, we can calculate the distribution of the capital deficit at trapping given that trapping has occurred. This is given by

 \vspace{0.3cm}

 \begin{align}
     \begin{split}
       F(y|x, \tau_{x}<\infty) := \frac{F(y;\tau_{x}<\infty|x)}{\psi(x)} = 1 - \left(1 - \frac{y}{x^{*}}\right)^{\alpha} \qquad \textit { for } \qquad 0 \leq y \leq x^{*}.
    \label{TheCapitalDeficitatTrapping-Subsection22-Equation3}
     \end{split}
 \end{align}

 \vspace{0.3cm}

Moreover, differentiating $F(y|x, \tau_{x}<\infty)$ w.r.t. $y$ leads to the p.d.f. of the capital deficit at trapping given that trapping has occurred,

 \vspace{0.3cm}

 \begin{align}
     \begin{split}
       f(y|x,\tau_{x}<\infty) := \frac{d}{dy}F(y|x,\tau_{x}<\infty) =\frac{\alpha}{x^{*}} \left(1 - \frac{y}{x^{*}}\right)^{\alpha - 1}  \qquad \textit { for } \qquad 0 < y < x^{*}.
       \label{TheCapitalDeficitatTrapping-Subsection22-Equation4}
     \end{split}
 \end{align}

 \vspace{0.3cm}

\end{remark}

\begin{figure}[H]
	\begin{subfigure}[b]{0.5\linewidth}
	   % Plot generated with the R code: UninsuredTrappingTimeLaplaceTransform.R
       % We could generate other plots if needed.
  		\includegraphics[width=8cm, height=8cm]{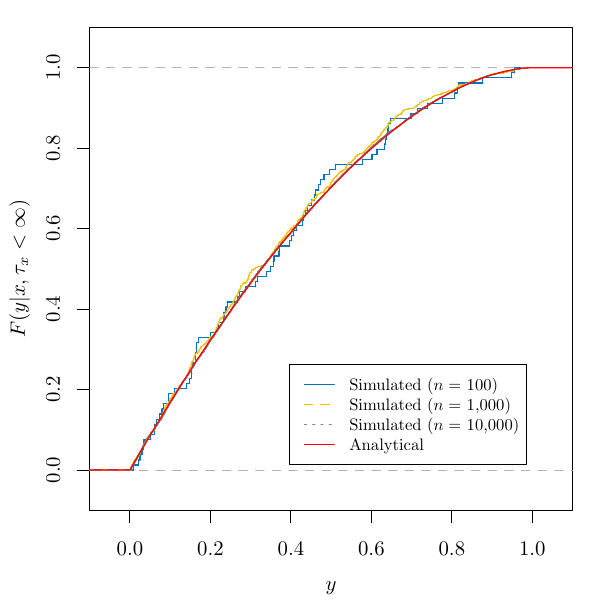}
		\caption{}
  		\label{TheCapitalDeficitatTrapping-Subsection22-Figure1-a}
	\end{subfigure}
	\begin{subfigure}[b]{0.5\linewidth}
	        % Plot generated with the R code: UninsuredTrappingProbability.R
        % We could generate other plots if needed. Maybe be could look for a seed to fix the obtained results.
  		\includegraphics[width=8cm, height=8cm]{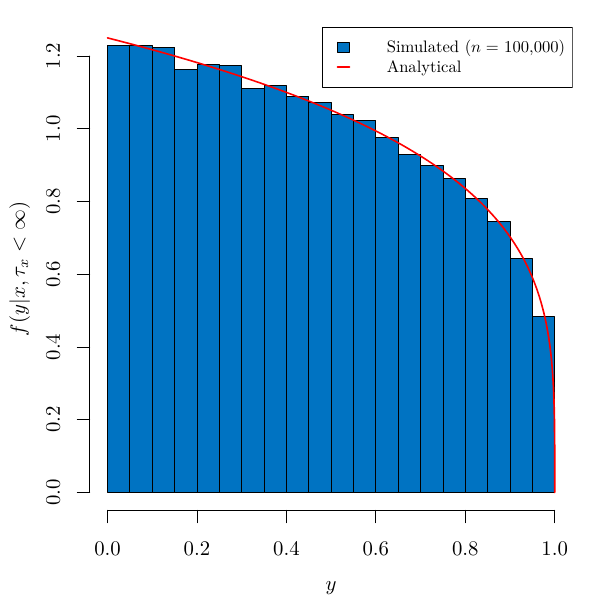}
		\caption{}
  		\label{TheCapitalDeficitatTrapping-Subsection22-Figure1-b}
	\end{subfigure}
	\caption{(a) $F(y|x, \tau_{x}<\infty)$ when $Z_{i} \sim Beta(1.75, 1)$, $r = 1.08$, $\lambda = 1$, $x = 1.25$, $x^{*} = 1$ (b) $f(y|x, \tau_{x}<\infty)$ when $Z_{i} \sim Beta(1.25, 1)$, $r = 1.08$, $\lambda = 1$, $x = 1.25$, $x^{*} = 1$.}
	\label{TheCapitalDeficitatTrapping-Subsection22-Figure1}
\end{figure}

Figure \ref{TheCapitalDeficitatTrapping-Subsection22-Figure1} compares both, analytical and simulated, distribution and p.d.f. of the capital deficit at trapping given that trapping occurs. Simulated quantities were generated using the Euler-Maruyama method, a well-known technique mainly used to approximate numerical solutions of Stochastic Differential Equations (SDEs) (see, for example, \citep{Book:Kloeden1995}). Not surprisingly, Figure \ref{TheCapitalDeficitatTrapping-Subsection22-Figure1-a} clearly shows that the simulated quantities converge to the theoretical distribution \eqref{TheCapitalDeficitatTrapping-Subsection22-Equation3} as the number of simulations $n$ increases, while Figure \ref{TheCapitalDeficitatTrapping-Subsection22-Figure1-b} displays how the theoretical p.d.f. given by \eqref{TheCapitalDeficitatTrapping-Subsection22-Equation4} perfectly fits the simulated observations. 

By comparing \eqref{TheGeneralisedBetaDistributionFamily-Section3-Equation2} with \eqref{TheCapitalDeficitatTrapping-Subsection22-Equation4} one concludes that the capital deficit at trapping given that trapping occurs follows the beta distribution of the first kind (B1). Indeed, if we denote the random variable $Y := \mid X_{\tau_{x}}-x^{*}\mid \Big| \ \tau_{x}<\infty$, we have that $Y \sim B1(y;b = x^{*},p=1,q=\alpha)$. Similarly, one can write $Y \stackrel{d}{=} x^{*} \cdot (1-Z_{i})$, where $\stackrel{d}{=}$ denotes equality in distribution.

\section{A Class of Poverty Measures and its Connection with the Capital Deficit at Trapping} \label{AClassofPovertyMeasures-Section3}

Poverty measures serve as the main tool for the evaluation of anti-poverty policies (e.g. cash transfer programmes) and poverty itself. Since \cite{Article:Sen1976}, following his axiomatic approach, researchers have formulated numerous poverty measures over the years. The Foster-Greer-Thorbecke (FGT) index \citep{Article:Foster1984} is undoubtedly one of the most important of these poverty measures and has been widely applied in empirical works. In fact, the FGT index has become the standard measure for international poverty assessments and is regularly reported on by individual countries and international organisations such as the World Bank (for a detailed review of the contributions of the FGT index over the 25 years since its publication, see \cite{Article:Foster2010}). The FGT index emerged as an alternative to the \lq \lq rank weighting\rq \rq \ approach, which was originally applied in the \lq \lq Sen measure\rq \rq  \ (see Theorem 1 from \cite{Article:Sen1976}), and accounts for the normalised gap and the rank order of a person in the group of the poor. The FGT index contemplates instead a \lq \lq short-fall weighting\rq \rq \ method, which considers the income short-fall expressed as a share of the poverty line. 

Let $F_{X}(x)$ be the distribution function of the income variable $X$ from a population with continuous p.d.f. $f_{X}(x)$ at a given point $x$. The FGT class of poverty measures indexed by $\gamma \geq 0$ is defined as follows

\vspace{0.3cm}

\begin{align}
	FGT_{\gamma}=\int^{x^{*}}_{0}\left(\frac{x^{*}-x}{x^{*}}\right)^{\gamma}f_{X}(x)\,dx,
	\label{AClassofPovertyMeasuresanditsConnectionwiththeCapitalDeficitatTrapping-Section3-Equation1}
\end{align}

\vspace{0.3cm}

where $x^{*}$ is the poverty line. Particular cases of the FGT class of poverty measures include $FGT_{0}$, which is simply the head-count index and as mentioned in Section \ref{Introduction-Section1}, calculates the proportion of households living below the poverty line. Another common measure is $FGT_{1}$, a normalisation of the income-gap ratio originally introduced by \cite{Article:Sen1976}. This poverty measure is commonly referred to as the poverty gap index. In contrast, the poverty severity index, $FGT_{2}$, is a weighted sum of income short-falls (as a proportion of the poverty line), where the weights are the proportionate income short-falls themselves. Note that, a larger $\gamma$ in \eqref{AClassofPovertyMeasuresanditsConnectionwiththeCapitalDeficitatTrapping-Section3-Equation1} gives greater emphasis to the poorest poor. Hence, this parameter is viewed as a measure of poverty aversion \citep{Article:Foster1984}. 

From \eqref{AClassofPovertyMeasuresanditsConnectionwiththeCapitalDeficitatTrapping-Section3-Equation1}, one can write 

\vspace{0.3cm}

\begin{align}
	FGT_{\gamma}=\frac{H(x^{*})}{{x^{*}}^{\gamma}}\cdot \mathbbm{E}\left[D(x^{*}, x)| x<x^{*}\right],
	\label{AClassofPovertyMeasuresanditsConnectionwiththeCapitalDeficitatTrapping-Section3-Equation2}
\end{align}

\vspace{0.3cm}

where $H(x^{*}) := F_{X}(x^{*})$ is the head-count index and $D(x^{*},x)=\left(x^{*}-x\right)^{\gamma}$ is a function that describes the level of deprivation suffered by an individual whose income $x$ is less than the poverty line $x^{*}$. Clearly, $D(x^{*},x)$ is in terms of an individual\rq s income short-fall $Y:=x^{*}-x$. 

We now consider a household\rq s capital process as defined in Section \ref{Introduction-Section1}. Under this model, a household\rq s income is generated through capital: $I_{t}=bX_{t}$, where $b>0$ holds (see equation (4) in \cite{Article:Kovacevic2011}). Taking $b = 1$ leads to the case for which a household\rq s income is equal to its capital. Thus, the results obtained in Section \ref{WhenandHowHouseholdsBecomePoor?-Section2} also apply to a household\rq s income. On this basis, from Section \ref{TheCapitalDeficitatTrapping-Subsection22} yields that $Y\sim B1(y;b=x^{*},p=1,q=\alpha)$, where the random variable $Y$ denotes the income short-fall (or income deficit) at trapping given that trapping occurs. In this case, the $FGT_{\gamma}$ index is given in terms of the $\gamma$th moment of $Y$,

\vspace{0.3cm}

\begin{align}
	FGT_{\gamma}=\frac{H(x^{*})}{{x^{*}}^{\gamma}} \cdot \mathbbm{E}\left[Y^{\gamma}\right]=H(x^{*})\cdot \left[\frac{B\left(1+\alpha,\gamma\right)}{B\left(1,\gamma\right)}\right],
	\label{AClassofPovertyMeasuresanditsConnectionwiththeCapitalDeficitatTrapping-Section3-Equation3}
\end{align}

\vspace{0.3cm}

where we used the fact that the $h$th moment of a random variable $W\sim B1(w;b,p,q)$ is given by

\vspace{0.3cm}

\begin{align}
	\mathbbm{E}\left[W^{h}\right]=\frac{b^{h}B(p+q,h)}{B(p,h)},
	\label{AClassofPovertyMeasuresanditsConnectionwiththeCapitalDeficitatTrapping-Section3-Equation4}
\end{align}

\vspace{0.3cm}

(see, for instance, Table 1 from \cite{Article:McDonald1984}).

\vspace{0.3cm}

\begin{remark}

One can also compute the $h$th moment of the capital deficit at trapping given that trapping occurs by means of the Gerber-Shiu expected discounted penalty function. Indeed, choosing $w(x_{1}, x_{2})=x_{2}^{h}$ yields a modified version of the IDE \eqref{WhenandHowHouseholdsBecomePoor?-Section2-Equation3}, with $A(x)= \alpha \cdot {x^{*}}^{h} \cdot B(\alpha, h+1) \cdot y\left(x\right)^{-\alpha}$ for $y(x)=\frac{x}{x^{*}}$. Thus, solving \eqref{WhenandHowHouseholdsBecomePoor?-Section2-Equation3} as in Proposition \ref{TheTrappingTime-Subsection21-Proposition1} yields to the $h$th moment of the discounted capital deficit at trapping,

\vspace{0.3cm}

\begin{align}
	\mathbbm{E}\left[\mid X_{\tau_{x}}-x^{*}\mid ^{h} e^{-\delta \tau_{x}} ; \tau_{x} < \infty \right]  = \alpha \cdot  {x^{*}}^{h} \cdot  		B\left(\alpha, h+1\right) \cdot m_{\delta}(x).
	\label{AClassofPovertyMeasuresanditsConnectionwiththeCapitalDeficitatTrapping-Section3-Equation5}
\end{align}

\vspace{0.3cm}

Setting $\delta = 0$ yields $\mathbbm{E}\left[\mid X_{\tau_{x}}-x^{*}\mid ^{h} ; \tau_{x} < \infty \right]$, the $h$th moment of the capital deficit at trapping. Consequently, the $h$th moment of the capital deficit at trapping given that trapping occurs is given by

\vspace{0.3cm}

\begin{align}
	\mathbbm{E}\left[Y^{h}\right]  := \frac{\mathbbm{E}\left[\mid X_{\tau_{x}}-x^{*}\mid ^{h} ; \tau_{x} < \infty \right]}{\psi(x)} = \alpha \cdot  {x^{*}}^{h} \cdot  B\left(\alpha, h+1\right) =\frac{\alpha \cdot x^{*h}\cdot h \cdot B(\alpha,h)}{h+\alpha},
	\label{AClassofPovertyMeasuresanditsConnectionwiththeCapitalDeficitatTrapping-Section3-Equation6}
\end{align}

\vspace{0.3cm}

where we applied the property $B(p,q+1)=\frac{q \cdot B(p,q)}{p+q}$ of the beta function. Clearly, \eqref{AClassofPovertyMeasuresanditsConnectionwiththeCapitalDeficitatTrapping-Section3-Equation6} is equal to \eqref{AClassofPovertyMeasuresanditsConnectionwiththeCapitalDeficitatTrapping-Section3-Equation4} for the case $b=x^{*}$, $p=1$ and $q=\alpha$.

\end{remark}

\section{An Application to Burkina Faso’s Household Microdata}\label{AnApplicationtoBurkinaFasosHouseholdMicrodata-Section4}

\subsection{Context and Data}  \label{ContextandData-Subsection41}

Burkina Faso is located in West Africa with an area of $274,200$ $\text{km}^{2}$. In 2021, the population was estimated at just over $20.3$ million, with the capital Ouagadougou being the country\rq s largest city. Historically, its economy has been largely based on agriculture, which provides a living for more than $80\%$ of the population. Burkina Faso\rq s main subsistence crops are sorghum, millet, maize and rice, while the country has been one of Africa\rq s leading producers of cotton and gold \citep{Article:Brugger2020,Article:Engels2023}. 

The country\rq s climate is characterised by a dry tropical climate that alternates a short rainy season with a long dry season. Due to its geographical location, bordering the Sahara Desert, Burkina Faso's climate is subject to seasonal and annual variations. Furthermore, the country is divided into three different climatic zones, the Sahelian zone in the north, the North-Sudanian zone in the centre and the South-Sudanian zone in the south, which receive an average annual rainfall of less than $600$ mm, between $600$ and $900$ mm and more than $900$ mm, respectively \citep{Article:AlvarBeltran2020}.

Household microdata from Burkina Faso\rq s Continuous Multisector Survey (\textit{Enquête Multisectorielle Continue (EMC)) 2014\footnote{For a detailed overview of the survey, interested readers may wish to consult the survey\rq s official report (in French): \cite{Book:INSD2015}.}} is used to evaluate the adequacy of the $B1(y;b=x^{*},p=1,q=\alpha)$ model \eqref{TheGeneralisedBetaDistributionFamily-Section3-Equation2} to describe income short-fall distribution. The survey was conducted from 17 January 2014 to 24 November 2014 by the National Institute of Statistics and Demography (\textit{Institut National de la Statistique et de la Démographie} (INSD)). The EMC had as main objective the generation of sound data for poverty analyses. A total of $10,411$ households were interviewed, with a $96.4\%$ of interviews accepted. 

The main variable of interest generated in the survey is consumption, which in the EMC is given in units of the West African CFA (\textit{Communauté financière en Afrique}) franc per person per day in average prices in Ouagadougou during the EMC field work. To identify the poor, a minimum food basket of around thirty products was defined. Determining the cost of this food basket and other basic needs, the absolute poverty line was estimated at $153,530$ CFA. A person is poor if he/she lives in a poor household and a household is poor if the annual per capita consumption is below the absolute poverty line which is equivalent to $421$ CFA per capita consumption per day.

\subsection{Estimators for the $\alpha$ Parameter of the B1 Model}\label{EstimatorsforParametersoftheB1Model-Subsection42}

In this article, we use the maximum likelihood (ML) method and the method-of-moments (MoM) to estimate the parameter $\alpha$ of the $B1(y;b=x^{*},p=1,q=\alpha)$ model \eqref{TheGeneralisedBetaDistributionFamily-Section3-Equation2}. Assume that $y_{1}, y_{2}, ...,y_{n}$ is a random sample of income short-fall of size $n$.  Letting $M_{k}=\frac{1}{n}\sum^{n}_{i=1}y^{k}_{i}$ denote the $k$th sample moment yields to the maximum likelihood estimator (MLE) and the method-of-moments estimator (MME) for $\alpha$, given by

\vspace{0.3cm}

\begin{align}
	\hat{\alpha}_{{\scaleto{MLE}{3pt}}} = \frac{n}{n \log{\left(x^{*}\right)}-\sum\limits_{i=1}^{n}\log{\left(x^{*}-y_{i}\right)}} \qquad \text{and} \qquad \hat{\alpha}_{{\scaleto{MME}{3pt}}} = \frac{x^{*}-M_{1}}{M_{1}},
	\label{EstimatorsforParametersoftheB1Model-Subsection42-Equation1}
\end{align}

\vspace{0.3cm}

respectively. We derive $\hat{\alpha}_{{\scaleto{MLE}{3pt}}}$ by maximising the log-likelihood function

\begin{align}
\ell\left(\alpha\right)=\ell\left(\alpha|\mathbf{y}\right):=\log L\left(\alpha|\mathbf{y}\right)=n\cdot \left[\log\left(\alpha\right) - \alpha \cdot \log\left(x^{*}\right)\right] + \left(\alpha - 1\right)\cdot\sum\limits_{i=1}^{n}\log\left(x^{*}-y_{i}\right),\label{EstimatorsforParametersoftheB1Model-Subsection42-Equation2}
\end{align}

where $L\left(\alpha|\mathbf{y}\right) = \prod_{i=1}^{n}f_{Y}\left(y_{i};b=x^{*},p=1,q=\alpha\right)=\left(\frac{\alpha}{x^{*}}\right)^{n}\prod_{i=1}^{n}\left(x^{*}-y_{i}\right)^{\alpha-1}$ is the likelihood function. Thus, we differentiate \eqref{EstimatorsforParametersoftheB1Model-Subsection42-Equation2} w.r.t $\alpha$ and equate it to zero. We then solve for the parameter $\alpha$ to obtain $\hat{\alpha}_{{\scaleto{MLE}{3pt}}}$. On the other hand, $\hat{\alpha}_{{\scaleto{MME}{3pt}}}$ is derived by equating the first sample moment ($M_{1}$) with the theoretical first moment (equation \eqref{AClassofPovertyMeasuresanditsConnectionwiththeCapitalDeficitatTrapping-Section3-Equation6} for $h=1$) and by subsequently solving for the parameter $\alpha$. Tables \ref{ResultsandDiscussion-Subsection44-Table1}, \ref{ResultsandDiscussion-Subsection44-Table2} and \ref{ResultsandDiscussion-Subsection44-Table3} show the MLEs and MMEs for $\alpha$ at the national level, by area of residence and by region, respectively. In addition, the maps in Figure \ref{ResultsandDiscussion-Subsection44-Figure1} display these estimates by region, giving a comprehensive geographical overview of the parameters. These results will be discussed more in detail in Section \ref{ResultsandDiscussion-Subsection44}.

\subsection{Evaluating the Goodness-of-Fit of the B1 Model} \label{EvaluatingtheGoodness-of-FitoftheB1Model-Subsection43}

The non-parametric one-sample Kolmogorov-Smirnov (KS) test and the $R^{2}$  coefficient are used to assess the goodness-of-fit of the B1 model. To conduct the KS test, we calculate the KS statistic, which is given by

\vspace{0.3cm}

\begin{align}
D=\max_{y}\left|F_{n}(y)-F(y)\right|,
	\label{EvaluatingtheGoodness-of-FitoftheB1Model-Subsection43-Equation1}
\end{align}

\vspace{0.3cm}

where $F_{n}(y)$ is the empirical distribution function defined as $F_{n}(y)=\frac{1}{n}\sum^{n}_{i=1}\mathbbm{1}_{\{y_{i}\leq y\}}$ and $F(y)$ is \eqref{TheCapitalDeficitatTrapping-Subsection22-Equation3}, the distribution function of the B1 model. Here, we follow \cite{Book:Hollander1999} to estimate the KS statistic $D$ for tied observations. The null ($H_{0}$) and alternative ($H_{1}$) hypotheses of the KS goodness-of-fit test are: 

$H_{0}$\textit{: the household income short-fall data follows the B1 model} and 

$H_{1}$\textit{: the household income short-fall data does not follows the B1 model}. 

The null hypothesis $H_{0}$ is rejected at a significance level $\alpha_{{\scaleto{KS}{3pt}}}$ if the p-value of the KS statistic is less than $\alpha_{{\scaleto{KS}{3pt}}}$. In this article, we use simulation to estimate the true p-value (see, for example, \cite{Book:Ross2023}).

We further support the KS test by considering the $R^{2}$ coefficient, which quantifies the degree of correlation between the observed and predicted probabilities under an assumed distribution. Here, a value of $R^{2}$ that is close to one indicates that the B1 model is a good fit for the household income short-fall data. The $R^{2}$ coefficient is computed as follows:

\vspace{0.3cm}

\begin{align}
R^{2}=\frac{\sum\limits_{i=1}^{n}\left[\hat{F}\left(y_{i}\right)-\bar{F}(y)\right]^2}{\sum\limits_{i=1}^n\left[\hat{F}\left(y_{i}\right)-\bar{F}(y)\right]^2+\sum\limits_{i=1}^{n}\left[F_{n}\left(y_{i}\right)-\hat{F}\left(y_{i}\right)\right]^2},
\label{EvaluatingtheGoodness-of-FitoftheReverseB1Model-Subsection42-Equation2}
\end{align}

\vspace{0.3cm}

where $F_{n}(y_{i})$ is the empirical distribution function for the $i$th household income short-fall, $\hat{F}(y_{i})$ is the estimated distribution function for the $i$th household income short-fall under the B1 model and $\bar{F}(y)$ is the average of $\hat{F}(y_{i})$.

To verify our assumptions and model specifications, we also consider graphical methods. We plot the distribution function \eqref{TheCapitalDeficitatTrapping-Subsection22-Equation3} and the p.d.f. \eqref{TheCapitalDeficitatTrapping-Subsection22-Equation4} against the empirical distribution $F_{n}(y)$ and the histogram of the observed income short-fall data, respectively. In addition, we use the B1 model quantile-quantile (Q-Q) and probability-probability (P-P) plots to support the assumption of a $B1(y;b = x^{*},p=1,q=\alpha)$ distribution. Appendices \ref{AppendixA:GoodnessofFitPlotsforBurkinaFaso}, \ref{AppendixB:GoodnessofFitPlotsbyAreaofResidence} and \ref{AppendixC:GoodnessofFitPlotsbyRegion} show these graphical methods at the national level, by area of residence and by region, respectively.

\subsection{Results and Discussion} \label{ResultsandDiscussion-Subsection44}

The assumption of the B1 distribution can be investigated based on the p-value of the Kolmogorov-Smirnov (KS) statistic and the $R^{2}$ coefficient, which are shown in Tables \ref{ResultsandDiscussion-Subsection44-Table1}, \ref{ResultsandDiscussion-Subsection44-Table2} and \ref{ResultsandDiscussion-Subsection44-Table3}. Tables \ref{ResultsandDiscussion-Subsection44-Table1} and \ref{ResultsandDiscussion-Subsection44-Table2} show that, at the national level and by area of residence, the null hypothesis $H_{0}$ is rejected and it is therefore concluded that the income short-fall data does not follow the B1 distribution. On the other hand, Table \ref{ResultsandDiscussion-Subsection44-Table3} exhibits that at the region level, with exception of the estimates for Nord and Plateau Central, all p-values of the KS test are higher than the significance level of $\alpha_{{\scaleto{KS}{3pt}}}=0.05$ (for either the MLE or the MME estimator or for both estimators). This indicates that we fail to reject the null hypothesis $H_{0}$ and therefore suggests we do not have sufficient evidence to conclude that the income short-fall data does not follow the B1 distribution for all other regions. This is borne out by the estimated values for the $R^{2}$ coefficient, which are found higher than $0.90$ for all the cases (except for the MLE estimate of the Nord region for which the null hypothesis $H_{0}$ is rejected), suggesting that the B1 model explains more than $90\%$ of the variation in the data, but the remaining (less than $10\%$) variation is attributed to errors and cannot be explained by the model.

\vspace{0.3cm}

\begin{table}[H]
\centering
\resizebox{\textwidth}{!}{
\begin{tabular}{llccccccccc}
\hline \hline Country & Type & $\mathrm{\alpha}$ & p-value (KS test) & $R^{2}$ & Poverty Gap Index ($FGT_{1}$) & Poverty Severity Index ($FGT_{2}$) \\
\hline \hline 
Burkina Faso & Direct &  - &  - &  - &  0.096 &  0.032 \\
Burkina Faso & MLE & 3.37 & 0.0000 & 0.9522 & 0.092 & 0.034 \\
Burkina Faso & MME & 3.16 & 0.0000 & 0.9670 & 0.096 & 0.037 \\
\hline \hline *p-value$>\alpha_{{\scaleto{KS}{3pt}}}=0.05.$
\end{tabular}}
\caption{B1 distribution fitted parameters and poverty measures for Burkina Faso.}
\label{ResultsandDiscussion-Subsection44-Table1}
\end{table}

\vspace{0.3cm}

Graphical methods displayed in Appendices \ref{AppendixA:GoodnessofFitPlotsforBurkinaFaso}, \ref{AppendixB:GoodnessofFitPlotsbyAreaofResidence} and \ref{AppendixC:GoodnessofFitPlotsbyRegion} provide an additional tool to evaluate the assumption of the B1 distribution. Plot (a) in the Appendices shows a density plot in which are plotted the p.d.f. \eqref{TheCapitalDeficitatTrapping-Subsection22-Equation4} of a B1 model and the histogram of the observed income short-fall data. On the other hand, plot (c) in the Appendices displays a distribution plot in which the B1 distribution function \eqref{TheCapitalDeficitatTrapping-Subsection22-Equation3} is plotted against the empirical distribution function $F_{n}(y)$. From Appendix \ref{AppendixC:GoodnessofFitPlotsbyRegion}, one can observe that the B1 model best fits the histogram and the empirical distribution for Centre, Hauts-Bassins and Sahel, which are the regions with the highest p-values of the KS statistic. Moreover, if the household income short-fall data is found to follow the B1 model, observations on both the quantile-quantile (Q-Q) and probability-probability (P-P) plots will appear to form almost a straight diagonal line. In general, household income short-fall data is positioned in the diagonal by region. This behaviour is even clearer for the three regions mentioned above. However, it is important to note that the model fails to capture the tail of the distribution.

\vspace{0.3cm}

Parameter estimates based on the maximum likelihood (ML) method and the Method of Moments (MoM) are shown in Tables \ref{ResultsandDiscussion-Subsection44-Table1}, \ref{ResultsandDiscussion-Subsection44-Table2} and \ref{ResultsandDiscussion-Subsection44-Table3}. From Section \ref{Introduction-Section1}, one can realise that under the assumption of $Beta(\alpha,1)-$distributed remaining proportions of capital, higher values of $\alpha$ yield to a greater expected remaining proportion of capital upon experiencing a capital loss (i.e. the distribution is left-skewed or equivalently, the remaining proportions of capital $Z_{i} \in [0,1]$ are more likely to have values close to one). On this basis, it is possible to assess the magnitude of capital losses experienced by households in Burkina Faso. For instance, Figures \ref{ResultsandDiscussion-Subsection44-Figure1-a} and \ref{ResultsandDiscussion-Subsection44-Figure1-b} display how households experience capital losses of varying magnitude depending on the geographical area in which they reside.  These findings can be useful, for example, in assessing the sensitivity of geographical regions to climate risks such as droughts and floods. Indeed, Burkina Faso\rq s economy is heavily dependent on rain-fed agriculture and livestock husbandry, which in turn makes it vulnerable to these risks (see, for example, \cite{Article:Zampaligre2014}). Moreover, the frequency and severity of these risks vary according to the climatic zone of the country\footnote{See \cite{Article:AlvarBeltran2020} for a detailed map of Burkina Faso with the different climatic zones.}. It is important to note that the poverty gap index ($FGT_{1}$) and the poverty severity index ($FGT_{2}$) shown in Table \ref{ResultsandDiscussion-Subsection44-Table3} for Boucle du Mouhoun, Centre-Ouest and Nord attain the highest values. This is also a result of the high head-count index in these regions (56\%, 47\% and 65\%, respectively). Similarly, households residing in the Centre-Est region seem to experience the most adverse capital losses, as its estimated value for $\alpha$ shown in Table \ref{ResultsandDiscussion-Subsection44-Table3} is the lowest among all regions. Indeed, Centre-Est\rq s (whose head-count index is 34\%) poverty severity index ($FGT_{2}$) is high in comparison to other regions despite the fact that these regions have a higher head-count index (e.g. 36\% for Centre-Nord, 41\% for Centre-Sud and 47\% for Est).

\vspace{0.3cm}

\begin{table}[H]
\centering
\resizebox{\textwidth}{!}{
\begin{tabular}{llccccccccc}
\hline \hline Area of Residence & Type & $\mathrm{\alpha}$ & p-value (KS test) & $R^{2}$ & Poverty Gap Index ($FGT_{1}$) & Poverty Severity Index ($FGT_{2}$) \\
\hline \hline Rural & Direct & - & - & - & 0.118 & 0.040 \\
Rural & MLE & 3.28 & 0.0000 & 0.9518 & 0.112 & 0.042 \\
Rural & MME & 3.08 & 0.0000 & 0.9672 & 0.118 & 0.046 \\
\hline Urban & Direct & - & - & - & 0.052 & 0.016 \\
Urban& MLE & 3.75 & 0.0001 & 0.9468 & 0.050 & 0.017 \\
Urban& MME & 3.55 & 0.0048 & 0.9600 & 0.052 & 0.019 \\
\hline \hline *p-value$>\alpha_{{\scaleto{KS}{3pt}}}=0.05.$
\end{tabular}}
\caption{B1 distribution fitted parameters and poverty measures by area of residence.}
\label{ResultsandDiscussion-Subsection44-Table2}
\end{table}

\vspace{0.3cm}

\begin{table}[H]
\centering
\resizebox{\textwidth}{!}{
\begin{tabular}{llccccccccc}
\hline \hline Region & Type & $\mathrm{\alpha}$  & p-value (KS test) & $R^{2}$ & Poverty Gap Index ($FGT_{1}$) & Poverty Severity Index ($FGT_{2}$)\\
\hline \hline Boucle du Mouhoun & Direct & - & -  & - & 0.143 & 0.052 \\
Boucle du Mouhoun & MLE & 3.08 & 0.0884* & 0.9604 & 0.136 & 0.054 \\
Boucle du Mouhoun & MME & 2.89 & 0.0104 & 0.9739 & 0.143 & 0.058 \\
\hline Cascades & Direct & - & - & - & 0.038 & 0.011 \\
Cascades & MLE & 4.25 & 0.1874* & 0.9617 & 0.037 & 0.012 \\
Cascades & MME & 4.11 & 0.2963* & 0.9663 & 0.038 & 0.013 \\
\hline Centre & Direct & - & -  & - & 0.037 & 0.012 \\
Centre & MLE & 3.76 & 0.3193* & 0.9816 & 0.036 & 0.013 \\ 
Centre & MME & 3.61 & 0.5502* & 0.9872 & 0.037 & 0.013 \\
\hline Centre-Est & Direct & - & - & - & 0.096 & 0.036 \\
Centre-Est & MLE & 2.79 & 0.0594* & 0.9368 & 0.091 & 0.038 \\
Centre-Est & MME & 2.58 & 0.2251* & 0.9546 & 0.096 & 0.042 \\
\hline Centre-Nord & Direct & - & - & - & 0.082 & 0.026 \\
Centre-Nord & MLE & 3.64 & 0.1913* & 0.9552 & 0.078 & 0.028 \\
Centre-Nord & MME & 3.43 & 0.2217* & 0.9679 & 0.082 & 0.030 \\
\hline Centre-Ouest & Direct & - & - & - & 0.107 & 0.034 \\
Centre-Ouest & MLE & 3.61 & 0.0549* & 0.9386 & 0.102 & 0.036 \\
Centre-Ouest & MME & 3.40 & 0.0523* & 0.9531 & 0.107 & 0.040 \\
\hline Centre-Sud & Direct & - & - & - & 0.095 & 0.030 \\
Centre-Sud & MLE & 3.57 & 0.2061* & 0.9483 & 0.090 & 0.032 \\
Centre-Sud & MME & 3.35 & 0.1376* & 0.9636 & 0.095 & 0.035 \\
\hline Est & Direct & - & - & - & 0.109 & 0.035\\
Est & MLE & 3.49 & 0.0700* & 0.9389 & 0.104 & 0.038 \\
Est & MME & 3.26 & 0.0180 & 0.9556 & 0.109 & 0.042 \\
\hline Hauts-Bassins & Direct & - & - & - & 0.076 & 0.025 \\
Hauts-Bassins & MLE & 3.75 & 0.5736* & 0.9865 & 0.074 & 0.026 \\
Hauts-Bassins & MME & 3.62 & 0.6231* & 0.9905 & 0.076 & 0.027 \\
\hline Nord & Direct & - & - & - & 0.176 & 0.063 \\
Nord & MLE & 2.95 & 0.0000 & 0.8879 & 0.165 & 0.067 \\
Nord & MME & 2.69 & 0.0002 & 0.9133 & 0.176 & 0.075 \\
\hline Plateau Central & Direct & - & - & - & 0.104 & 0.034 \\
Plateau Central & MLE & 3.45 & 0.0409 & 0.9330 & 0.098 & 0.036 \\
Plateau Central & MME & 3.22 & 0.0246 & 0.9511 & 0.104 & 0.040 \\
\hline Sahel & Direct & - & - & - & 0.050 & 0.015 \\
Sahel & MLE & 4.23 & 0.6521* & 0.9799 & 0.048 & 0.015 \\
Sahel & MME & 4.07 & 0.6094* & 0.9861 & 0.050 & 0.016 \\
\hline Sud-Ouest & Direct & - & - & - & 0.081 & 0.027 \\
Sud-Ouest & MLE & 3.51 & 0.0128 & 0.9618 & 0.078 & 0.028 \\
Sud-Ouest & MME & 3.33 & 0.0545* & 0.9724 & 0.081 & 0.031 \\
\hline \hline *p-value$>\alpha_{{\scaleto{KS}{3pt}}}=0.05.$
\end{tabular}}
\caption{B1 distribution fitted parameters and poverty measures by region.}
\label{ResultsandDiscussion-Subsection44-Table3}
\end{table}

\vspace{0.3cm}

\begin{figure}[H]
	\begin{subfigure}[b]{0.5\linewidth}
	   % Plot generated with the R code: UninsuredTrappingTimeLaplaceTransform.R
       % We could generate other plots if needed.
  		\includegraphics[width=8.5cm, height=6cm]{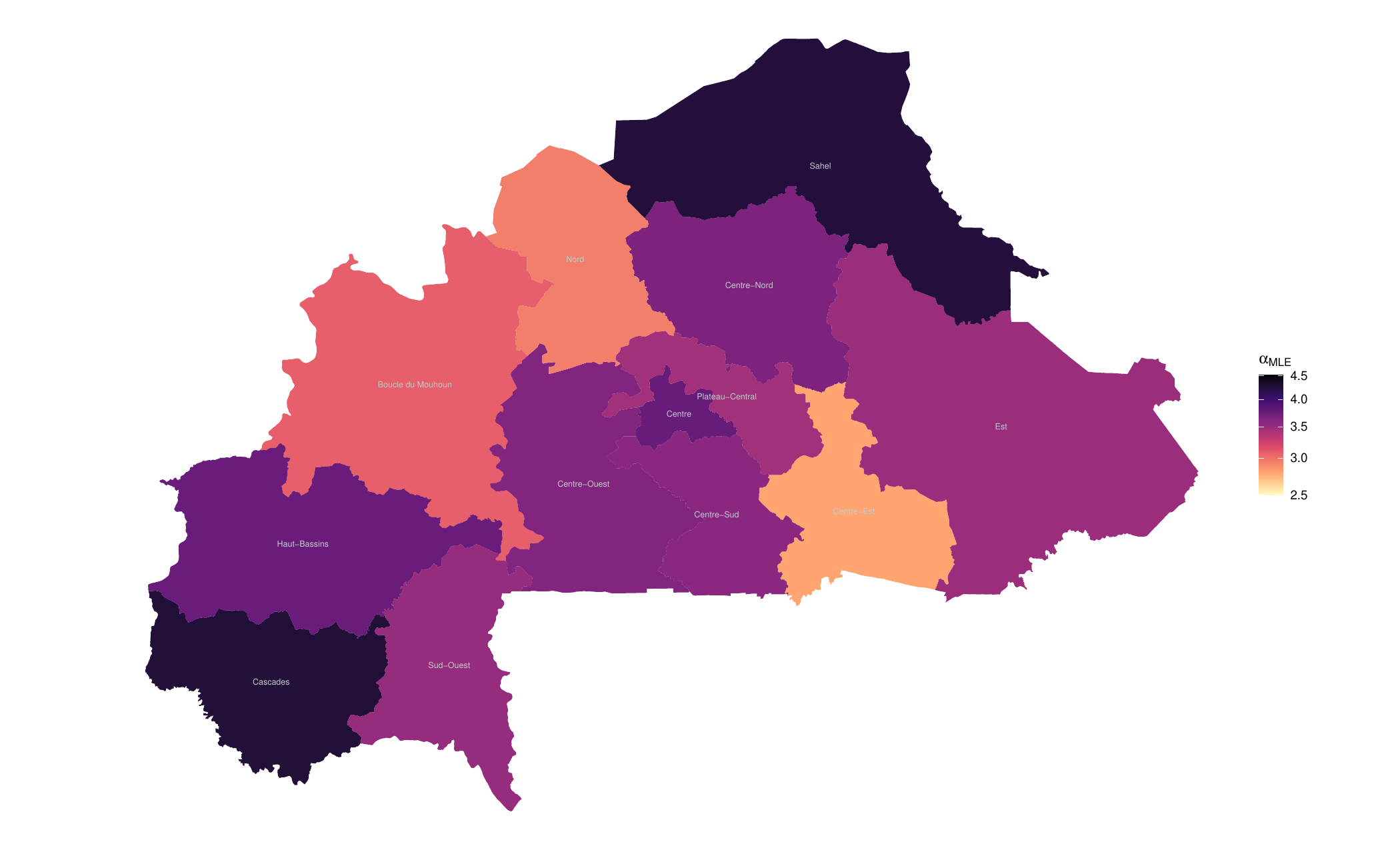}
		\caption{}
  		\label{ResultsandDiscussion-Subsection44-Figure1-a}
	\end{subfigure}
	\begin{subfigure}[b]{0.5\linewidth}
	        % Plot generated with the R code: UninsuredTrappingProbability.R
        % We could generate other plots if needed. Maybe be could look for a seed to fix the obtained results.
  		\includegraphics[width=8.5cm, height=6cm]{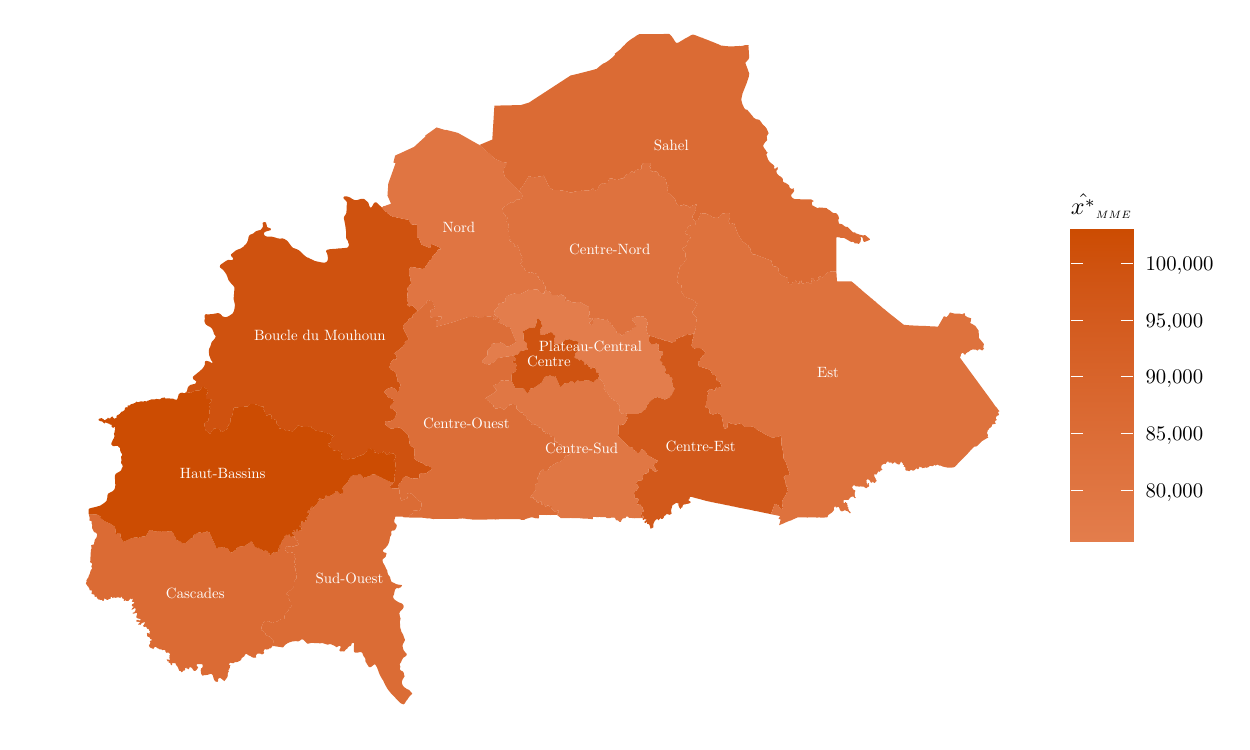}
		\caption{}
  		\label{ResultsandDiscussion-Subsection44-Figure1-b}
	\end{subfigure}
	\caption{Parameter estimators by region: (a) $\hat{\alpha}_{{\scaleto{MLE}{3pt}}}$ (b) $\hat{\alpha}_{{\scaleto{MME}{3pt}}}$.}
	\label{ResultsandDiscussion-Subsection44-Figure1}
\end{figure}

\vspace{0.3cm}

\begin{figure}[H]
\begin{center}\end{center}
	\begin{subfigure}[b]{0.5\linewidth}
	   % Plot generated with the R code: FittingDistributionIncome.R
       	   % We could generate other plots if needed.
  		\includegraphics[width=8cm, height=8cm]{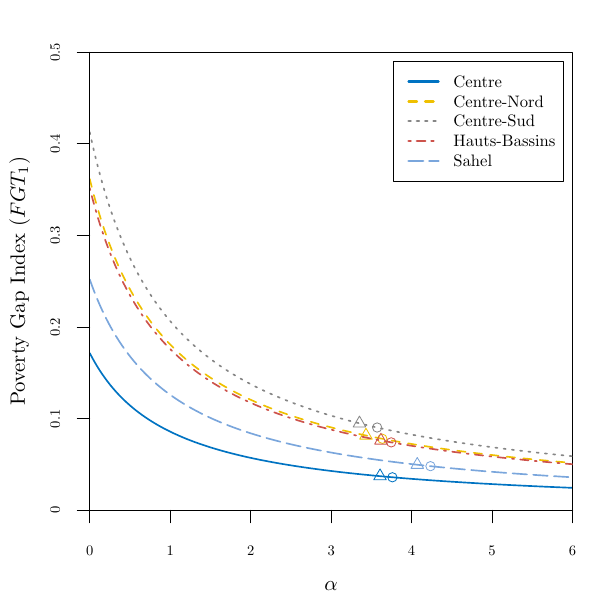}
  		\label{ResultsandDiscussion-Subsection44-Figure2-a}
		%\caption{}
	\end{subfigure}
	\begin{subfigure}[b]{0.5\linewidth}
	        % Plot generated with the R code: FittingDistributionIncome.R
                % We could generate other plots if needed. Maybe be could look for a seed to fix the obtained results.
  		\includegraphics[width=8cm, height=8cm]{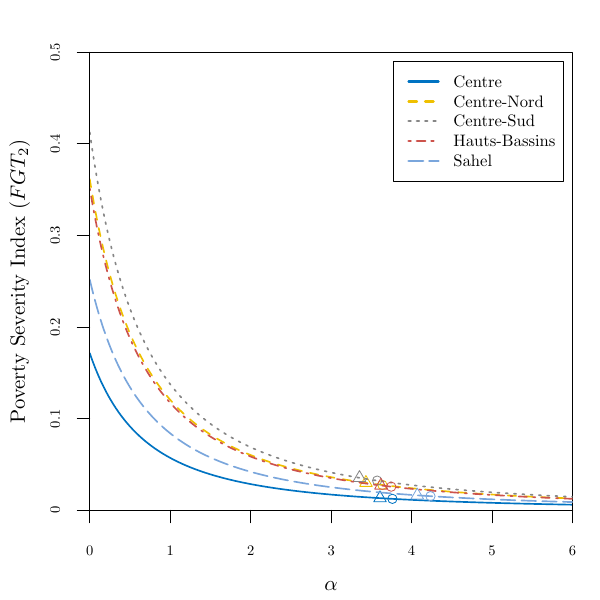}
  		\label{ResultsandDiscussion-Subsection44-Figure2-b}
		%\caption{}
	\end{subfigure}
		\caption{Sensitivity of the poverty gap and the poverty severity index to the parameter $\alpha$ for selected regions. Markers $\bigcirc$ and $\triangle$ display the estimated poverty measures using $\hat{\alpha}_{{\scaleto{MLE}{3pt}}}$ and $\hat{\alpha}_{{\scaleto{MME}{3pt}}}$, respectively.}
	\label{ResultsandDiscussion-Subsection44-Figure2}
\end{figure}

\vspace{0.3cm}

The robustness of the poverty gap index ($FGT_{1}$) and the poverty severity index ($FGT_{2}$) at the national level, by area of residence and by region, when specifying the B1 model as the income short-fall distribution can be evaluated in Tables \ref{ResultsandDiscussion-Subsection44-Table1}, \ref{ResultsandDiscussion-Subsection44-Table2} and \ref{ResultsandDiscussion-Subsection44-Table3}, respectively. Comparing the estimates using the B1 distribution assumption with the direct (empirical) values of the poverty measures, one can see how close the estimates are to the direct values of the FGT indices, thus reinforcing the assumption of the B1 distribution for income short-fall data. 

In Figure \ref{ResultsandDiscussion-Subsection44-Figure2}, we contrast the level of poverty and the changes that would have occurred in the poverty level of selected regions. Both the poverty gap and the poverty severity index show an enormous progress in poverty reduction for greater expected remaining proportion of capital upon experiencing a capital loss (i.e. higher values of $\alpha$). Higher values for $\alpha$ can be attained with risk mitigation strategies such as subsidised insurance programmes \citep{Article:Flores-Contro2021}.

\section{Conclusion} \label{Conclusion-Section5}

This article studies the Gerber-Shiu expected discounted penalty function for the household capital process introduced in \cite{Article:Kovacevic2011}. The Gerber-Shiu function incorporates information on the trapping time, the capital surplus immediately before trapping and the capital deficit at trapping.  Recent work focuses on only analysing the infinite-time trapping probability \citep{Article:Henshaw2023}, therefore overlooking quantities of particular interest such as the undershoot and the overshoot of a household\rq s capital at trapping. To the best of our knowledge, we derive for the first time a functional equation for the Gerber-Shiu function and we solve it for the particular case in which the remaining proportions of capital upon experiencing a capital loss are $Beta(\alpha,1)-$distributed. As a result, we obtain closed-form expressions for important quantities such as the Laplace transform of the trapping time and the distribution of the capital deficit at trapping. These quantities are particularly important as they provide crucial information towards understanding a household’s transition into poverty.

Using risk theory techniques, we derive a microeconomic foundation for the beta of the first kind (B1) as a suitable model to represent the distribution of personal income deficit (or income short-fall). It is indeed interesting that our findings are in line with previous research in development economics, where the generalised beta (GB) distribution family and its derivatives (including the B1 model) have shown to be appropriate models to describe the distribution of personal income.

Affinities between the capital deficit at trapping and a class of poverty measures, known as the Foster-Greer-Thorbecke (FGT) index, are also presented. In addition, we provide empirical evidence of the suitability of the B1 distribution for modelling Burkina Faso\rq s household income short-fall data from the Continuous Multisector Survey (\textit{Enquête Multisectorielle Continue (EMC)}) 2014. Indeed, in this article, the B1 model is fitted to Burkina Faso\rq s household income short-fall data, and it is found that the B1 distribution fitted to the data well, suggesting that this model is a good candidate for describing the income short-fall distribution. Moreover, we show how the poverty gap index and the poverty severity index can be calculated from the estimated B1 income short-fall distribution. One of the main advantages of parametric distributions such as the B1 distribution is that (poverty) indicators can be presented as functions of the parameters of the chosen distribution. Thus parametric modeling allows to gain insight into the relationship between (poverty) indicators and the distribution of the parameters.

Future research can consider other distributions supported in $[0,1]$ for the remaining proportions of capital. In this way, one could arrive at other distributions for the capital deficit at trapping that have also been used previously to model personal income (e.g. the lognormal distribution and the power-law distribution). However, this is not straightforward, as finding a closed-form solution for the Integro-Differential Equation (IDE) derived in Theorem \ref{WhenandHowHouseholdsBecomePoor?-Section2-Theorem1} when considering more general distributions for the remaining proportion of capital is challenging. In addition, it might also be interesting to carry out the same analysis with household microdata from other countries in order to verify the results obtained with Burkina Faso\rq s EMC\footnote{In the supplementary materials, an application to household microdata from Ivory Coast\rq s Harmonised Survey of Household Living Conditions (\textit{Enquête Harmonisée Sur le Conditions de Vie des Ménages (EHCVM)}) 2018 is included.}.

%\section*{{Acknowledgements}}
%The author acknowledges Jorge M. Ramirez Osorio and S\'everine Arnold (-Gaille) for their precious and very useful comments on the paper.

% Appendices come here
\setcitestyle{numbers} % Set the citation style to ``numbers''.
\bibliographystyle{chicago} % Set the Bibliography Style.
\bibliography{main}

\setcitestyle{authoryear}

% Appendices come here

\section*{Appendices}
\addcontentsline{toc}{section}{Appendices}
\renewcommand{\thesubsection}{\Alph{subsection}}
\setcounter{subsection}{0}

\numberwithin{equation}{subsubsection}

\subsection{Goodness-of-Fit Plots for Burkina Faso}\label{AppendixA:GoodnessofFitPlotsforBurkinaFaso}

\begin{figure}[H]
\begin{center}Burkina Faso\end{center}
	\begin{subfigure}[b]{0.5\linewidth}
	   % Plot generated with the R code: FittingDistributionIncome.R
       	   % We could generate other plots if needed.
  		\includegraphics[width=8cm, height=8cm]{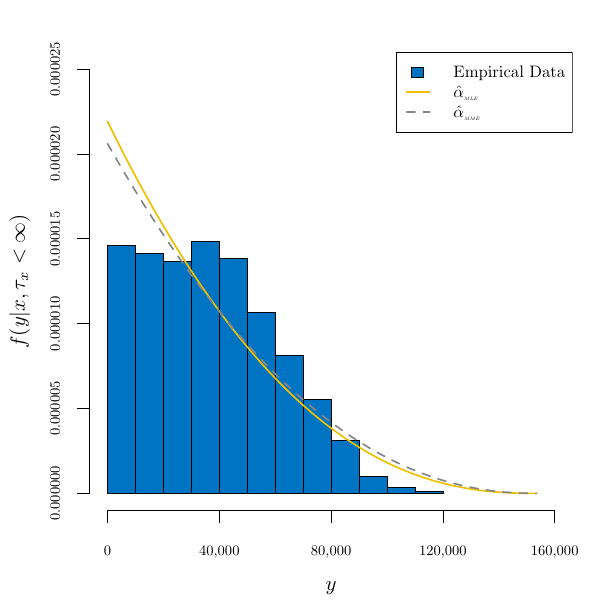}
		\caption{Comparison of Probability Density}
  		\label{AppendixA:GoodnessofFitPlotsforBurkinaFaso-Figure1-a}
	\end{subfigure}
	\begin{subfigure}[b]{0.5\linewidth}
	        % Plot generated with the R code: FittingDistributionIncome.R
                % We could generate other plots if needed. Maybe be could look for a seed to fix the obtained results.
  		\includegraphics[width=8cm, height=8cm]{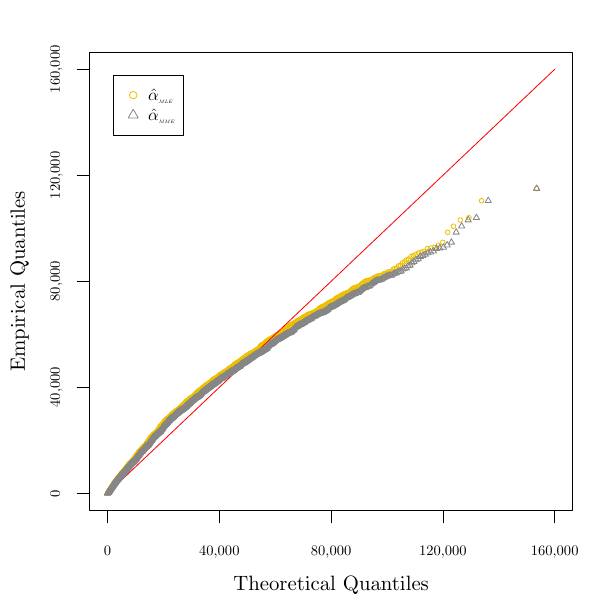}
		\caption{Quantile-Quantile (Q-Q) Plot}
  		\label{AppendixA:GoodnessofFitPlotsforBurkinaFaso-Figure1-b}
	\end{subfigure}
	\begin{subfigure}[b]{0.5\linewidth}
	        % Plot generated with the R code: FittingDistributionIncome.R
                % We could generate other plots if needed. Maybe be could look for a seed to fix the obtained results.
  		\includegraphics[width=8cm, height=8cm]{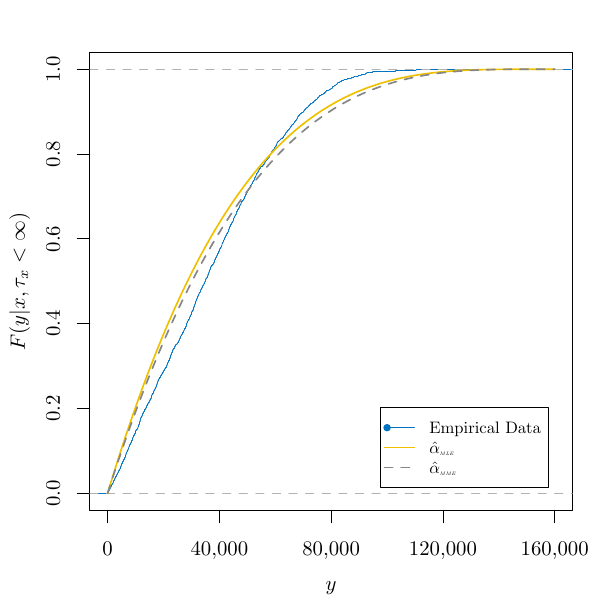}
		\caption{Comparison of Probability Distributions}
  		\label{AppendixA:GoodnessofFitPlotsforBurkinaFaso-Figure1-c}
	\end{subfigure}
	\begin{subfigure}[b]{0.5\linewidth}
	        % Plot generated with the R code: FittingDistributionIncome.R
                % We could generate other plots if needed. Maybe be could look for a seed to fix the obtained results.
  		\includegraphics[width=8cm, height=8cm]{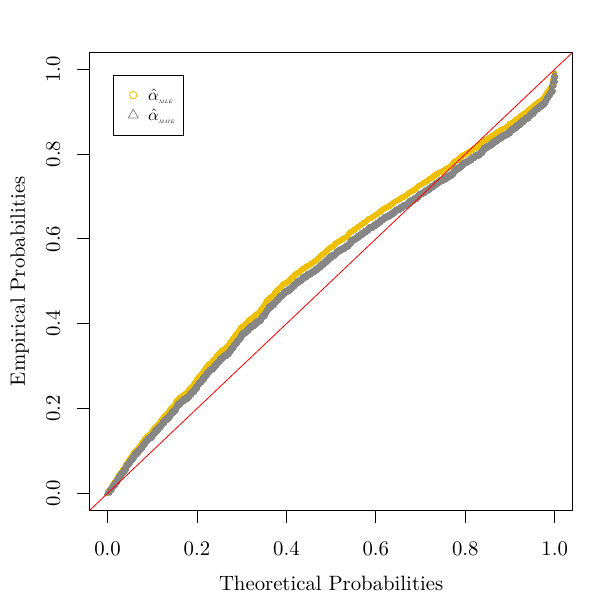}
		\caption{Probability-Probability (P-P) Plot}
  		\label{AppendixA:GoodnessofFitPlotsforBurkinaFaso-Figure1-d}
	\end{subfigure}
	\label{AppendixA:GoodnessofFitPlotsforBurkinaFaso-Figure1}
\end{figure}

\subsection{Goodness-of-Fit Plots by Area of Residence}\label{AppendixB:GoodnessofFitPlotsbyAreaofResidence}

\begin{figure}[H]
\begin{center}Rural\end{center}
	\begin{subfigure}[b]{0.5\linewidth}
	   % Plot generated with the R code: FittingDistributionIncome.R
       	   % We could generate other plots if needed.
  		\includegraphics[width=8cm, height=8cm]{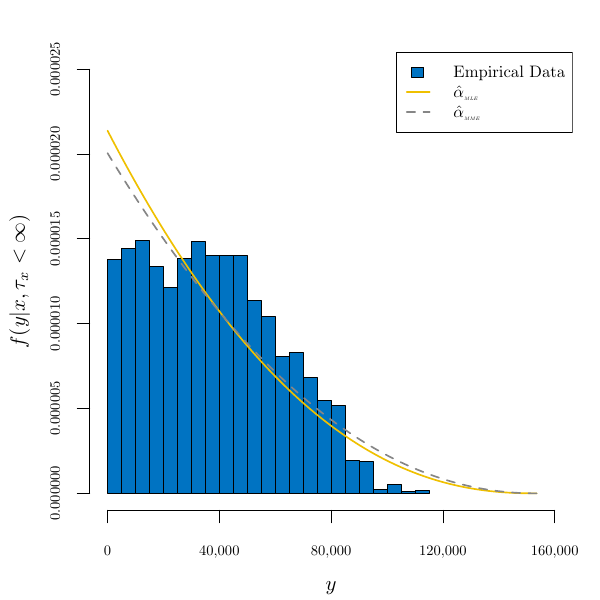}
		\caption{Comparison of Probability Density}
  		\label{AppendixB:GoodnessofFitPlotsbyAreaofResidence-Figure1-a}
	\end{subfigure}
	\begin{subfigure}[b]{0.5\linewidth}
	        % Plot generated with the R code: FittingDistributionIncome.R
                % We could generate other plots if needed. Maybe be could look for a seed to fix the obtained results.
  		\includegraphics[width=8cm, height=8cm]{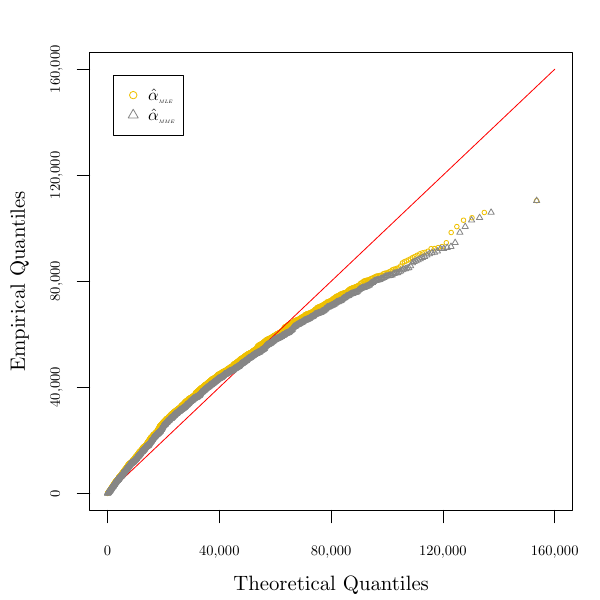}
		\caption{Quantile-Quantile (Q-Q) Plot}
  		\label{AppendixB:GoodnessofFitPlotsbyAreaofResidence-Figure1-b}
	\end{subfigure}
	\begin{subfigure}[b]{0.5\linewidth}
	        % Plot generated with the R code: FittingDistributionIncome.R
                % We could generate other plots if needed. Maybe be could look for a seed to fix the obtained results.
  		\includegraphics[width=8cm, height=8cm]{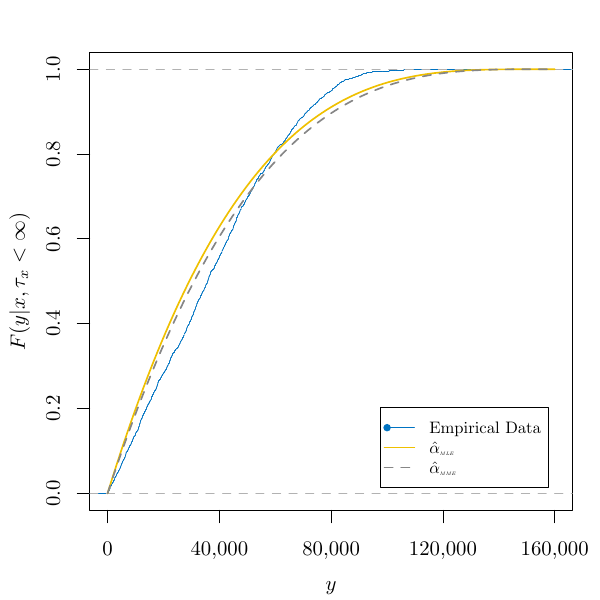}
		\caption{Comparison of Probability Distributions}
  		\label{AppendixB:GoodnessofFitPlotsbyAreaofResidence-Figure1-c}
	\end{subfigure}
	\begin{subfigure}[b]{0.5\linewidth}
	        % Plot generated with the R code: FittingDistributionIncome.R
                % We could generate other plots if needed. Maybe be could look for a seed to fix the obtained results.
  		\includegraphics[width=8cm, height=8cm]{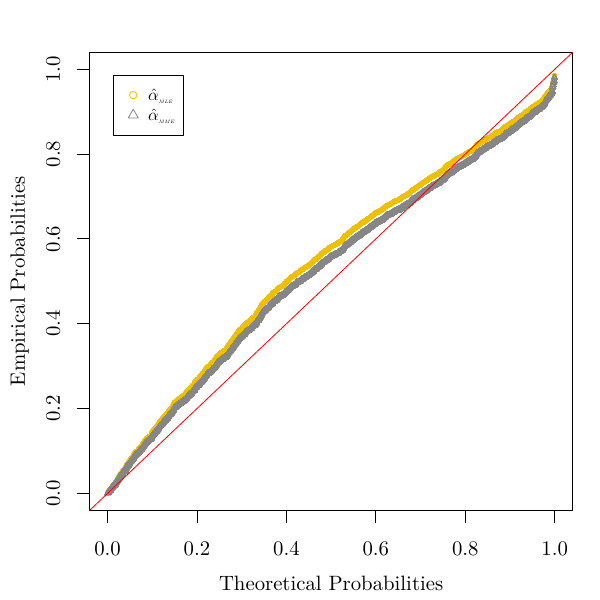}
		\caption{Probability-Probability (P-P) Plot}
  		\label{AppendixB:GoodnessofFitPlotsbyAreaofResidence-Figure1-d}
	\end{subfigure}
	\label{AppendixB:GoodnessofFitPlotsbyAreaofResidence-Figure1}
\end{figure}

\begin{figure}[H]
\begin{center}Urban\end{center}
	\begin{subfigure}[b]{0.5\linewidth}
	   % Plot generated with the R code: FittingDistributionIncome.R
       	   % We could generate other plots if needed.
  		\includegraphics[width=8cm, height=8cm]{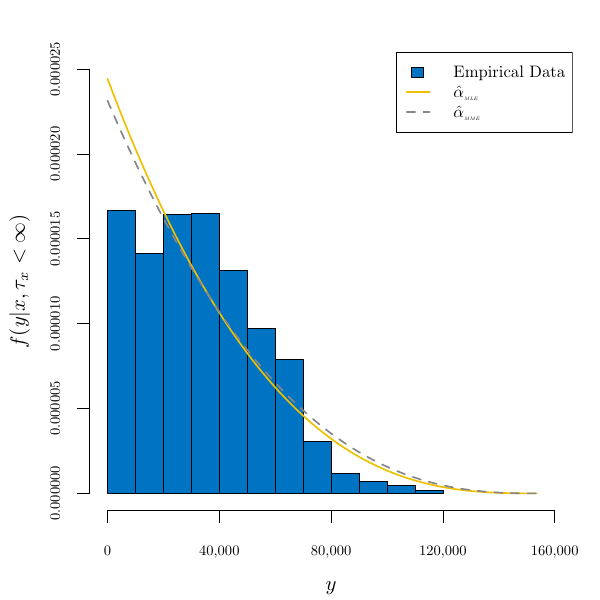}
		\caption{Comparison of Probability Density}
  		\label{AppendixB:GoodnessofFitPlotsbyAreaofResidence-Figure2-a}
	\end{subfigure}
	\begin{subfigure}[b]{0.5\linewidth}
	        % Plot generated with the R code: FittingDistributionIncome.R
                % We could generate other plots if needed. Maybe be could look for a seed to fix the obtained results.
  		\includegraphics[width=8cm, height=8cm]{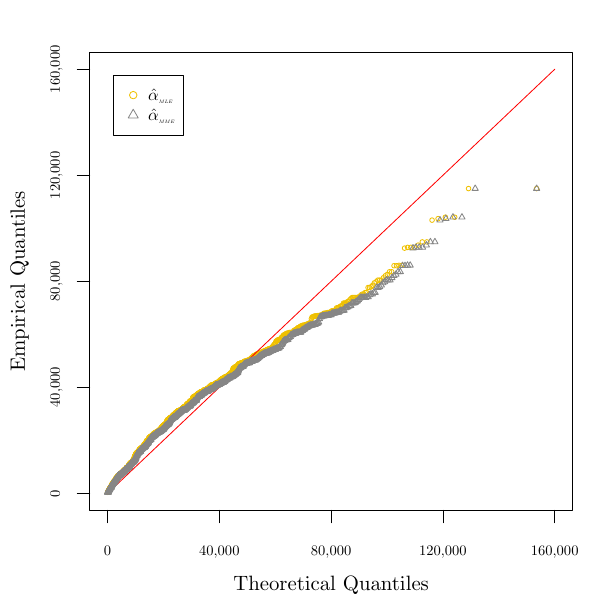}
		\caption{Quantile-Quantile (Q-Q) Plot}
  		\label{AppendixB:GoodnessofFitPlotsbyAreaofResidence-Figure2-b}
	\end{subfigure}
	\begin{subfigure}[b]{0.5\linewidth}
	        % Plot generated with the R code: FittingDistributionIncome.R
                % We could generate other plots if needed. Maybe be could look for a seed to fix the obtained results.
  		\includegraphics[width=8cm, height=8cm]{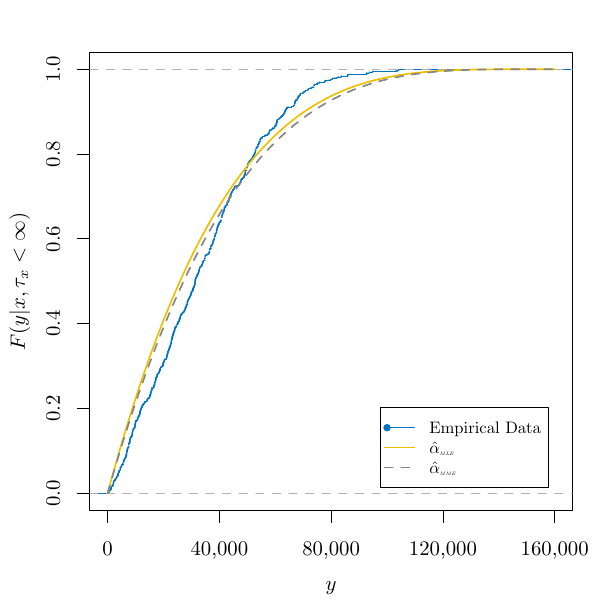}
		\caption{Comparison of Probability Distributions}
  		\label{AppendixB:GoodnessofFitPlotsbyAreaofResidence-Figure2-c}
	\end{subfigure}
	\begin{subfigure}[b]{0.5\linewidth}
	        % Plot generated with the R code: FittingDistributionIncome.R
                % We could generate other plots if needed. Maybe be could look for a seed to fix the obtained results.
  		\includegraphics[width=8cm, height=8cm]{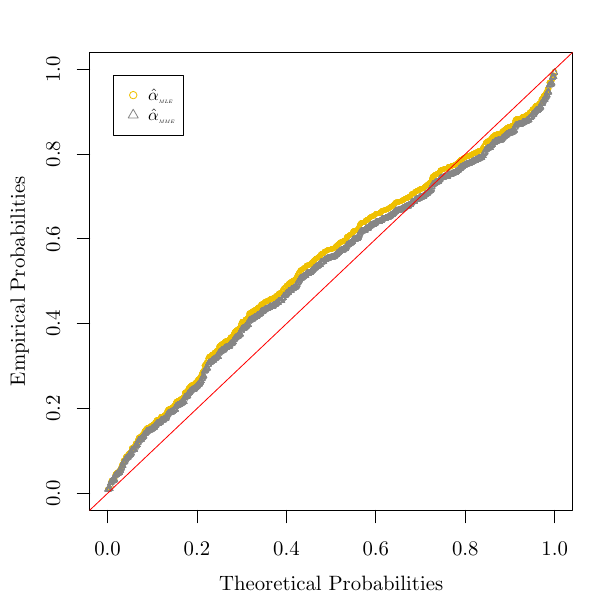}
		\caption{Probability-Probability (P-P) Plot}
  		\label{AppendixB:GoodnessofFitPlotsbyAreaofResidence-Figure2-d}
	\end{subfigure}
	\label{AppendixB:GoodnessofFitPlotsbyAreaofResidence-Figure2}
\end{figure}

\subsection{Goodness-of-Fit Plots by Region}\label{AppendixC:GoodnessofFitPlotsbyRegion}

\begin{figure}[H]
\begin{center}Boucle du Mouhoun\end{center}
	\begin{subfigure}[b]{0.5\linewidth}
	   % Plot generated with the R code: FittingDistributionIncomeRegionAdministrative.R
       	   % We could generate other plots if needed.
  		\includegraphics[width=8cm, height=8cm]{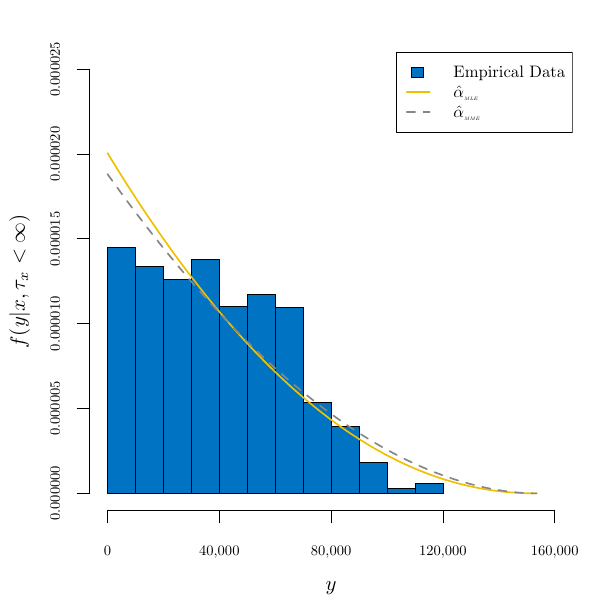}
		 \caption{Comparison of Probability Density}
  		\label{AppendixC:GoodnessofFitPlotsbyRegion-Figure1-a}
	\end{subfigure}
	\begin{subfigure}[b]{0.5\linewidth}
	        % Plot generated with the R code: FittingDistributionIncomeRegionAdministrative.R
                % We could generate other plots if needed. Maybe be could look for a seed to fix the obtained results.
  		\includegraphics[width=8cm, height=8cm]{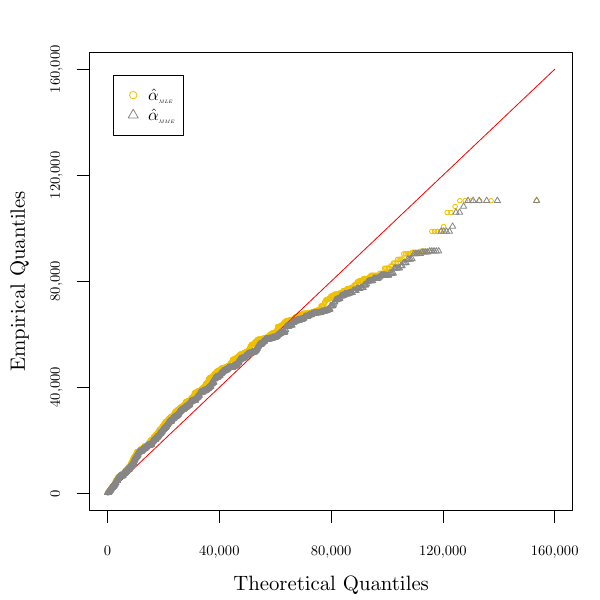}
		\caption{Quantile-Quantile (Q-Q) Plot}
  		\label{AppendixC:GoodnessofFitPlotsbyRegion-Figure1-b}
	\end{subfigure}
	\begin{subfigure}[b]{0.5\linewidth}
	        % Plot generated with the R code: FittingDistributionIncomeRegionAdministrative.R
                % We could generate other plots if needed. Maybe be could look for a seed to fix the obtained results.
  		\includegraphics[width=8cm, height=8cm]{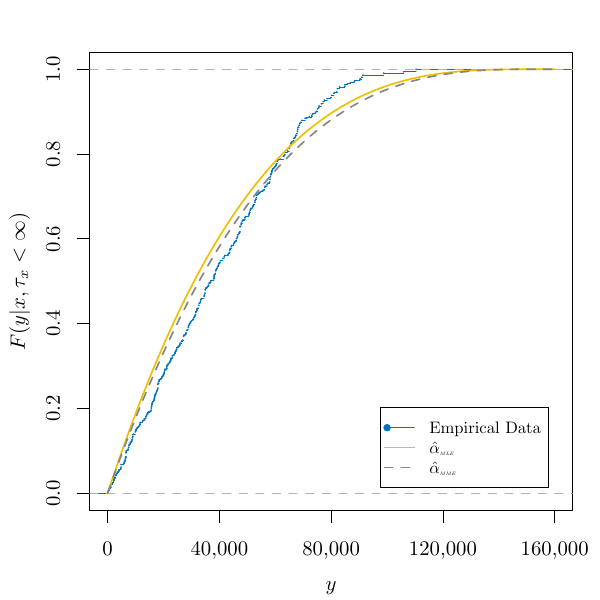}
		\caption{Comparison of Probability Distributions}
  		\label{AppendixC:GoodnessofFitPlotsbyRegion-Figure1-c}
	\end{subfigure}
	\begin{subfigure}[b]{0.5\linewidth}
	        % Plot generated with the R code: FittingDistributionIncomeRegionAdministrative.R
                % We could generate other plots if needed. Maybe be could look for a seed to fix the obtained results.
  		\includegraphics[width=8cm, height=8cm]{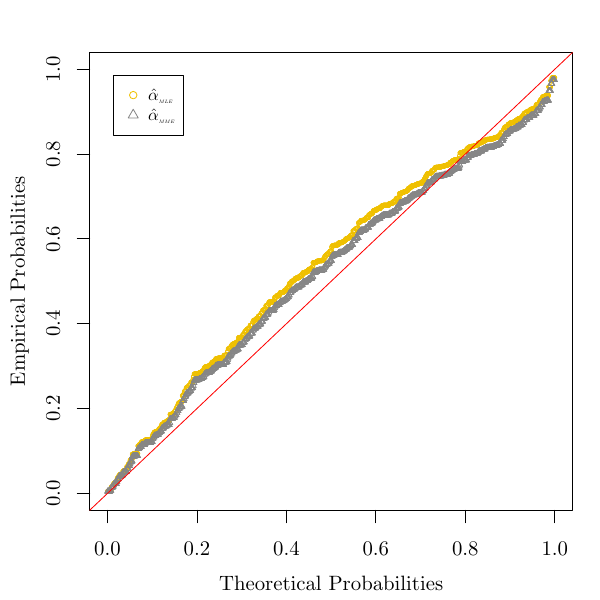}
		\caption{Probability-Probability (P-P) Plot}
  		\label{AppendixC:GoodnessofFitPlotsbyRegion-Figure1-d}
	\end{subfigure}
	\label{AppendixC:GoodnessofFitPlotsbyRegion-Figure1}
\end{figure}

\begin{figure}[H]
\begin{center}Cascades\end{center}
	\begin{subfigure}[b]{0.5\linewidth}
	   % Plot generated with the R code: FittingDistributionIncomeRegionAdministrative.R
       	   % We could generate other plots if needed.
  		\includegraphics[width=8cm, height=8cm]{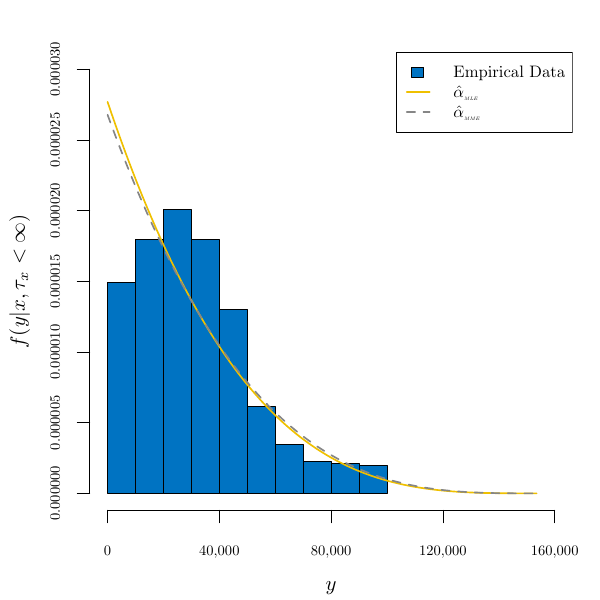}
		\caption{Comparison of Probability Density}
  		\label{AppendixC:GoodnessofFitPlotsbyRegion-Figure2-a}
	\end{subfigure}
	\begin{subfigure}[b]{0.5\linewidth}
	        % Plot generated with the R code: FittingDistributionIncomeRegionAdministrative.R
                % We could generate other plots if needed. Maybe be could look for a seed to fix the obtained results.
  		\includegraphics[width=8cm, height=8cm]{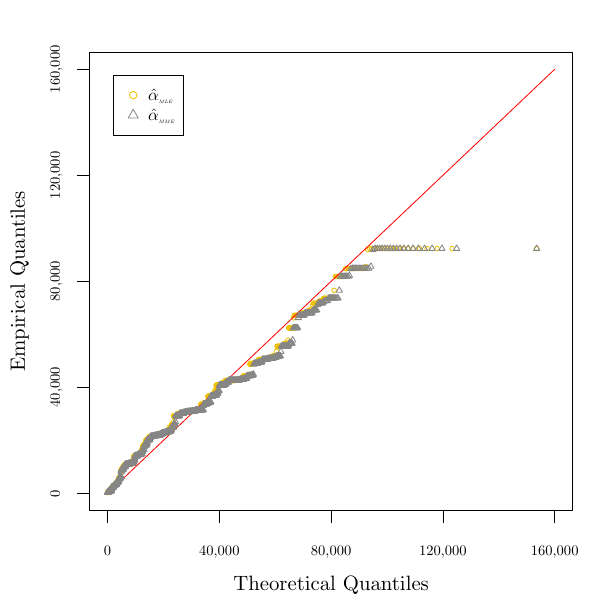}
		\caption{Quantile-Quantile (Q-Q) Plot}
  		\label{AppendixC:GoodnessofFitPlotsbyRegion-Figure2-b}
	\end{subfigure}
	\begin{subfigure}[b]{0.5\linewidth}
	        % Plot generated with the R code: FittingDistributionIncomeRegionAdministrative.R
                % We could generate other plots if needed. Maybe be could look for a seed to fix the obtained results.
  		\includegraphics[width=8cm, height=8cm]{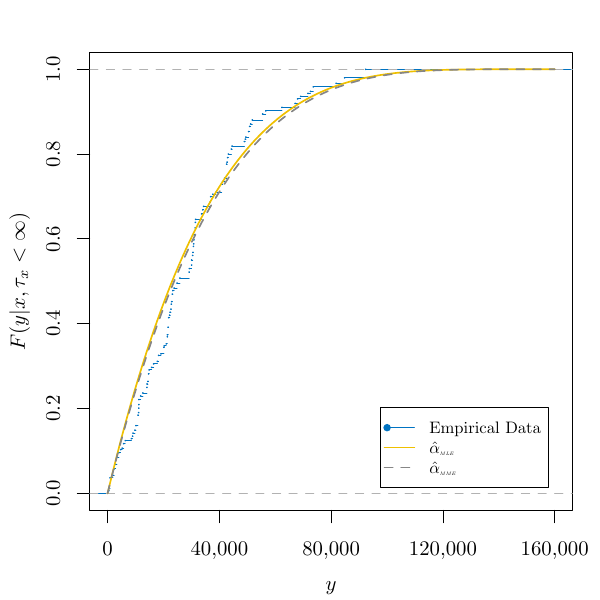}
		\caption{Comparison of Probability Distributions}
  		\label{AppendixC:GoodnessofFitPlotsbyRegion-Figure2-c}
	\end{subfigure}
	\begin{subfigure}[b]{0.5\linewidth}
	        % Plot generated with the R code: FittingDistributionIncomeRegionAdministrative.R
                % We could generate other plots if needed. Maybe be could look for a seed to fix the obtained results.
  		\includegraphics[width=8cm, height=8cm]{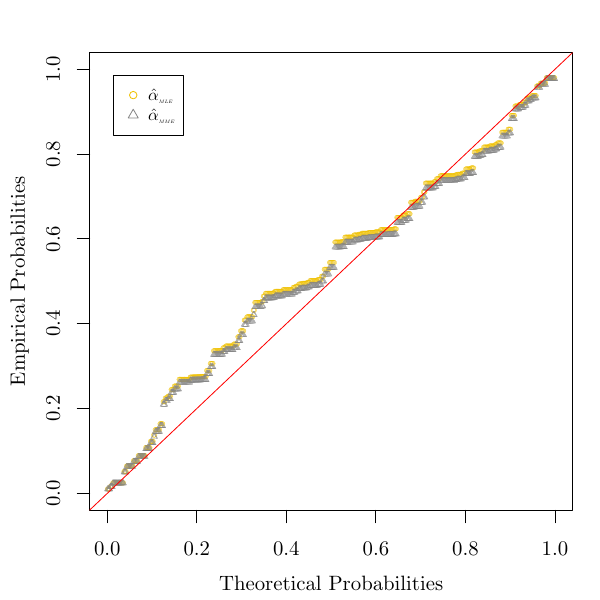}
		\caption{Probability-Probability (P-P) Plot}
  		\label{AppendixC:GoodnessofFitPlotsbyRegion-Figure2-d}
	\end{subfigure}
	\label{AppendixC:GoodnessofFitPlotsbyRegion-Figure2}
\end{figure}

\begin{figure}[H]
\begin{center}Centre\end{center}
	\begin{subfigure}[b]{0.5\linewidth}
	   % Plot generated with the R code: FittingDistributionIncomeRegionAdministrative.R
       	   % We could generate other plots if needed.
  		\includegraphics[width=8cm, height=8cm]{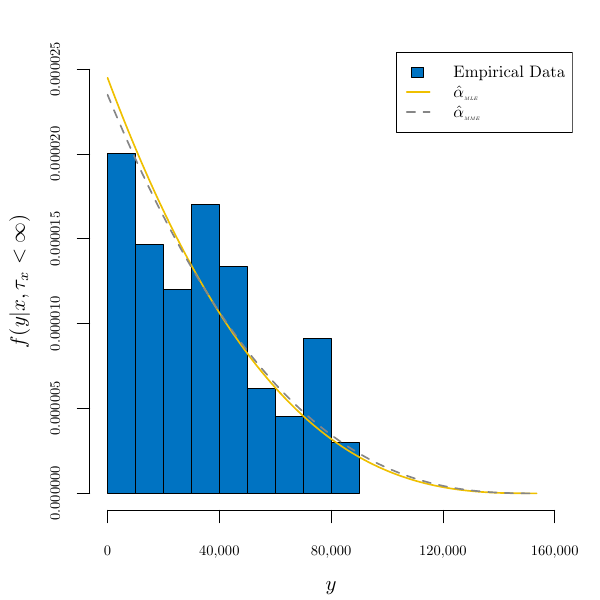}
		\caption{Comparison of Probability Density}
  		\label{AppendixC:GoodnessofFitPlotsbyRegion-Figure3-a}
	\end{subfigure}
	\begin{subfigure}[b]{0.5\linewidth}
	        % Plot generated with the R code: FittingDistributionIncomeRegionAdministrative.R
                % We could generate other plots if needed. Maybe be could look for a seed to fix the obtained results.
  		\includegraphics[width=8cm, height=8cm]{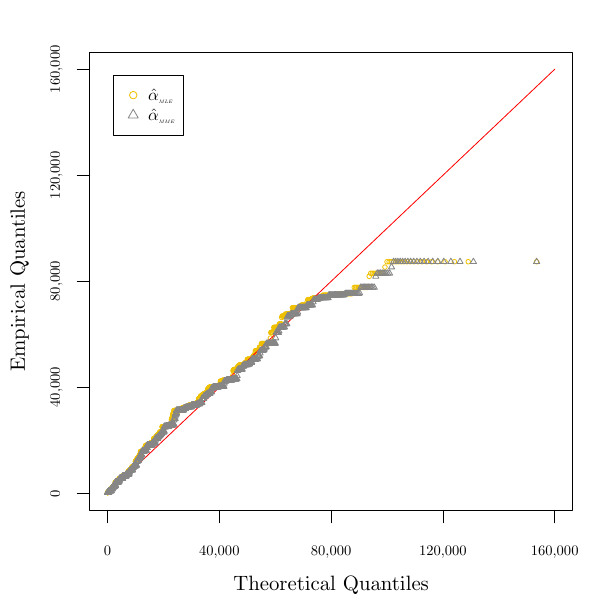}
		\caption{Quantile-Quantile (Q-Q) Plot}
  		\label{AppendixC:GoodnessofFitPlotsbyRegion-Figure3-b}
	\end{subfigure}
	\begin{subfigure}[b]{0.5\linewidth}
	        % Plot generated with the R code: FittingDistributionIncomeRegionAdministrative.R
                % We could generate other plots if needed. Maybe be could look for a seed to fix the obtained results.
  		\includegraphics[width=8cm, height=8cm]{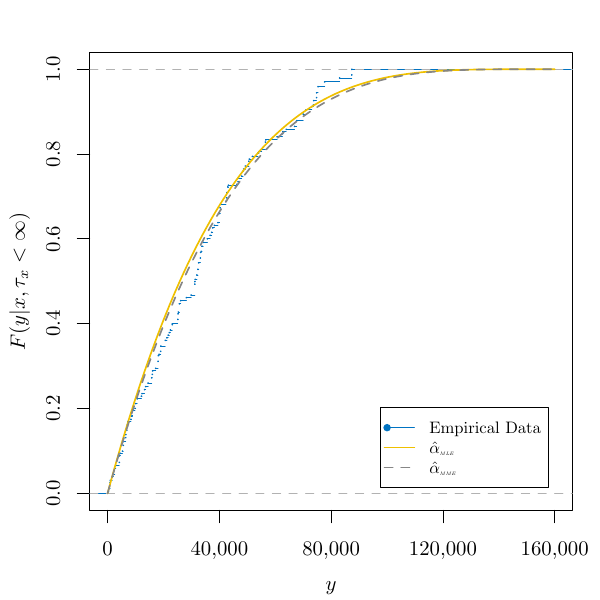}
		\caption{Comparison of Probability Distributions}
  		\label{AppendixC:GoodnessofFitPlotsbyRegion-Figure3-c}
	\end{subfigure}
	\begin{subfigure}[b]{0.5\linewidth}
	        % Plot generated with the R code: FittingDistributionIncomeRegionAdministrative.R
                % We could generate other plots if needed. Maybe be could look for a seed to fix the obtained results.
  		\includegraphics[width=8cm, height=8cm]{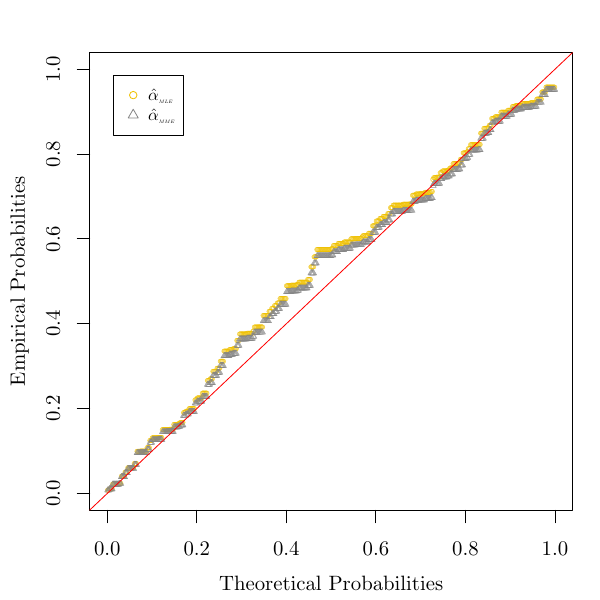}
		\caption{Probability-Probability (P-P) Plot}
  		\label{AppendixC:GoodnessofFitPlotsbyRegion-Figure3-d}
	\end{subfigure}
	\label{AppendixC:GoodnessofFitPlotsbyRegion-Figure3}
\end{figure}

\begin{figure}[H]
\begin{center}Centre-Est\end{center}
	\begin{subfigure}[b]{0.5\linewidth}
	   % Plot generated with the R code: FittingDistributionIncomeRegionAdministrative.R
       	   % We could generate other plots if needed.
  		\includegraphics[width=8cm, height=8cm]{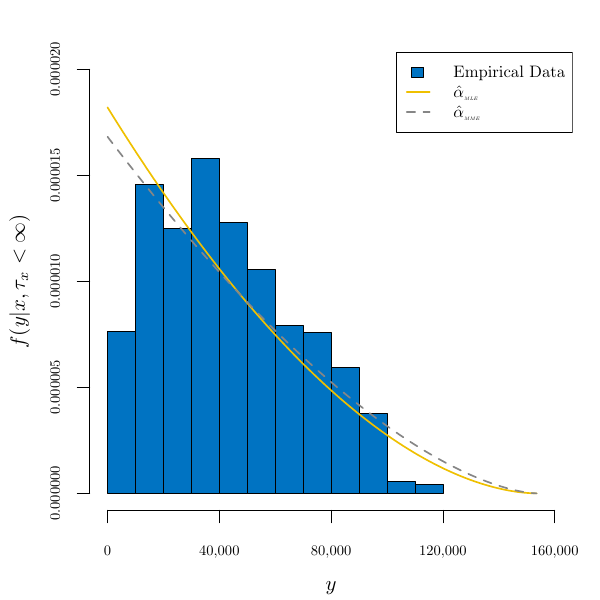}
		\caption{Comparison of Probability Density}
  		\label{AppendixC:GoodnessofFitPlotsbyRegion-Figure4-a}
	\end{subfigure}
	\begin{subfigure}[b]{0.5\linewidth}
	        % Plot generated with the R code: FittingDistributionIncomeRegionAdministrative.R
                % We could generate other plots if needed. Maybe be could look for a seed to fix the obtained results.
  		\includegraphics[width=8cm, height=8cm]{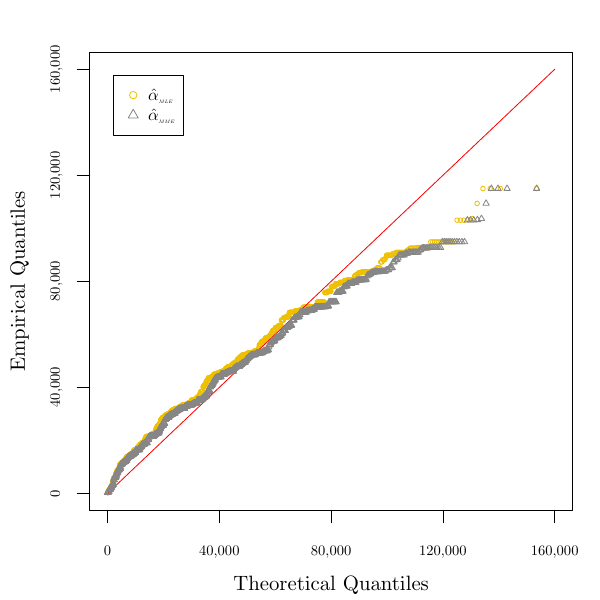}
		\caption{Quantile-Quantile (Q-Q) Plot}
  		\label{AppendixC:GoodnessofFitPlotsbyRegion-Figure4-b}
	\end{subfigure}
	\begin{subfigure}[b]{0.5\linewidth}
	        % Plot generated with the R code: FittingDistributionIncomeRegionAdministrative.R
                % We could generate other plots if needed. Maybe be could look for a seed to fix the obtained results.
  		\includegraphics[width=8cm, height=8cm]{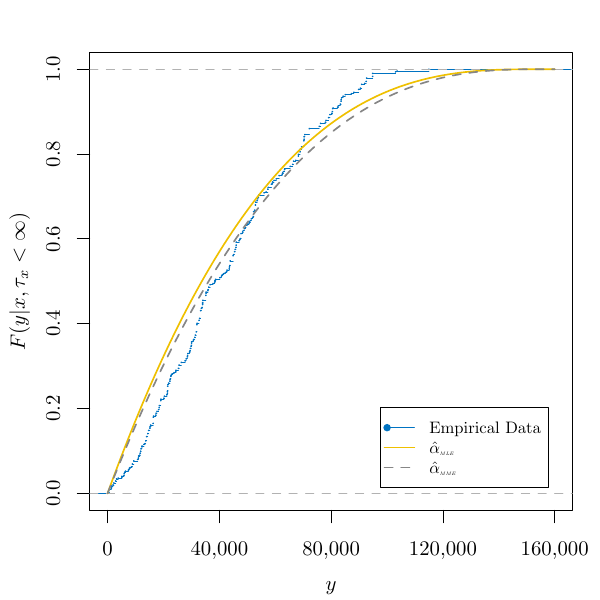}
		\caption{Comparison of Probability Distributions}
  		\label{AppendixC:GoodnessofFitPlotsbyRegion-Figure4-c}
	\end{subfigure}
	\begin{subfigure}[b]{0.5\linewidth}
	        % Plot generated with the R code: FittingDistributionIncomeRegionAdministrative.R
                % We could generate other plots if needed. Maybe be could look for a seed to fix the obtained results.
  		\includegraphics[width=8cm, height=8cm]{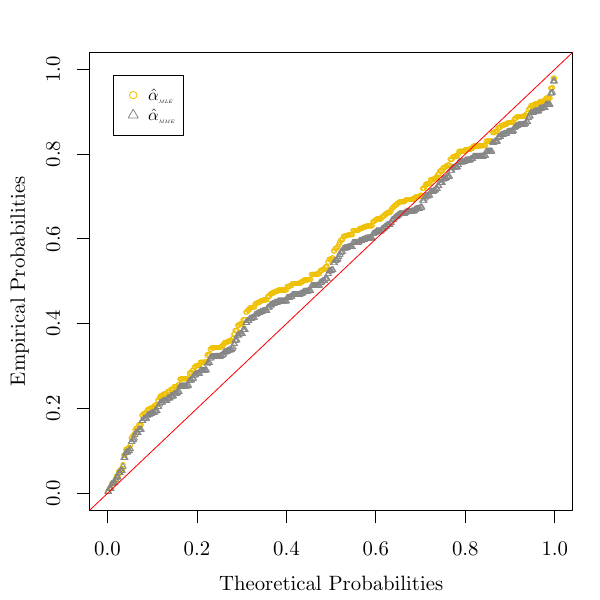}
		\caption{Probability-Probability (P-P) Plot}
  		\label{AppendixC:GoodnessofFitPlotsbyRegion-Figure4-d}
	\end{subfigure}
	\label{AppendixC:GoodnessofFitPlotsbyRegion-Figure4}
\end{figure}

\begin{figure}[H]
\begin{center}Centre-Nord\end{center}
	\begin{subfigure}[b]{0.5\linewidth}
	   % Plot generated with the R code: FittingDistributionIncomeRegionAdministrative.R
       	   % We could generate other plots if needed.
  		\includegraphics[width=8cm, height=8cm]{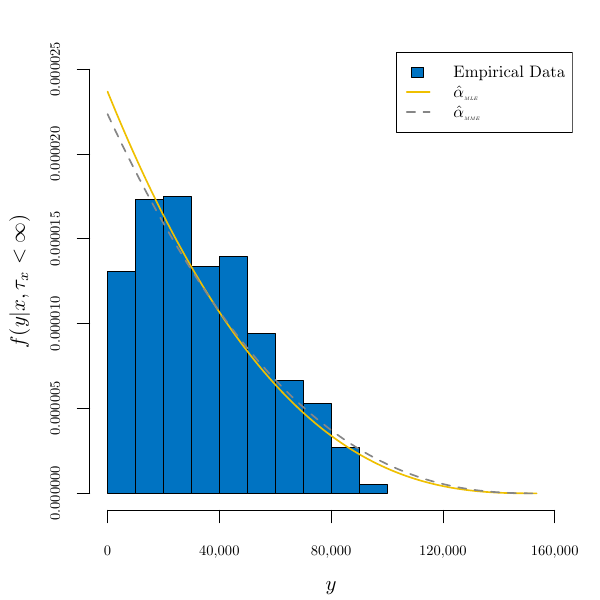}
		\caption{Comparison of Probability Density}
  		\label{AppendixC:GoodnessofFitPlotsbyRegion-Figure5-a}
	\end{subfigure}
	\begin{subfigure}[b]{0.5\linewidth}
	        % Plot generated with the R code: FittingDistributionIncomeRegionAdministrative.R
                % We could generate other plots if needed. Maybe be could look for a seed to fix the obtained results.
  		\includegraphics[width=8cm, height=8cm]{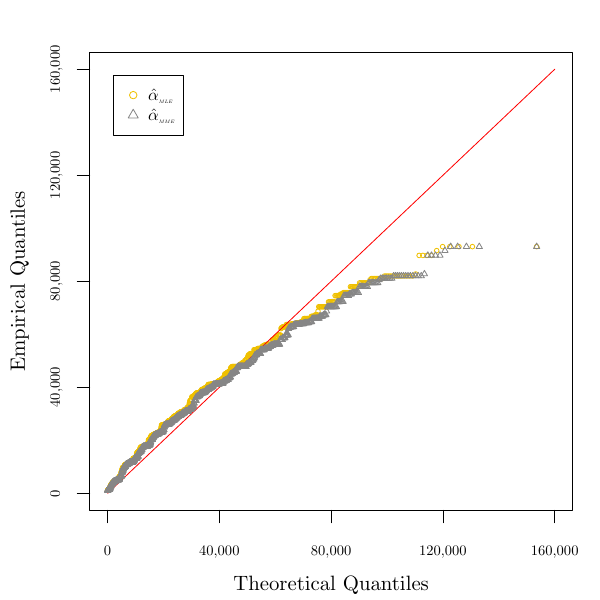}
		\caption{Quantile-Quantile (Q-Q) Plot}
  		\label{AppendixC:GoodnessofFitPlotsbyRegion-Figure5-b}
	\end{subfigure}
	\begin{subfigure}[b]{0.5\linewidth}
	        % Plot generated with the R code: FittingDistributionIncomeRegionAdministrative.R
                % We could generate other plots if needed. Maybe be could look for a seed to fix the obtained results.
  		\includegraphics[width=8cm, height=8cm]{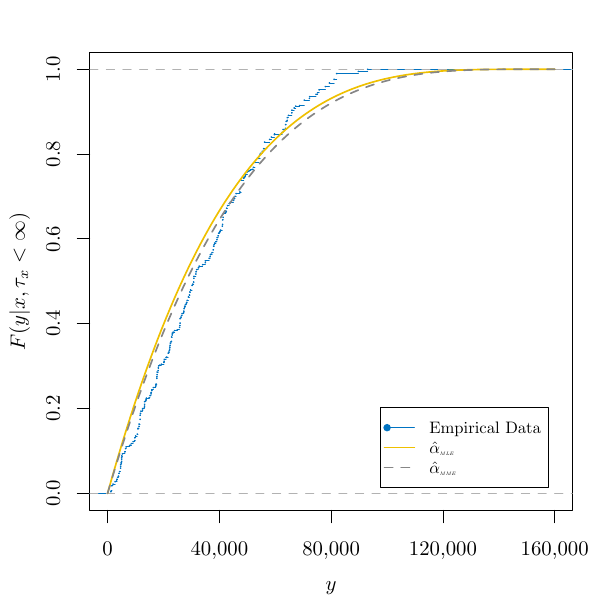}
		\caption{Comparison of Probability Distributions}
  		\label{AppendixC:GoodnessofFitPlotsbyRegion-Figure5-c}
	\end{subfigure}
	\begin{subfigure}[b]{0.5\linewidth}
	        % Plot generated with the R code: FittingDistributionIncomeRegionAdministrative.R
                % We could generate other plots if needed. Maybe be could look for a seed to fix the obtained results.
  		\includegraphics[width=8cm, height=8cm]{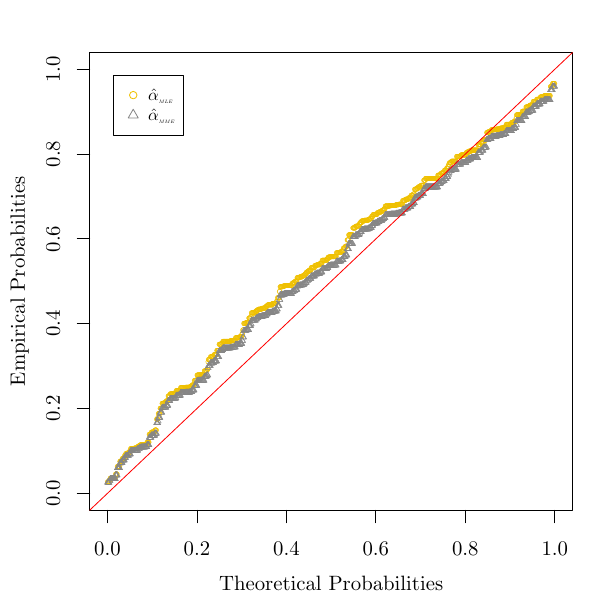}
		\caption{Probability-Probability (P-P) Plot}
  		\label{AppendixC:GoodnessofFitPlotsbyRegion-Figure5-d}
	\end{subfigure}
	\label{AppendixC:GoodnessofFitPlotsbyRegion-Figure5}
\end{figure}

\begin{figure}[H]
\begin{center}Centre-Ouest\end{center}
	\begin{subfigure}[b]{0.5\linewidth}
	   % Plot generated with the R code: FittingDistributionIncomeRegionAdministrative.R
       	   % We could generate other plots if needed.
  		\includegraphics[width=8cm, height=8cm]{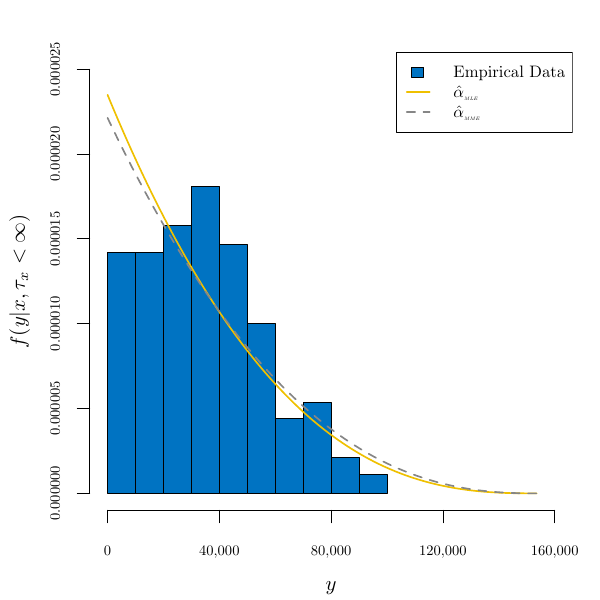}
		\caption{Comparison of Probability Density}
  		\label{AppendixC:GoodnessofFitPlotsbyRegion-Figure6-a}
	\end{subfigure}
	\begin{subfigure}[b]{0.5\linewidth}
	        % Plot generated with the R code: FittingDistributionIncomeRegionAdministrative.R
                % We could generate other plots if needed. Maybe be could look for a seed to fix the obtained results.
  		\includegraphics[width=8cm, height=8cm]{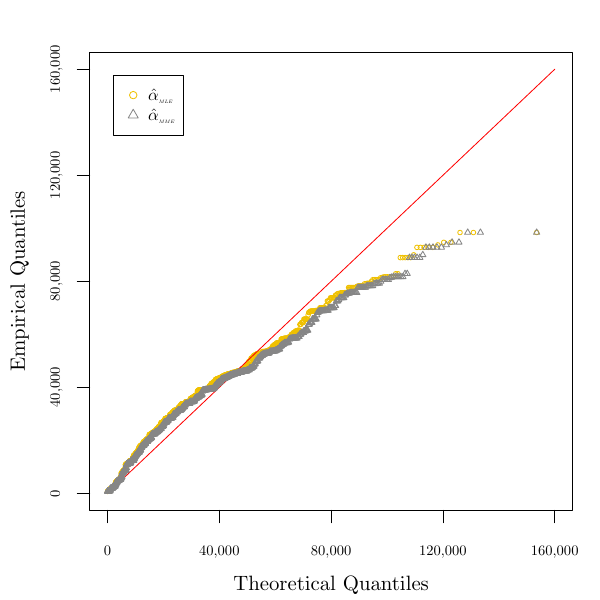}
		\caption{Quantile-Quantile (Q-Q) Plot}
  		\label{AppendixC:GoodnessofFitPlotsbyRegion-Figure6-b}
	\end{subfigure}
	\begin{subfigure}[b]{0.5\linewidth}
	        % Plot generated with the R code: FittingDistributionIncomeRegionAdministrative.R
                % We could generate other plots if needed. Maybe be could look for a seed to fix the obtained results.
  		\includegraphics[width=8cm, height=8cm]{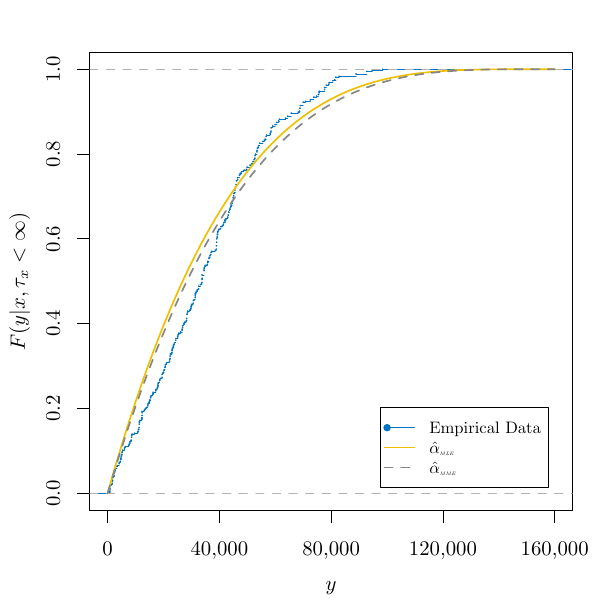}
		\caption{Comparison of Probability Distributions}
  		\label{AppendixC:GoodnessofFitPlotsbyRegion-Figure6-c}
	\end{subfigure}
	\begin{subfigure}[b]{0.5\linewidth}
	        % Plot generated with the R code: FittingDistributionIncomeRegionAdministrative.R
                % We could generate other plots if needed. Maybe be could look for a seed to fix the obtained results.
  		\includegraphics[width=8cm, height=8cm]{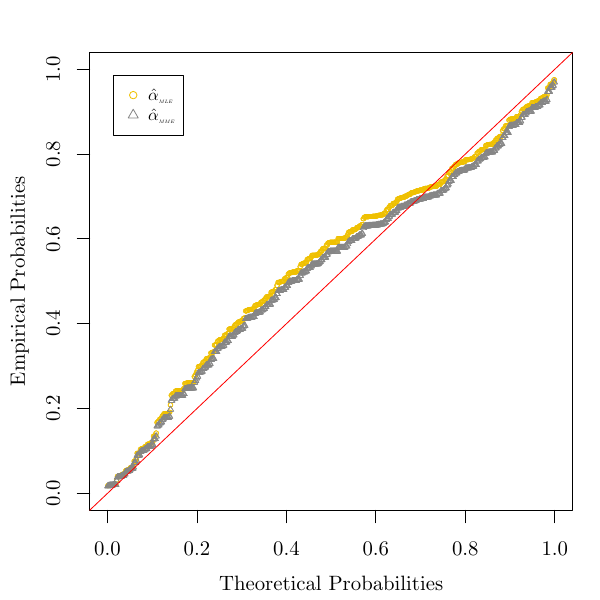}
		\caption{Probability-Probability (P-P) Plot}
  		\label{AppendixC:GoodnessofFitPlotsbyRegion-Figure6-d}
	\end{subfigure}
	\label{AppendixC:GoodnessofFitPlotsbyRegion-Figure6}
\end{figure}

\begin{figure}[H]
\begin{center}Centre-Sud\end{center}
	\begin{subfigure}[b]{0.5\linewidth}
	   % Plot generated with the R code: FittingDistributionIncomeRegionAdministrative.R
       	   % We could generate other plots if needed.
  		\includegraphics[width=8cm, height=8cm]{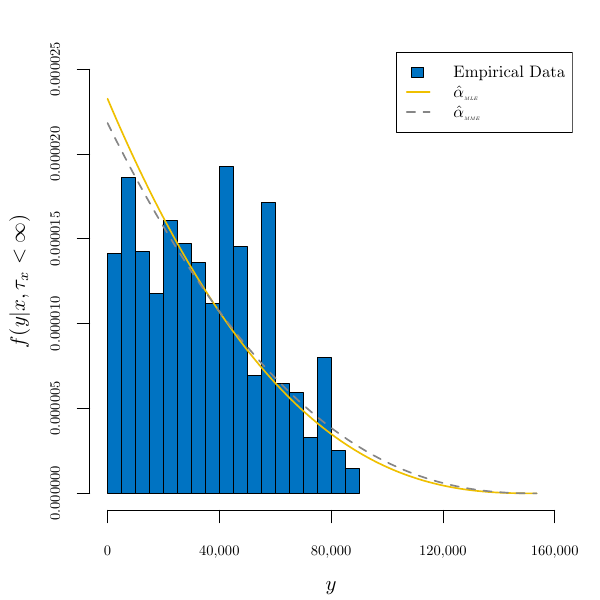}
		\caption{Comparison of Probability Density}
  		\label{AppendixC:GoodnessofFitPlotsbyRegion-Figure7-a}
	\end{subfigure}
	\begin{subfigure}[b]{0.5\linewidth}
	        % Plot generated with the R code: FittingDistributionIncomeRegionAdministrative.R
                % We could generate other plots if needed. Maybe be could look for a seed to fix the obtained results.
  		\includegraphics[width=8cm, height=8cm]{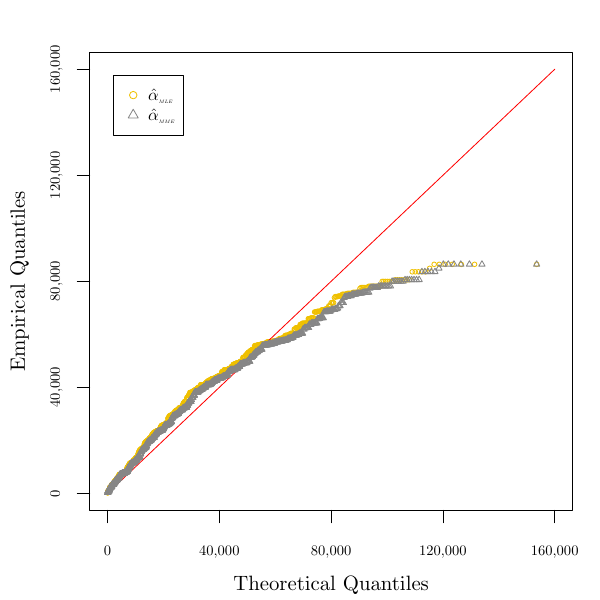}
		\caption{Quantile-Quantile (Q-Q) Plot}
  		\label{AppendixC:GoodnessofFitPlotsbyRegion-Figure7-b}
	\end{subfigure}
	\begin{subfigure}[b]{0.5\linewidth}
	        % Plot generated with the R code: FittingDistributionIncomeRegionAdministrative.R
                % We could generate other plots if needed. Maybe be could look for a seed to fix the obtained results.
  		\includegraphics[width=8cm, height=8cm]{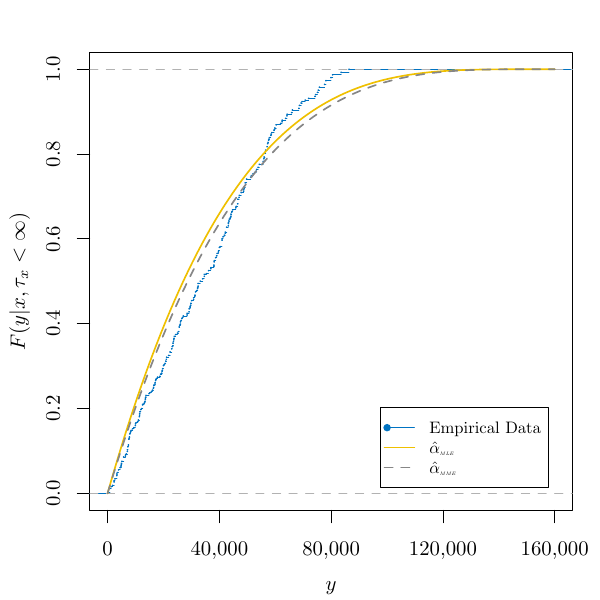}
		\caption{Comparison of Probability Distributions}
  		\label{AppendixC:GoodnessofFitPlotsbyRegion-Figure7-c}
	\end{subfigure}
	\begin{subfigure}[b]{0.5\linewidth}
	        % Plot generated with the R code: FittingDistributionIncomeRegionAdministrative.R
                % We could generate other plots if needed. Maybe be could look for a seed to fix the obtained results.
  		\includegraphics[width=8cm, height=8cm]{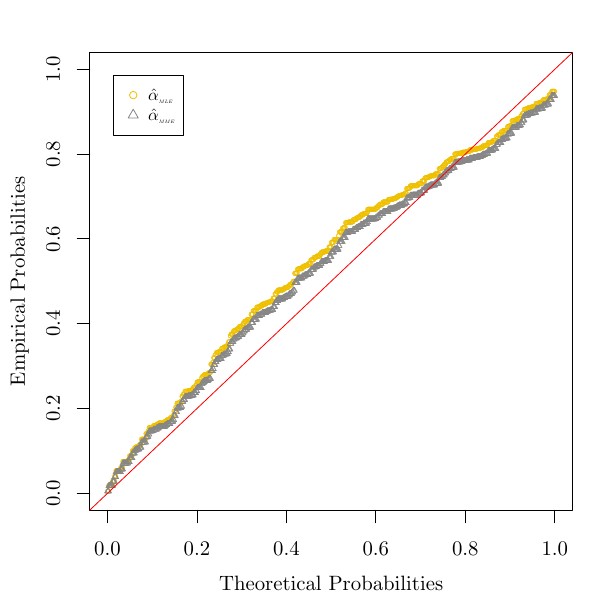}
		\caption{Probability-Probability (P-P) Plot}
  		\label{AppendixC:GoodnessofFitPlotsbyRegion-Figure7-d}
	\end{subfigure}
	\label{AppendixC:GoodnessofFitPlotsbyRegion-Figure7}
\end{figure}

\begin{figure}[H]
\begin{center}Est\end{center}
	\begin{subfigure}[b]{0.5\linewidth}
	   % Plot generated with the R code: FittingDistributionIncomeRegionAdministrative.R
       	   % We could generate other plots if needed.
  		\includegraphics[width=8cm, height=8cm]{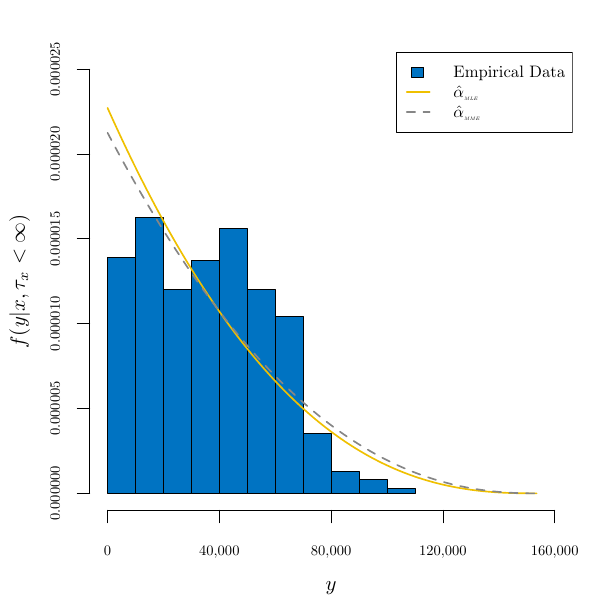}
		\caption{Comparison of Probability Density}
  		\label{AppendixC:GoodnessofFitPlotsbyRegion-Figure8-a}
	\end{subfigure}
	\begin{subfigure}[b]{0.5\linewidth}
	        % Plot generated with the R code: FittingDistributionIncomeRegionAdministrative.R
                % We could generate other plots if needed. Maybe be could look for a seed to fix the obtained results.
  		\includegraphics[width=8cm, height=8cm]{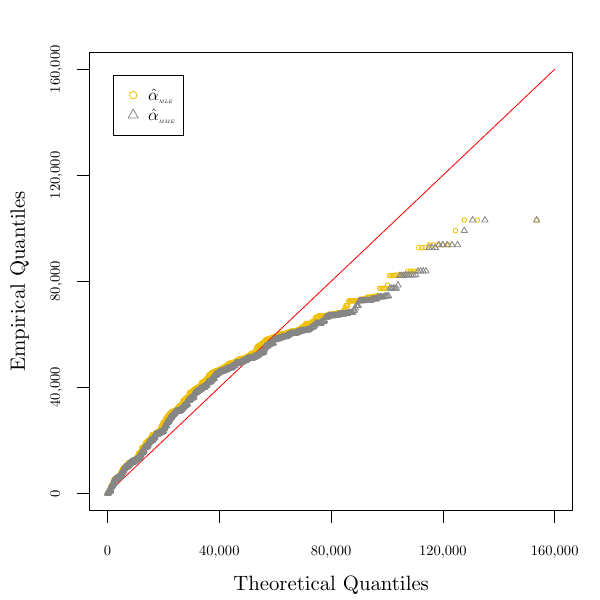}
		\caption{Quantile-Quantile (Q-Q) Plot}
  		\label{AppendixC:GoodnessofFitPlotsbyRegion-Figure8-b}
	\end{subfigure}
	\begin{subfigure}[b]{0.5\linewidth}
	        % Plot generated with the R code: FittingDistributionIncomeRegionAdministrative.R
                % We could generate other plots if needed. Maybe be could look for a seed to fix the obtained results.
  		\includegraphics[width=8cm, height=8cm]{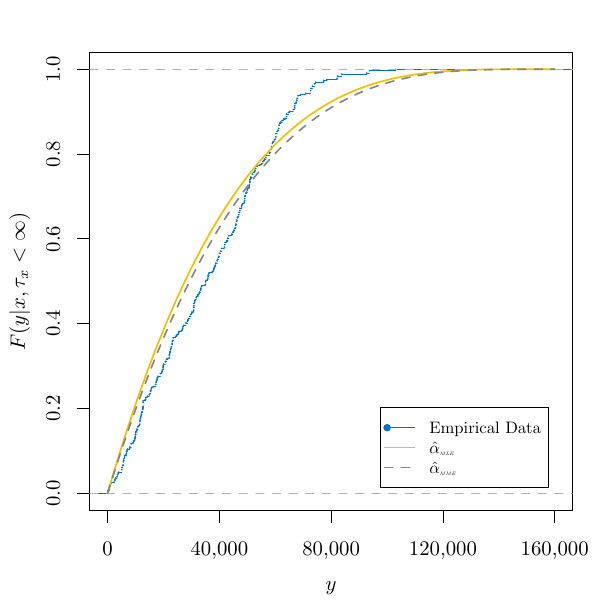}
		\caption{Comparison of Probability Distributions}
  		\label{AppendixC:GoodnessofFitPlotsbyRegion-Figure8-c}
	\end{subfigure}
	\begin{subfigure}[b]{0.5\linewidth}
	        % Plot generated with the R code: FittingDistributionIncomeRegionAdministrative.R
                % We could generate other plots if needed. Maybe be could look for a seed to fix the obtained results.
  		\includegraphics[width=8cm, height=8cm]{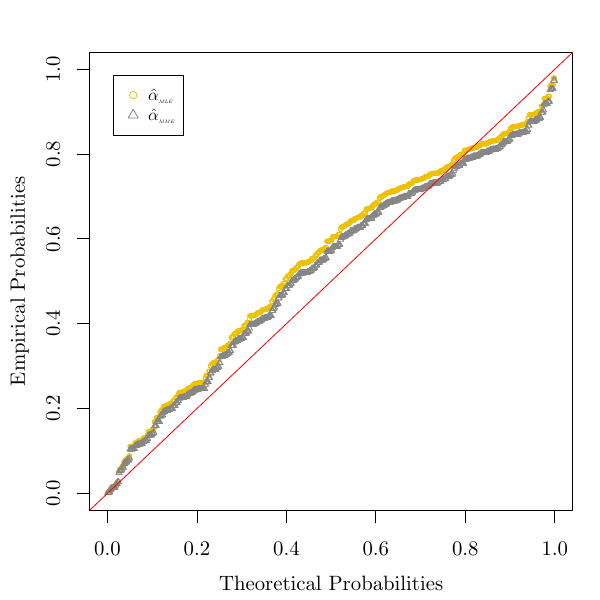}
		\caption{Probability-Probability (P-P) Plot}
  		\label{AppendixC:GoodnessofFitPlotsbyRegion-Figure8-d}
	\end{subfigure}
	\label{AppendixC:GoodnessofFitPlotsbyRegion-Figure8}
\end{figure}

\begin{figure}[H]
\begin{center}Hauts-Bassins\end{center}
	\begin{subfigure}[b]{0.5\linewidth}
	   % Plot generated with the R code: FittingDistributionIncomeRegionAdministrative.R
       	   % We could generate other plots if needed.
  		\includegraphics[width=8cm, height=8cm]{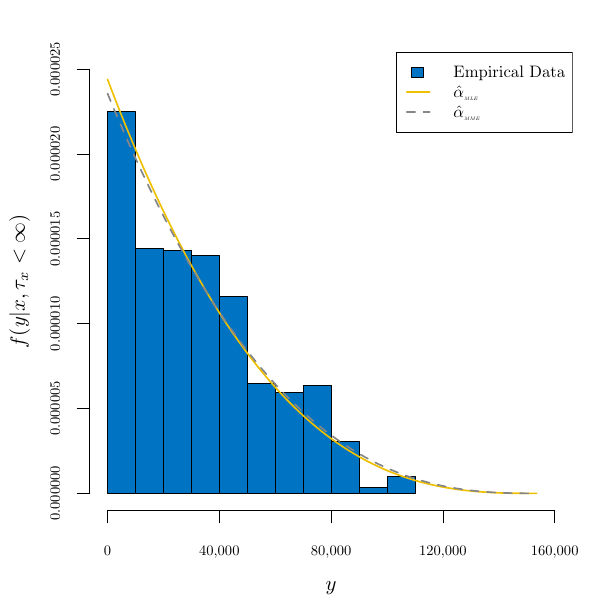}
		\caption{Comparison of Probability Density}
  		\label{AppendixC:GoodnessofFitPlotsbyRegion-Figure9-a}
	\end{subfigure}
	\begin{subfigure}[b]{0.5\linewidth}
	        % Plot generated with the R code: FittingDistributionIncomeRegionAdministrative.R
                % We could generate other plots if needed. Maybe be could look for a seed to fix the obtained results.
  		\includegraphics[width=8cm, height=8cm]{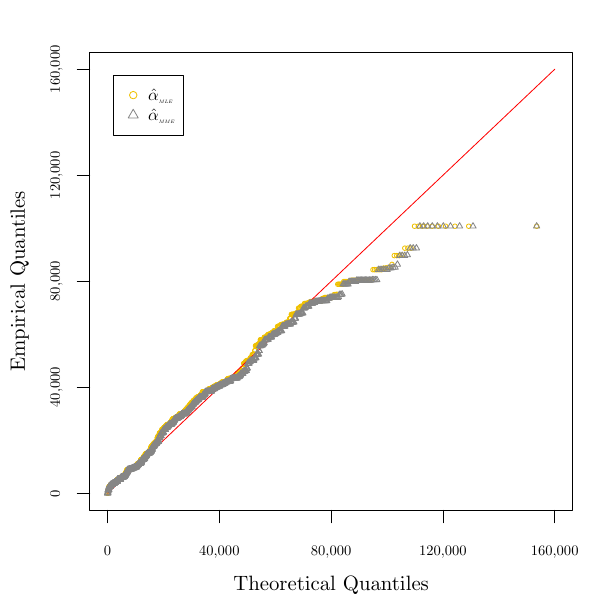}
		\caption{Quantile-Quantile (Q-Q) Plot}
  		\label{AppendixC:GoodnessofFitPlotsbyRegion-Figure9-b}
	\end{subfigure}
	\begin{subfigure}[b]{0.5\linewidth}
	        % Plot generated with the R code: FittingDistributionIncomeRegionAdministrative.R
                % We could generate other plots if needed. Maybe be could look for a seed to fix the obtained results.
  		\includegraphics[width=8cm, height=8cm]{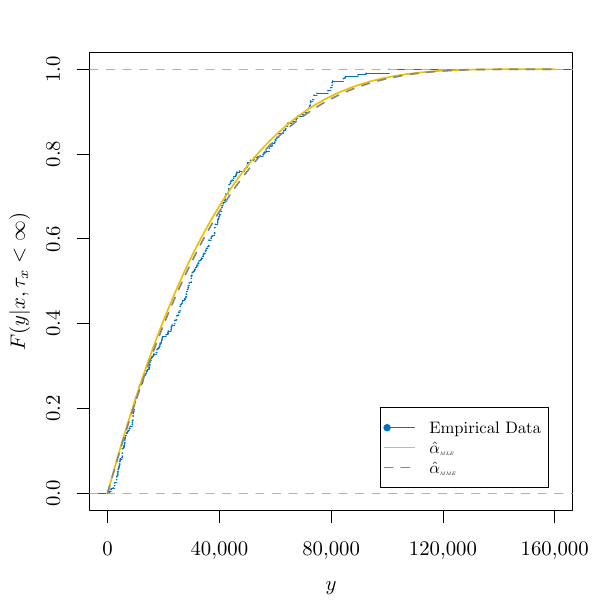}
		\caption{Comparison of Probability Distributions}
  		\label{AppendixC:GoodnessofFitPlotsbyRegion-Figure9-c}
	\end{subfigure}
	\begin{subfigure}[b]{0.5\linewidth}
	        % Plot generated with the R code: FittingDistributionIncomeRegionAdministrative.R
                % We could generate other plots if needed. Maybe be could look for a seed to fix the obtained results.
  		\includegraphics[width=8cm, height=8cm]{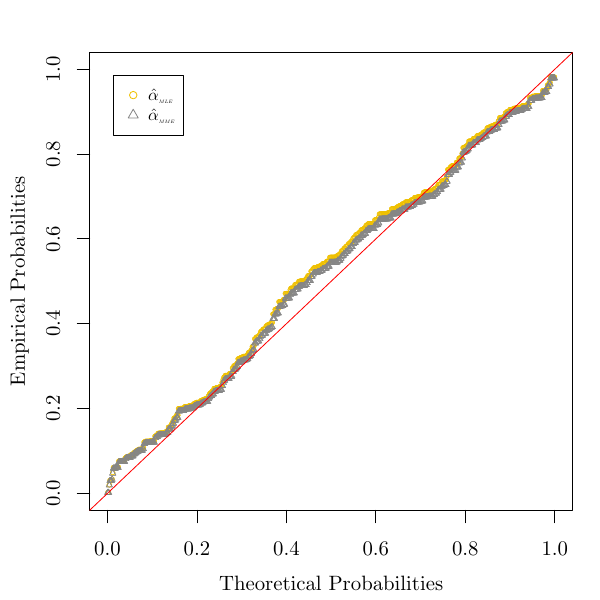}
		\caption{Probability-Probability (P-P) Plot}
  		\label{AppendixC:GoodnessofFitPlotsbyRegion-Figure9-d}
	\end{subfigure}
	\label{AppendixC:GoodnessofFitPlotsbyRegion-Figure9}
\end{figure}

\begin{figure}[H]
\begin{center}Nord\end{center}
	\begin{subfigure}[b]{0.5\linewidth}
	   % Plot generated with the R code: FittingDistributionIncomeRegionAdministrative.R
       	   % We could generate other plots if needed.
  		\includegraphics[width=8cm, height=8cm]{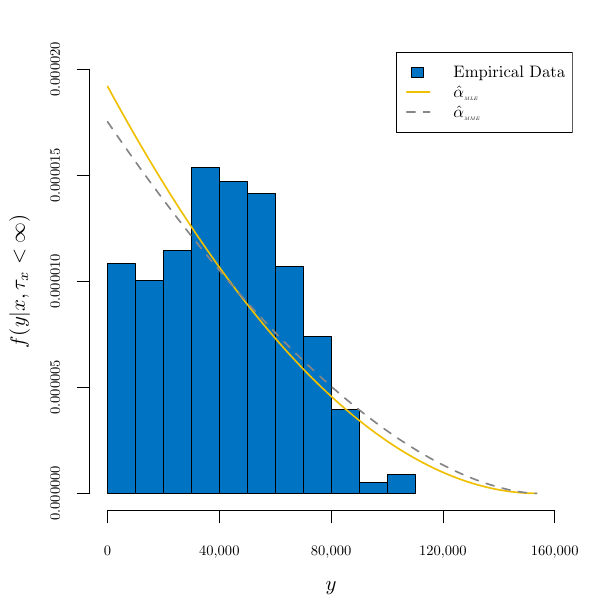}
		\caption{Comparison of Probability Density}
  		\label{AppendixC:GoodnessofFitPlotsbyRegion-Figure10-a}
	\end{subfigure}
	\begin{subfigure}[b]{0.5\linewidth}
	        % Plot generated with the R code: FittingDistributionIncomeRegionAdministrative.R
                % We could generate other plots if needed. Maybe be could look for a seed to fix the obtained results.
  		\includegraphics[width=8cm, height=8cm]{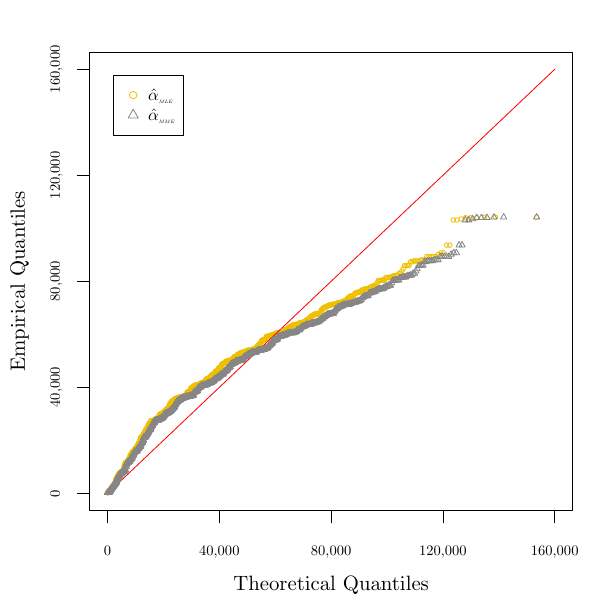}
		\caption{Quantile-Quantile (Q-Q) Plot}
  		\label{AppendixC:GoodnessofFitPlotsbyRegion-Figure10-b}
	\end{subfigure}
	\begin{subfigure}[b]{0.5\linewidth}
	        % Plot generated with the R code: FittingDistributionIncomeRegionAdministrative.R
                % We could generate other plots if needed. Maybe be could look for a seed to fix the obtained results.
  		\includegraphics[width=8cm, height=8cm]{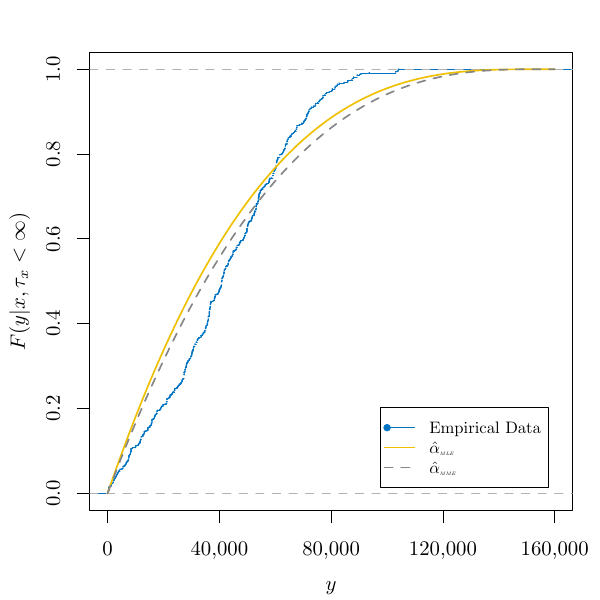}
		\caption{Comparison of Probability Distributions}
  		\label{AppendixC:GoodnessofFitPlotsbyRegion-Figure10-c}
	\end{subfigure}
	\begin{subfigure}[b]{0.5\linewidth}
	        % Plot generated with the R code: FittingDistributionIncomeRegionAdministrative.R
                % We could generate other plots if needed. Maybe be could look for a seed to fix the obtained results.
  		\includegraphics[width=8cm, height=8cm]{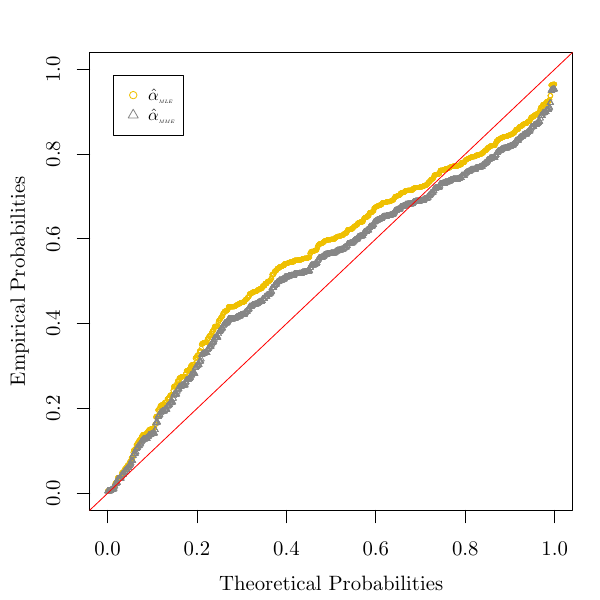}
		\caption{Probability-Probability (P-P) Plot}
  		\label{AppendixC:GoodnessofFitPlotsbyRegion-Figure10-d}
	\end{subfigure}
	\label{AppendixC:GoodnessofFitPlotsbyRegion-Figure10}
\end{figure}

\begin{figure}[H]
\begin{center}Plateau Central\end{center}
	\begin{subfigure}[b]{0.5\linewidth}
	   % Plot generated with the R code: FittingDistributionIncomeRegionAdministrative.R
       	   % We could generate other plots if needed.
  		\includegraphics[width=8cm, height=8cm]{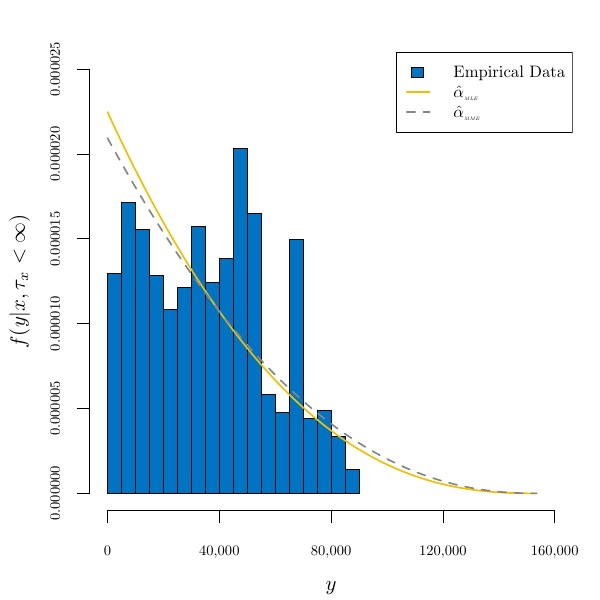}
		\caption{Comparison of Probability Density}
  		\label{AppendixC:GoodnessofFitPlotsbyRegion-Figure11-a}
	\end{subfigure}
	\begin{subfigure}[b]{0.5\linewidth}
	        % Plot generated with the R code: FittingDistributionIncomeRegionAdministrative.R
                % We could generate other plots if needed. Maybe be could look for a seed to fix the obtained results.
  		\includegraphics[width=8cm, height=8cm]{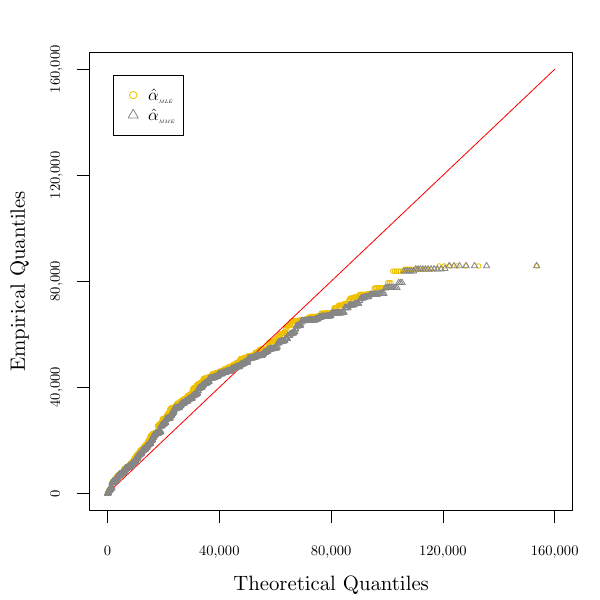}
		\caption{Quantile-Quantile (Q-Q) Plot}
  		\label{AppendixC:GoodnessofFitPlotsbyRegion-Figure11-b}
	\end{subfigure}
	\begin{subfigure}[b]{0.5\linewidth}
	        % Plot generated with the R code: FittingDistributionIncomeRegionAdministrative.R
                % We could generate other plots if needed. Maybe be could look for a seed to fix the obtained results.
  		\includegraphics[width=8cm, height=8cm]{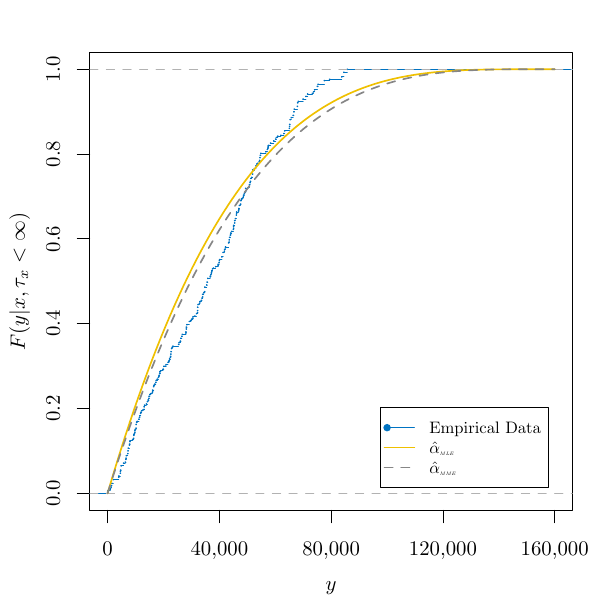}
		\caption{Comparison of Probability Distributions}
  		\label{AppendixC:GoodnessofFitPlotsbyRegion-Figure11-c}
	\end{subfigure}
	\begin{subfigure}[b]{0.5\linewidth}
	        % Plot generated with the R code: FittingDistributionIncomeRegionAdministrative.R
                % We could generate other plots if needed. Maybe be could look for a seed to fix the obtained results.
  		\includegraphics[width=8cm, height=8cm]{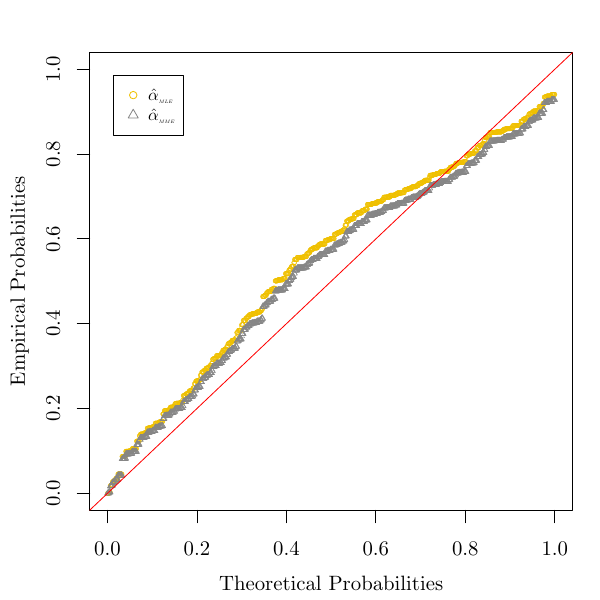}
		\caption{Probability-Probability (P-P) Plot}
  		\label{AppendixC:GoodnessofFitPlotsbyRegion-Figure11-d}
	\end{subfigure}
	\label{AppendixC:GoodnessofFitPlotsbyRegion-Figure11}
\end{figure}

\begin{figure}[H]
\begin{center}Sahel\end{center}
	\begin{subfigure}[b]{0.5\linewidth}
	   % Plot generated with the R code: FittingDistributionIncomeRegionAdministrative.R
       	   % We could generate other plots if needed.
  		\includegraphics[width=8cm, height=8cm]{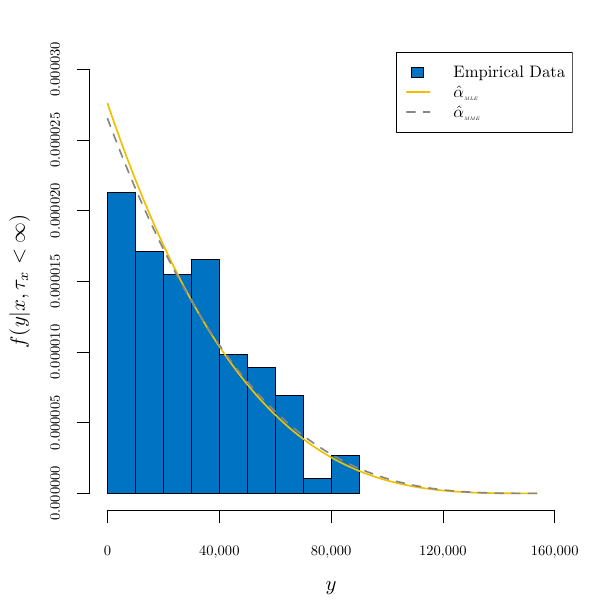}
		\caption{Comparison of Probability Density}
  		\label{AppendixC:GoodnessofFitPlotsbyRegion-Figure12-a}
	\end{subfigure}
	\begin{subfigure}[b]{0.5\linewidth}
	        % Plot generated with the R code: FittingDistributionIncomeRegionAdministrative.R
                % We could generate other plots if needed. Maybe be could look for a seed to fix the obtained results.
  		\includegraphics[width=8cm, height=8cm]{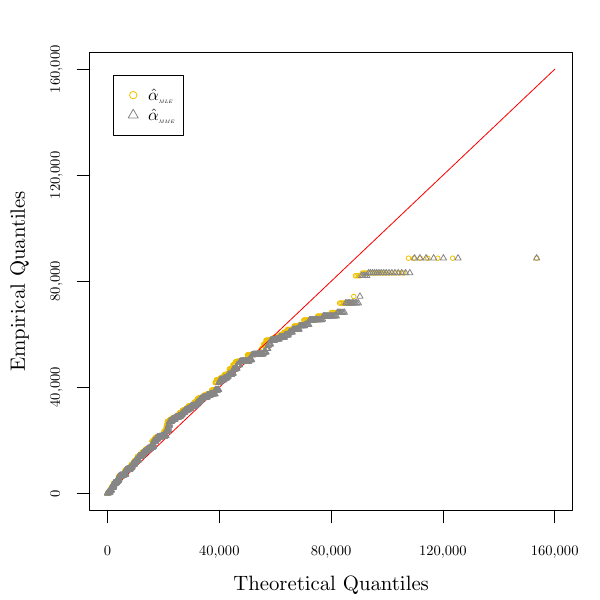}
		\caption{Quantile-Quantile (Q-Q) Plot}
  		\label{AppendixC:GoodnessofFitPlotsbyRegion-Figure12-b}
	\end{subfigure}
	\begin{subfigure}[b]{0.5\linewidth}
	        % Plot generated with the R code: FittingDistributionIncomeRegionAdministrative.R
                % We could generate other plots if needed. Maybe be could look for a seed to fix the obtained results.
  		\includegraphics[width=8cm, height=8cm]{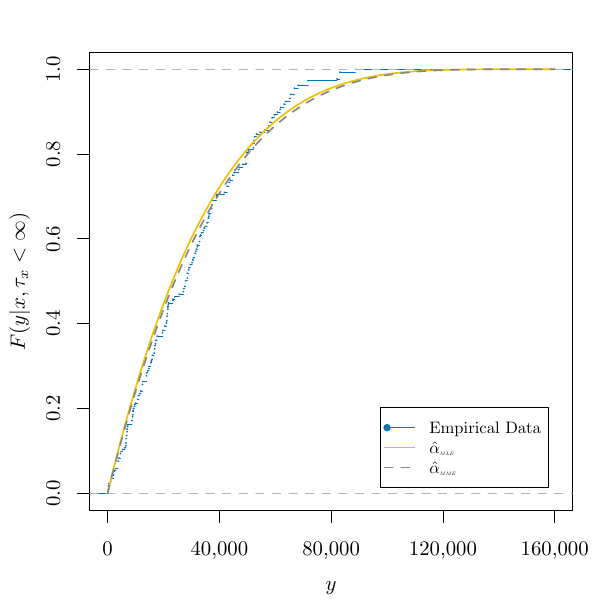}
		\caption{Comparison of Probability Distributions}
  		\label{AppendixC:GoodnessofFitPlotsbyRegion-Figure12-c}
	\end{subfigure}
	\begin{subfigure}[b]{0.5\linewidth}
	        % Plot generated with the R code: FittingDistributionIncomeRegionAdministrative.R
                % We could generate other plots if needed. Maybe be could look for a seed to fix the obtained results.
  		\includegraphics[width=8cm, height=8cm]{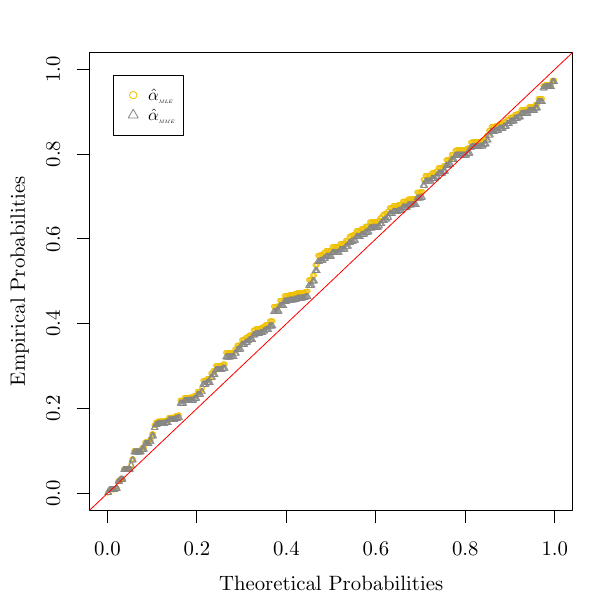}
		\caption{Probability-Probability (P-P) Plot}
  		\label{AppendixC:GoodnessofFitPlotsbyRegion-Figure12-d}
	\end{subfigure}
	\label{AppendixC:GoodnessofFitPlotsbyRegion-Figure12}
\end{figure}

\begin{figure}[H]
\begin{center}Sud-Ouest\end{center}
	\begin{subfigure}[b]{0.5\linewidth}
	   % Plot generated with the R code: FittingDistributionIncomeRegionAdministrative.R
       	   % We could generate other plots if needed.
  		\includegraphics[width=8cm, height=8cm]{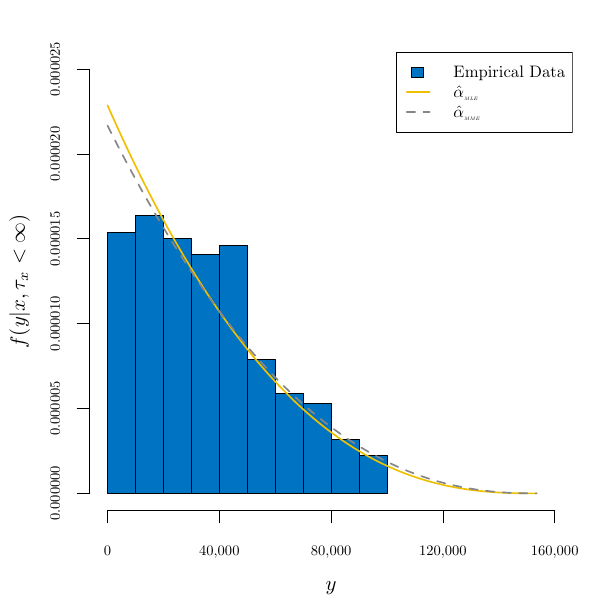}
		\caption{Comparison of Probability Density}
  		\label{AppendixC:GoodnessofFitPlotsbyRegion-Figure13-a}
	\end{subfigure}
	\begin{subfigure}[b]{0.5\linewidth}
	        % Plot generated with the R code: FittingDistributionIncomeRegionAdministrative.R
                % We could generate other plots if needed. Maybe be could look for a seed to fix the obtained results.
  		\includegraphics[width=8cm, height=8cm]{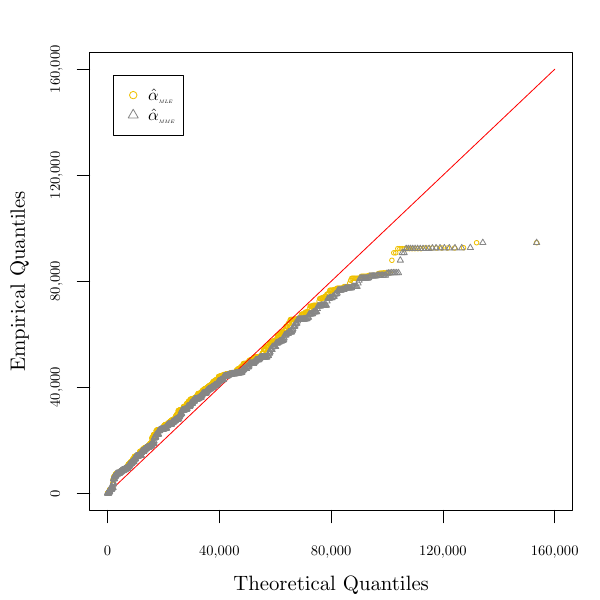}
		\caption{Quantile-Quantile (Q-Q) Plot}
  		\label{AppendixC:GoodnessofFitPlotsbyRegion-Figure13-b}
	\end{subfigure}
	\begin{subfigure}[b]{0.5\linewidth}
	        % Plot generated with the R code: FittingDistributionIncomeRegionAdministrative.R
                % We could generate other plots if needed. Maybe be could look for a seed to fix the obtained results.
  		\includegraphics[width=8cm, height=8cm]{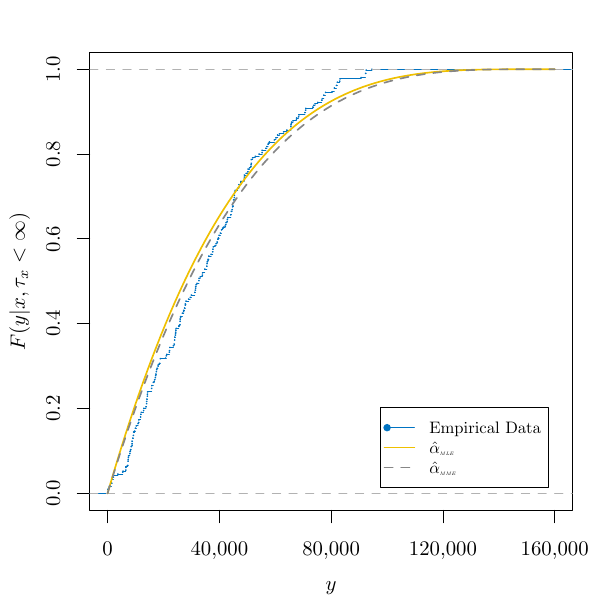}
		\caption{Comparison of Probability Distributions}
  		\label{AppendixC:GoodnessofFitPlotsbyRegion-Figure13-c}
	\end{subfigure}
	\begin{subfigure}[b]{0.5\linewidth}
	        % Plot generated with the R code: FittingDistributionIncomeRegionAdministrative.R
                % We could generate other plots if needed. Maybe be could look for a seed to fix the obtained results.
  		\includegraphics[width=8cm, height=8cm]{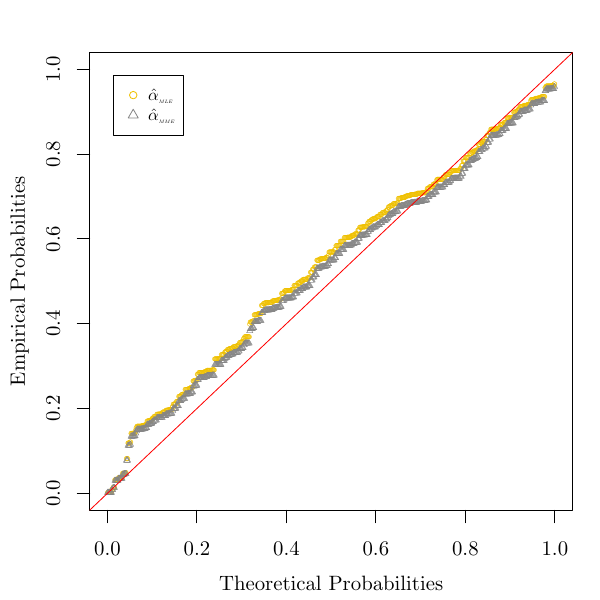}
		\caption{Probability-Probability (P-P) Plot}
  		\label{AppendixC:GoodnessofFitPlotsbyRegion-Figure13-d}
	\end{subfigure}
	\label{AppendixC:GoodnessofFitPlotsbyRegion-Figure13}
\end{figure}

\typeout{get arXiv to do 4 passes: Label(s) may have changed. Rerun}

\end{document}